%% file: paper.tex
\documentclass[acmsmall, nonacm]{acmart}

\setcopyright{rightsretained}
\acmDOI{10.1145/3656437}
\acmYear{2024}
\copyrightyear{2024}
\acmSubmissionID{pldi24main-p436-p}
\acmJournal{PACMPL}
\acmVolume{8}
\acmNumber{PLDI}
\acmArticle{207}
\acmMonth{6}

\usepackage{booktabs}   %
\usepackage{subcaption} %
\input{macros}

\begin{document}

\title{Hyper Hoare Logic: (Dis-)Proving Program Hyperproperties}
\subtitle{(Extended version)}

\author{Thibault Dardinier}
\orcid{0000-0003-2719-4856}             %
\affiliation{
  \department{Department of Computer Science}              %
  \institution{ETH Zurich}            %
  \city{Zurich}
  \country{Switzerland}                    %
}
\email{thibault.dardinier@inf.ethz.ch}          %

\author{Peter Müller}
\orcid{0000-0001-7001-2566}             %
\affiliation{
  \department{Department of Computer Science}             %
  \institution{ETH Zurich}           %
  \city{Zurich}
  \country{Switzerland}                   %
}
\email{peter.mueller@inf.ethz.ch}         %

\begin{abstract}
Hoare logics are proof systems that allow one to formally establish properties of computer programs.
Traditional Hoare logics prove properties of individual program executions
(such as functional correctness).
Hoare logic has been generalized to prove also properties of multiple executions of a program (so-called hyperproperties, such as determinism or non-interference).
These program logics prove the \emph{absence} of (bad combinations of) executions.
On the other hand, program logics similar to Hoare logic have been proposed to \emph{disprove} program properties (e.g., Incorrectness Logic), by proving the \emph{existence} of (bad combinations of) executions.
All of these logics have in common that they specify program properties using assertions over a fixed number of states, for instance, a single pre- and post-state for functional properties or pairs of pre- and post-states for non-interference.

In this paper, we present \logic, a generalization of Hoare logic that lifts assertions to properties of arbitrary \emph{sets} of states.
The resulting logic is simple yet expressive: its judgments can express arbitrary
\thibault{\emph{program hyperproperties},
a particular class of hyperproperties over the set of terminating executions of a program
(including properties of individual program executions)}.
By allowing assertions to reason about sets of states, \logic{} can reason about both the \emph{absence} and the \emph{existence} of (combinations of) executions, and, thereby, supports both proving and disproving program (hyper-)properties within the same logic, including (hyper-)properties that no existing Hoare logic can express. We prove that \logic{} is sound and complete, 
and demonstrate that it captures important proof principles naturally.
All our technical results have been proved in Isabelle/HOL.
\end{abstract}



\maketitle

\section{Introduction}
\label{sec:introduction}
\input{introduction}

\section{Hyper-Triples, Informally}
\label{sec:hyper-triples}
\input{hyper_triples}

\section{Hyper Hoare Logic}
\label{sec:logic}
\input{logic}

\section{Syntactic Rules}
\label{sec:syntactic}
\input{syntactic}

\section{Proof Principles for Loops}
\label{sec:loops}
\input{loops}

\section{Related Work}
\label{sec:related_work}
\input{related_work}

\section{Conclusion and Future Work}
\label{sec:conclusion}
\input{conclusion}

\begin{acks}
We thank the anonymous POPL'24 and PLDI'24 reviewers for their valuable feedback.
This work was partially funded by the Swiss National Science Foundation (SNSF) under Grant No. 197065.
\end{acks}

\section*{Data Availability Statement}
All technical results presented in this paper have been formalized and proven in Isabelle/HOL, and our formalization is publicly available~\cite{artifact,HyperHoareLogic-AFP}.

\bibliography{references}

\ifextended{
\clearpage
\appendix

\section{Technical Definitions Omitted from the Paper}
\label{app:definitions}
\input{appendix/definitions}
\clearpage

\section{Example of a Program Hyperproperty Relating an Unbounded Number of Executions}
\label{app:unbounded}
\input{appendix/unbounded}

\clearpage

\section{Expressing Judgments of Hoare Logics as Hyper-Triples}
\label{app:expressivity}
\input{appendix/expressivity}
\clearpage

\section{Compositionality}
\label{app:compositionality}

\input{appendix/compositionality}

\clearpage

\section{Termination-Based Reasoning}
\label{app:termination}
\input{appendix/termination}

\clearpage

\section{Example: Monotonicity of the Fibonacci Sequence}
\label{app:fibonacci}
\input{appendix/fibonacci}

\clearpage

\section{Example: Existence of a Minimum}
\label{app:minimum}
\input{appendix/minimum}
\clearpage

\section{Synchronous Reasoning Over Different Branches}
\label{app:synchronous}
\input{appendix/synchronous}
\clearpage
}{}

\end{document}

%% file: macros.tex
\usepackage{listings}
\usepackage{color}
\usepackage{soul}
\usepackage{makecell}
\usepackage{multirow}
\usepackage{dirtree}
\usepackage{xspace}
\usepackage{float}

\usepackage{caption}
\usepackage{subcaption}

\usepackage{hyperref}
\usepackage{breakurl}

\usepackage{wrapfig}

\definecolor{light-gray}{gray}{0.9}
\definecolor{darkgreen}{rgb}{0.0,0.6,0.0}

\usepackage{tikz}

\newcommand\secref[1]{Sect.~\ref{#1}}
\newcommand\appref[1]{App.~\ref{#1}}
\newcommand\defref[1]{Def.~\ref{#1}}

\newcommand\figref[1]{Fig.~\ref{#1}}
\newcommand\thmref[1]{Thm.~\ref{#1}}
\newcommand\lemref[1]{Lemma~\ref{#1}}
\newcommand\propref[1]{Prop.~\ref{#1}}

\newcommand{\wrt}{{{w.r.t.\@}}}

\newcommand{\eg}{{{e.g.,~}}}
\newcommand{\ie}{{{i.e.,~}}}

\usepackage[normalem]{ulem}

\def\extended{ }
\newcommand{\ifextended}[2]{\ifdefined\extended{#1}\else{#2}\fi}

\def\nocolour{ }

\newcommand{\depsOnFibonacci}[1]{#1}

\newcommand{\proofsketch}[1]{\begin{proof}#1\end{proof}}

\newcommand{\peter}[1]{\ifdefined\nocolour{#1}\else{\color{blue}{#1}}\fi}

\newcommand{\thibault}[1]{\ifdefined\nocolour{{#1}}\else{\color{darkgreen}{#1}}\fi}

\newcommand{\isabelle}{Isabelle/HOL}

\newcommand{\proven}{proved}

\definecolor{darkred}{rgb}{0.55, 0.0, 0.0}

\usepackage{mathtools}
\usepackage{comment}

\usepackage{prooftree}

\makeatletter

\newskip \point

\def \premisespacing{\quad}

\def \RulePremisesNewlineMore[#1]#2.#3#4{\@ifnextchar\bgroup{\RulePremisesNewlineMore[#1]{#2}.{#3\premisespacing#4}}{\@ifnextchar.{\RulePremisesNewline[#1]{{\begin{array}{c}#2\\#3\premisespacing#4\end{array}}}}{\RuleMultiPremise[#1]{{\begin{array}{c}#2\\#3\end{array}}}{#4}}}}

\def \RulePremisesNewline[#1]#2.#3{\@ifnextchar\bgroup{\RulePremisesNewlineMore[#1]{#2}.{#3}}{\@ifnextchar.{\RulePremisesNewline[#1]{{\begin{array}{c}#2\\#3\end{array}}}}{\RuleMultiPremise[#1]{#2}{#3}}}}

\def \RuleMultiPremise[#1]#2#3{\@ifnextchar\bgroup{\RuleMultiPremise[#1]{#2\premisespacing#3}}{\@ifnextchar.{\RulePremisesNewline[#1]{#2\premisespacing#3}}{\prooftree #2\justifies#3 \using{#1}\endprooftree}}}

\def \RuleWithName[#1]#2{\@ifnextchar\bgroup {\RuleMultiPremise[#1]{#2}}{\@ifnextchar.{\RulePremisesNewline[#1]{#2}}{\prooftree \justifies #2 \using{#1} \endprooftree}}}

\def \RuleWithInfo[#1]{\@ifnextchar[{\RuleWithNameAndCondition[#1]}{\RuleWithName[(#1)]}}

\def \RuleWithNameAndCondition[#1][#2]{\RuleWithName[(#1)^{#2}]}

\def \Inf{\proofrulebaseline=2ex \abovedisplayskip12\point\belowdisplayskip12\point \abovedisplayshortskip8\point\belowdisplayshortskip8\point \@ifnextchar[{\RuleWithInfo}{\RuleWithName[ ]}}

\makeatother

\newcommand{\outline }[1]{\ensuremath{\blue{ \left\{ #1 \right\} }}}
\newcommand{\outlineStart }[1]{\ensuremath{\blue{ \{ #1}}}
\newcommand{\outlineEnd }[1]{\ensuremath{\blue{ #1 \} }}}
\newcommand{\simpleHoare}[3]{\ensuremath{ \blue{ \{ #1 \} } \; #2 \; \blue{ \{ #3 \} } }}

\newcommand{\hoare}[3]{\ensuremath{{\models} \simpleHoare{#1}{#2}{#3} }}
\newcommand{\shoare}[3]{\ensuremath{{\vdash} \simpleHoare{#1}{#2}{#3} }}

\newcommand{\simpleHoareSmall}[3]{\simpleHoare{#1}{#2}{#3}}

\newcommand{\hoareSmall}[3]{\ensuremath{{\models} \simpleHoareSmall{#1}{#2}{#3} }}

\newcommand{\normalHoare}[3]{\ensuremath{\brown{\{ #1 \}} \; #2 \; \brown{\{ #3 \}} }}
\newcommand{\nhoare}[4]{\ensuremath{{\models_{\mathit{#1}}} \normalHoare{#2}{#3}{#4}}}
\newcommand{\nhoareSmall}[4]{\ensuremath{{\models_{\mathit{#1}}} \normalHoare{#2}{#3}{#4}}}

\usepackage{amsthm}
\newtheorem{definition}{Definition}
\newtheorem{notation}{Notation}
\newtheorem{example}{Example}
\newtheorem{lemma}{Lemma}
\newtheorem{theorem}{Theorem}
\newtheorem{proposition}{Proposition}

\usepackage{xcolor}

\newcommand{\cskip}{ \mathbf{skip} }
\newcommand{\cifonly}[2]{ \mathbf{if} \; (#1) \; \{#2\} }

\newcommand{\cif}[3]{ \mathbf{if} \; (#1) \; \{#2\} \; \mathbf{else} \; \{#3\} }
\newcommand{\cnif}[2]{#1 + #2}
\newcommand{\cseq}{ \mathbf{;} \; }
\newcommand{\chavoc}[1]{ \mathit{#1} \coloneqq{} \mathit{nonDet()}}
\newcommand{\cwhile}[2]{ \mathbf{while} \; (#1) \; \{#2\} }
\newcommand{\cnwhile}[1]{#1^*}
\newcommand{\cassume}[1]{\mathbf{assume} \; \mathit{#1}}
\newcommand{\cassign}[2]{\mathit{#1} \coloneqq \mathit{#2}}

\newcommand{\vars}{\mathit{PVars}}
\newcommand{\vals}{\mathit{PVals}}
\newcommand{\lvars}{\mathit{LVars}}
\newcommand{\lvals}{\mathit{LVals}}
\newcommand{\states}{\mathit{PStates}}
\newcommand{\extstates}{\mathit{ExtStates}}

\usepackage{esvect}
\newcommand{\vecto}[1]{
    \vv{\vphantom{d}#1}
}

\newcommand{\bigstep}[3]{ \langle #1, #2 \rangle \rightarrow #3 }
\newcommand{\bigstepK}[4]{ \langle #1, \vecto{#2} \rangle \overset{#4}{\rightarrow} \vecto{#3} }

\newcommand{\logic}{Hyper Hoare Logic}
\newcommand{\sem}{\mathit{sem}}
\newcommand{\powerset}[1]{\mathbb{P}(#1)}

\newcommand{\hyperassertion}{hyper-assertion}
\newcommand{\hyperassertions}{hyper-assertions}

\newcommand{\lproj}[1]{#1^L}
\newcommand{\pproj}[1]{#1^P}

\newcommand{\sat}[2]{#2(#1)}

\newcommand{\namerule}[1]{\textit{#1}}

\newcommand{\phyperpropertyfont}[1]{\mathcal{#1}}

\renewcommand*{\proofname}{Proof sketch}

\newcommand{\hypersat}[2]{#1 \in \phyperpropertyfont{#2}}

\newcommand{\freevars}[1]{\mathit{fv}(#1)}

\usepackage{pifont}
\newcommand{\checked}{\textcolor{green!75!black}{\ding{51}\xspace} \hfill}

\newcommand{\range}[2]{[#1, #2]}

\definecolor{royalblue}{rgb}{0.25, 0.41, 0.88}
\newcommand{\blue}[1]{{\color{royalblue}{#1}}}

\definecolor{somebrown}{rgb}{0.8, 0.35, 0.1}
\newcommand{\brown}[1]{{\color{somebrown}{#1}}}

\newcommand{\none}{\ensuremath{\varnothing}}
\newcommand{\mypar}[1]{\paragraph{#1}}

\newcommand{\inSet}[1]{\ensuremath{ \langle #1 \rangle }}

\newcommand{\always}{\ensuremath{ \square }}

\newcommand{\otimesn}{\ensuremath{ \textstyle \bigotimes\nolimits_{n \in \mathbb{N}}}}

\newcommand{\transformAssume}[2]{\ensuremath{ \mathit{\Pi}_{#1}[#2] }}
\newcommand{\transformAssumeEmpty}[1]{\ensuremath{ \mathit{\Pi}_{#1} }}
\newcommand{\transformAssignEmpty}[2]{\ensuremath{ \mathcal{A}^{#1}_{#2} }}
\newcommand{\transformAssign}[3]{\ensuremath{ \mathcal{A}^{#1}_{#2}[ #3 ] }}
\newcommand{\transformHavoc}[2]{\ensuremath{ \mathcal{H}_{#1}[#2] }}
\newcommand{\transformHavocEmpty}[1]{\ensuremath{ \mathcal{H}_{#1} }}

\newcommand{\expApplied}[2]{\ensuremath{ #1(#2) }}

\newcommand{\hasMin}[1]{\ensuremath{ \mathit{hasMin}_{#1} }}
\newcommand{\mono}[2]{\ensuremath{ \mathit{mono}_{#2}^{#1} }}
\newcommand{\gni}[2]{\ensuremath{ \mathit{GNI}_{#1}^{#2} }}
\newcommand{\emp}{\ensuremath{ \mathit{emp} }}
\newcommand{\low}[1]{\ensuremath{ \mathit{low}(#1) }}

\newcommand{\logEntails}[1]{\ensuremath{ \overset{#1}{\Rightarrow} }}
\newcommand{\singleton}{\ensuremath{ \mathit{isSingleton} }}

\newcommand{\invariant}[2]{\ensuremath{ \mathit{inv}^{#1}(#2) }}

\newcommand{\hoareTot}[3]{\ensuremath{\models_{\Downarrow} \simpleHoare{#1}{#2}{#3} }}
\newcommand{\shoareTot}[3]{\ensuremath{\vdash_{\Downarrow} \simpleHoare{#1}{#2}{#3} }}

\newcommand{\getvar}[2]{\ensuremath{#1(#2)}}
\newcommand{\concat}{\ensuremath{\; {++} \;}}

\newcommand{\simp}{\ensuremath{{\Rightarrow}}}
\newcommand{\sand}{\ensuremath{{\land}}}

\newcommand{\sminus}{\ensuremath{{-}}}

\newcommand{\seq}{\ensuremath{{=}}}
\newcommand{\sgeq}{\ensuremath{{\geq}}}
\newcommand{\sline}{\\[-4pt]}

\newcommand{\satSyntactic}[4]{ \ensuremath{#1, #2, #3 \models #4} }
\newcommand{\evalExp}[3]{ \ensuremath{ \llbracket #1 \rrbracket^{#2}_{#3} } }

%% file: introduction.tex
Hoare Logic~\cite{FloydLogic, HoareLogic}
is a logic designed to formally prove functional correctness of computer programs.
It enables proving judgments (so-called \emph{Hoare triples})
of the form $\normalHoare{P}{C}{Q}$, where $C$ is a program command,
and $P$ (the \emph{precondition}) and $Q$ (the \emph{postcondition}) are assertions over execution states.
The Hoare triple $\normalHoare{P}{C}{Q}$ is valid if and only if executing $C$ in an initial state that satisfies $P$
can only lead to  final states that satisfy $Q$.

Hoare Logic is widely used to prove the absence of runtime errors, functional correctness, resource bounds, etc.\ All of these properties have in common that they are properties of \emph{individual} program executions.
However, classical Hoare Logic cannot reason about  properties of \emph{multiple} program executions (so-called \emph{hyperproperties}~\cite{hyperproperties}), such as determinism (executing the program twice in the same initial state results in the same final state) or information flow security, which is often phrased as non-interference~\cite{volpano1996nonInterference}
(executing the program twice with the same low-sensitivity inputs results in the same low-sensitivity outputs). To overcome such limitations and to reason about more types of properties, Hoare Logic has been extended and adapted in various ways. We refer to those extensions and adaptations collectively as \emph{Hoare logics}.

Among them are several logics that can establish properties of two~\cite{Benton04,francez1983product,SurveyRHL,RelationalSL,aguirre2017relational,next700RHL,AmtoftSIF,Costanzo2014,SecCSL,CommCSL} or even $k$~\cite{CHL16,HypersafetyCompositionally} executions of the same program, where $k > 2$ is useful for properties such as transitivity and associativity. \emph{Relational Hoare logics} are able to prove \emph{relational properties}, \ie properties relating  executions of two (potentially different) programs, for instance, to prove program equivalence.

All of these logics have in common that they can prove only properties that hold \emph{for all} (combinations of) executions, that is, they prove the \emph{absence} of bad (combinations of) executions; to achieve that, their judgments \emph{overapproximate} the possible executions of a program. Overapproximate logics cannot prove the \emph{existence} of (combinations of) executions, and thus cannot establish certain interesting  program properties, such as the presence of a bug or non-determinism.

To overcome this limitation, recent work~\cite{ReverseHL,IncorrectnessLogic,ISL,CISL,InsecurityLogic} proposed Hoare logics that can prove the \emph{existence} of (individual) executions, for instance, to \emph{disprove} functional correctness.
We call such Hoare logics  \emph{underapproximate}.
Tools based on underapproximate Hoare logics have proven useful for finding bugs on an industrial scale~\cite{Blackshear2018,Gorogiannis2019,Distefano2019,Le2022}.
More recent work has proposed Hoare logics that combine underapproximate and overapproximate reasoning
\thibault{for single executions,
such as Outcome Logic~\cite{OutcomeLogic} and Exact Separation Logic~\cite{maksimovic2022exact},
and for $\forall^* \exists^*$-hyperproperties,
such as RHLE~\cite{RHLE}.}
\peter{Another recent work, BiKAT~\cite{BiKat},
can be used to prove
$\forall \exists$-properties between two programs $C_1$ and $C_2$,
by providing an \emph{alignment witness} $W$ that overapproximates the behavior of $C_1$ and underapproximates the behavior of $C_2$,
and proving that $W$ satisfies a relevant $\forall \forall$-property,
which the authors demonstrate by providing inference rules
for foward and backward simulations ($\forall \exists$-properties).
BiKAT can also in principle be used to prove $\exists \forall$-properties,
by essentially proving the negation of a $\forall \exists$-property,
\ie by proving that no such alignment witness $W$ exists,
but the authors provide no inference rules for $\exists \forall$-properties.}

\mypar{The problem}

\begin{figure*}
\begin{center}
    \resizebox{\columnwidth}{!}{%
    \begin{tabular}{|c|c|c|c|p{0.08\textwidth}|} 
     \hline
      &
     \multicolumn{4}{c|}{\textbf{Number of executions}} \\
     \hline
     \textbf{Type} & $\mathbf{1}$ & $\mathbf{2}$ & $\boldsymbol{k}$ &  \multicolumn{1}{c|}{$\boldsymbol{\infty}$} \\ %
     \hline
     Overapproximate (hypersafety) & \checked{} HL, OL, RHL, CHL, RHLE, MHRM\peter{, BiKAT} & \checked{} RHL, CHL, RHLE, MHRM\peter{, BiKAT} & \checked{} CHL, RHLE & \checked{} \none{} \\ \hline
     Backward underapproximate & \checked{} IL, InSec\peter{, BiKAT} & \checked{} InSec\peter{, BiKAT} & \checked{} \none{} & \checked{} \none{} \\ \hline
     Forward underapproximate & \checked{} OL, RHLE, MHRM\peter{, BiKAT} & \checked{} RHLE, MHRM\peter{, BiKAT} & \checked{} RHLE & \checked{} \none{} \\ \hline
     $\forall^* \exists^*$ & not applicable & \checked{} RHLE, MHRM\peter{, BiKAT} & \checked{} RHLE & \checked{} \none{} \\ \hline
     $\exists^* \forall^*$ & not applicable & \checked{} \peter{BiKAT} & \checked{} \none{} & \checked{} \none{} \\ \hline
     Set properties & not applicable & not applicable & not applicable & \checked{} \none{} \\ \hline
    \end{tabular}}
\end{center}

\caption{(Non-exhaustive) overview of Hoare logics, classified in two dimensions:
The type of properties a logic can establish, and the number of program executions these properties can relate (column ``$\mathbf{\infty}$'' subsumes an unbounded and an infinite number of executions).
\ifextended{%
We explain the distinction between backward and forward underapproximate properties in \appref{subsec:liveness}.
\appref{app:unbounded} gives examples of (hypersafety and set) properties for an unbounded number of executions.%
}{%
In our extended version~\cite{hhlExtended},
we explain the distinction between backward and forward underapproximate properties,
and we give examples of (hypersafety and set) properties for an unbounded number of executions.%
}
$\forall^* \exists^*$- and $\exists^* \forall^*$-hyperproperties are discussed in \secref{sec:hyper-triples}.
A green checkmark indicates that a property is handled by our \logic{} for the programming language defined in \secref{subsec:language},
and $\varnothing$ indicates that no other Hoare logic supports it.
The acronyms refer to the following.
CHL: Cartesian Hoare Logic~\cite{CHL16},
HL: Hoare Logic~\cite{HoareLogic,FloydLogic},
IL: Incorrectness Logic~\cite{IncorrectnessLogic} or Reverse Hoare Logic~\cite{ReverseHL},
InSec: Insecurity Logic~\cite{InsecurityLogic},
OL: Outcome Logic~\cite{OutcomeLogic},
RHL: Relational Hoare Logic~\cite{Benton04},
RHLE~\cite{RHLE},
MHRM~\cite{next700RHL},
\peter{BiKAT~\cite{BiKat}}.
}
\label{fig:overview}
\end{figure*}

\figref{fig:overview} presents a (non-exhaustive) overview of the landscape of Hoare logics, where logics are classified in two dimensions:
the type of properties they can establish, 
and the number of program executions those properties can relate.
The table reveals two open problems. First, some types of hyperproperties cannot be expressed by any existing Hoare logic (represented by $\varnothing$).
For example, to prove that a program \peter{violates generalized non-interference,
one needs to show that there \emph{exist} two executions $\tau_1$ and $\tau_2$ such that
\emph{all} executions with the same high-sensitivity input as $\tau_1$ have a different low-sensitivity output than $\tau_2$.\footnote{\peter{Assuming
no public (low-sensitivity) input}.}
}
Such $\exists^*\forall^*$-hyperproperties cannot be proved by any existing Hoare logic.
Second, the existing logics cover different, often disjoint program properties, which may hinder practical applications:
reasoning about a wide spectrum of properties of a given program requires the application of several logics, each with its own judgments; properties expressed in different, incompatible logics cannot be composed within the same proof system.

\mypar{This work}
We present \emph{\logic{}}, a novel Hoare logic that enables
\emph{proving or disproving} any \thibault{\emph{program hyperproperty},
a particular class of hyperproperties over the set of terminating executions of a program
(formally defined in \secref{subsec:expressivity}),
which includes properties of individual program executions.
In the rest of this paper, when the context is clear, we use \emph{hyperproperties} to refer to program hyperproperties.
}
As indicated by the green checkmarks in \figref{fig:overview}, these include many different types of properties, relating \emph{any} (potentially unbounded or even infinite) number of program executions,
and many hyperproperties that no existing Hoare logic can handle.
Among them are $\exists^* \forall^*$-hyperproperties
such as violations of generalized non-interference (\secref{subsec:assume-rule}),
and hyperproperties relating an unbounded or infinite number of executions
such as quantifying information flow
\thibault{based on Shannon entropy}
or
min-capacity~\cite{Shannon48,assaf2017hypercollecting,Yasuoka2010,Smith2009}
(\thibault{we give an example in} \ifextended{\appref{app:unbounded}}{our extended version~\cite{hhlExtended}}).
\thibault{Moreover, \logic{} offers rules to compose hyper-triples with different types of hyperproperties, such as $\exists \forall$ with $\forall \forall$, or $\forall \forall \exists$ with $\forall \forall$}.

\logic{} is based on a simple yet powerful idea: We lift 
pre- and postconditions from assertions over a \emph{fixed} number of execution states
to \emph{hyper-assertions} over \emph{sets} of execution states.
\logic{} then establishes \emph{hyper-triples} of the form $\simpleHoare{P}{C}{Q}$,
where $P$ and $Q$ are hyper-assertions.
Such a hyper-triple is valid iff
for any set of initial states $S$ that satisfies $P$,
the set of all final states that can be reached by executing $C$ in some state from $S$ satisfies $Q$.
By allowing assertions to quantify \emph{universally} over states, \logic{} can express overapproximate properties, whereas \emph{existential} quantification expresses underapproximate properties. Combinations of universal and existential quantification in the same assertion, as well as assertions over infinite sets of states, allow \logic{}
to prove or disprove properties beyond existing logics.

\mypar{Contributions}
Our main contributions are:

\begin{itemize}
\item We present \logic, a novel Hoare logic that can prove or disprove arbitrary hyperproperties over terminating executions. 

\item We formalize our logic and prove soundness and completeness in \isabelle{}~\cite{Isabelle}.

\item We derive easy-to-use syntactic rules for a restricted class of \emph{syntactic} hyper-assertions, as well as additional loop rules that capture different reasoning principles.

\item We prove compositionality rules for hyper-triples, which enable the flexible composition of hyper-triples of different forms and, thus, facilitate modular proofs.

\item We demonstrate the expressiveness of \logic, both 
on judgments of existing Hoare logics and on hyperproperties that no existing Hoare logic supports.
\thibault{We also prove in \isabelle{} that hyper-triples capture precisely program hyperproperties,
and that any invalid hyper-triple can be disproved by proving another hyper-triple.
}
\end{itemize}

\mypar{Outline}
\secref{sec:hyper-triples} informally presents hyper-triples, and shows how they can be used to specify hyperproperties.
\secref{sec:logic} introduces the rules of \logic{}, proves that these rules are sound and complete for establishing valid hyper-triples,
\thibault{defines program hyperproperties, proves that hyper-triples capture precisely those,
and proves that invalid hyper-triples can be disproved by proving other hyper-triples.}
Secs.~\ref{sec:syntactic} and~\ref{sec:loops} derive additional rules that enable concise proofs in common cases.
We discuss related work in \secref{sec:related_work} and conclude in \secref{sec:conclusion}.
\ifextended{%
The appendix contains further details and extensions.
In particular, \appref{app:expressivity} shows how to express judgments of existing logics in \logic, and \appref{app:compositionality} presents compositionality rules.}{%
Our extended version~\cite{hhlExtended} contains further details and extensions.
In particular, we show how to express judgments of existing logics in \logic, and present compositionality rules.}
All our technical results (Secs.\ \ref{sec:logic}, \ref{sec:syntactic}, \ref{sec:loops}, and the \ifextended{appendix}{extended version}) have been proved in \isabelle{} \cite{Isabelle},
and our mechanization is publicly available~\cite{artifact,HyperHoareLogic-AFP}.

%% file: hyper_triples.tex
In this section, we illustrate how hyper-triples can be used to express different types of hyperproperties,
including over- and underapproximate hyperproperties for single (\secref{subsec:randomNumber}) and multiple (\secref{subsec:k-safety} and \secref{subsec:beyond-k-safety}) executions.

\subsection{Overapproximation and Underapproximation}
\label{subsec:randomNumber}

Consider the command $C_0 \triangleq (\cassign{x}{\mathit{randIntBounded}(0, 9)})$, which generates a random integer between $0$ and $9$ (both included), and assigns it to the variable $x$. Its functional correctness properties include: (P1)~The final value of $x$ is in the interval $[0,9]$,
and (P2)~every value in $[0,9]$ can occur for every initial state (\ie the output is not determined by the initial state).

Property P1 expresses the \emph{absence} of bad executions,
in which the output $x$ is outside the interval $[0,9]$.
This property can be expressed in classical Hoare logic, with the triple $\normalHoare{\top}{C_0}{0 \leq x \leq 9}$.
In \logic{}, where assertions are properties of sets of states, we  express it using a postcondition that \emph{universally} quantifies over all possible final states: In all final states, the value of $x$ should be in $[0,9]$.
The hyper-triple $\simpleHoare{\top}{C_0}{ \forall \inSet{\varphi'} \ldotp 0 \leq \varphi'(x) \leq 9}$
expresses this property.
The postcondition, written in the syntax that will be introduced in \secref{sec:syntactic}, is semantically equivalent
to $\outline{\lambda S' \ldotp \forall \varphi' \in S' \ldotp 0 \leq \varphi'(x) \leq 9}$.
This hyper-triple means that, for any set $S$ of initial states $\varphi$ (satisfying the trivial precondition $\top$),
the set $S'$ of all final states $\varphi'$ that can be reached by executing $C_0$ in some initial state $\varphi \in S$ satisfies the postcondition,
\ie all final states $\varphi' \in S'$ have a value for $x$ between $0$ and $9$.
This hyper-triple illustrates a systematic way of expressing classical Hoare triples as hyper-triples~(%
\ifextended{see \appref{subsec:safety}}{as shown in our extended version~\cite{hhlExtended}}).
Property P2 expresses the \emph{existence} of desirable executions and can be expressed using an underapproximate Hoare logic.
In \logic{}, we use a postcondition that \emph{existentially} quantifies over all possible final states: For each  $n\in [0,9]$, there exists a final state where $x=n$.
The hyper-triple $\simpleHoare{ \exists \inSet{\varphi} \ldotp \top }{C_0}{ \forall n \ldotp
0 \leq n \leq 9 \Rightarrow \exists \inSet{\varphi'} \ldotp \varphi'(x) = n}$ expresses P2.
The precondition is semantically equivalent to $\outline{\lambda S \ldotp \exists \varphi \in S}$.
It requires the set $S$ of initial states to be non-empty (otherwise the set of states reachable from states in $S$ by executing $C_0$ would also be empty, and the postcondition would not hold).
The postcondition ensures that, for any $n\in [0,9]$, it is possible to reach at least one state $\varphi'$ with $\varphi'(x) = n$.

This example shows that hyper-triples can express both under- and overapproximate properties, which allows \logic{} to reason about both the \emph{absence} of bad executions and the \emph{existence} of good executions.
Moreover,
hyper-triples can also be used to prove the existence of \emph{incorrect} executions,
which has proven useful in practice to find bugs without false positives~\cite{IncorrectnessLogic, Le2022}.
To the best of our knowledge, the only other Hoare logics that can express both properties P1 and P2 are Outcome Logic~\cite{OutcomeLogic}, 
Exact Separation Logic~\cite{maksimovic2022exact},
\peter{and BiKAT~\cite{BiKat}}.\footnote{While RHLE~\cite{RHLE} can in principle reason about the existence of executions, it is unclear how to express the existence \emph{for all} numbers $n$.}
However, \peter{the two first} are limited to properties of single executions,
\peter{and the latter to properties relating two executions.}
Thus, \peter{these logics} cannot handle \peter{the general class of} $k$-safety hyperproperties,
\peter{which we discuss next}.

\subsection{(Dis-)Proving $k$-Safety Hyperproperties}
\label{subsec:k-safety}

A $k$-safety hyperproperty~\cite{hyperproperties} is a property that characterizes \emph{all combinations of}
$k$ executions of the same program.
An important example is information flow security, which requires that programs that manipulate secret data (such as passwords) do not expose secret information to their users. In other words, the content of high-sensitivity (secret) variables must not leak into low-sensitivity (public) variables.
For deterministic programs, information flow security is often formalized as \emph{non-interference} (NI)~\cite{volpano1996nonInterference}, a 2-safety hyperproperty:
Any two executions of the program with the same low-sensitivity (\emph{low} for short) inputs (but potentially different high-sensitivity inputs) must have the same low outputs.
That is, for all pairs of executions $\tau_1, \tau_2$,
if $\tau_1$ and $\tau_2$ agree on the initial values of all low variables, they must also agree on the final values of all low variables.
This ensures that the final values of low variables are  not influenced by the values of high variables.
Assuming for simplicity that we have only one low variable $l$,
the hyper-triple $\simpleHoare{\low{l}}{C_1}{\low{l}}$,
where $\low{l} \triangleq (\forall \inSet{\varphi_1}, \inSet{\varphi_2} \ldotp \varphi_1(l) = \varphi_2(l))$,
expresses that $C_1$ satisfies NI:
If all states in $S$ have the same value for $l$,
then all final states reachable by executing $C_1$ in any initial state $\varphi \in S$ will have the same value for $l$. 
Note that this set-based definition is equivalent to the standard definition based on pairs of executions;
Instantiating $S$ with a set of two states directly yields the standard definition.

Non-interference requires that all final states have the same value for $l$, irrespective of the initial state that leads to any given final state. Other $k$-safety hyperproperties need to relate initial and final states. For example, the program $\cassign{y}{f(x)}$ is \emph{monotonic} iff for any two executions with $\varphi_1(x) \geq \varphi_2(x)$, we have $\varphi'_1(y) \geq \varphi'_2(y)$, where $\varphi_1$ and $\varphi_2$ are the initial states $\varphi'_1$ and $\varphi'_2$ are the \emph{corresponding} final states.
To relate initial and final states, \logic{} uses \emph{logical variables} (also called \emph{auxiliary variables}~\cite{auxVariables}). These variables cannot appear in a program, and thus are guaranteed to have the same values in the initial and final states of an execution. We use this property to tag corresponding states, as illustrated by the hyper-triple for monotonicity:
$\simpleHoare{\mono{t}{x}}{\cassign{y}{f(x)}}{\mono{t}{y}}$,
where $\mono{t}{x} \triangleq (\forall \inSet{\varphi_1}, \inSet{\varphi_2} \ldotp \varphi_1(t) = 1 \land \varphi_2(t) = 2 \Rightarrow \varphi_1(x) \geq \varphi_2(x))$. Here, $t$ is a logical variable used to distinguish the two executions of the program.

\paragraph{Disproving $k$-safety hyperproperties}
As explained in the introduction, being able to prove that a property does \emph{not} hold is valuable in practice, because it allows building tools that can find bugs without false positives.
\logic{} is able to \emph{disprove} hyperproperties by proving a hyperproperty that is essentially its negation.
For example, we can prove that the insecure program
$C_2 \triangleq (\cif{h > 0}{ \cassign{l}{1} }{ \cassign{l}{0} })$, where $h$ is a high variable,
\emph{violates} non-interference (NI),
using the following hyper-triple:
$\simpleHoare{\low{l}
\land (\exists \inSet{\varphi_1}, \inSet{\varphi_2} \ldotp \varphi_1(h) > 0  \land \varphi_2(h) \leq 0)
}{C_2}{
    \exists \inSet{\varphi'_1}, \inSet{\varphi'_2} \ldotp
    \varphi'_1(l) \neq \varphi'_2(l)
}$.
The postcondition is the negation of the postcondition for $C_1$ above, hence expressing that $C_2$ \emph{violates} NI\@.
Note that the precondition needs to be stronger than for $C_1$. Since the postcondition has to hold for \emph{all} sets that satisfy the precondition, we have to require that the set of initial states includes two states that will definitely lead to different final values of $l$.

\peter{The general class of $k$-safety hyperproperties includes properties that relate more than $2$ executions,
such as transitivity ($k = 3$) and associativity ($k = 4$)~\cite{CHL16}.}
The only other Hoare logic that can be used to both prove and disprove \peter{arbitrary} $k$-safety hyperproperties is RHLE\thibault{~\cite{RHLE}},
since it supports $\forall^* \exists^*$-hyperproperties,
which includes both hypersafety (that is, $\forall^*$) properties and their negation (that is, $\exists^*$-hyperproperties).
However, RHLE does not support $\exists^* \forall^*$-hyperproperties, and thus cannot disprove $\forall^* \exists^*$-hyperproperties such as generalized non-interference, as we discuss next.

\subsection{Beyond $k$-Safety}
\label{subsec:beyond-k-safety}

NI is widely used to express information flow security for deterministic programs,
but is overly restrictive for non-deterministic programs.
For example, the command 
$C_3 \triangleq (\chavoc{y} \cseq \cassign{l}{h + y})$
is information flow secure. Since the secret $h$ is added to
an unbounded non-deterministically chosen integer $y$, any secret $h$ can result in any\footnote{This property holds  for both unbounded and bounded arithmetic.}
value for the public variable $l$ and, thus, we cannot learn anything certain about $h$ from observing the value of $l$.
However, because of non-determinism, $C_3$ does not satisfy NI:
Two executions with the same initial values for $l$ could get different values for $y$,
and thus have different final values for $l$.

Information flow security for non-deterministic programs
(such as $C_3$)
is often formalized as \emph{generalized non-interference} (GNI)~\cite{mccullough1987GNI,mclean1996formalGNI}, a security notion weaker than NI\@. GNI allows two executions  $\tau_1$ and $\tau_2$ with the same low inputs to have \emph{different} low outputs, provided that there is a third execution $\tau$ with the same low inputs that has the same high inputs as $\tau_1$ and the same low outputs as $\tau_2$. That is, the difference in the low outputs between $\tau_1$ and $\tau_2$ cannot be attributed to their secret inputs.\footnote{GNI is often formulated without the requirement that $\tau_1$ and $\tau_2$ have the same low inputs, \eg in~\citet{hyperproperties}.
This alternative formulation can also be expressed in \logic{}, with the hyper-triple
$\simpleHoare{ \forall \inSet{\varphi} \ldotp \varphi(l_{\mathit{in}}) = \varphi(l) }{C_3}{
\forall \inSet{\varphi'_1}, \inSet{\varphi'_2} \ldotp \exists \inSet{\varphi'} \ldotp
\varphi'(h) = \varphi'_1(h) \land \varphi'(l_{\mathit{in}}) = \varphi'_2(l_{\mathit{in}}) \land \varphi'(l) = \varphi'_2(l)
}$.
The precondition binds, in each state, the initial value of $l$ to the logical variable $l_{\mathit{in}}$,
which enables the postcondition to refer to the initial value of $l$.
}
The non-deterministic program $C_3$ satisfies GNI, which can be expressed via the hyper-triple\footnote{We assume here for simplicity that $h$ is not modified by $C_3$.}
$\simpleHoare{\low{l}
}{C_3}{
    \forall \inSet{\varphi'_1}, \inSet{\varphi'_2} \ldotp
    \exists \inSet{\varphi'} \ldotp
    \varphi'(h) = \varphi'_1(h) \land
    \varphi'(l) = \varphi'_2(l)}$.
The final states $\varphi'_1$ and $\varphi'_2$ correspond to the executions $\tau_1$ and $\tau_2$, respectively,
and $\varphi'$ corresponds to execution $\tau$.

As before,
the expressivity of hyper-triples enables us not only to express that a program \emph{satisfies} complex hyperproperties such as GNI, but also that a program \emph{violates} them.
For example, the program $C_4 \triangleq (\chavoc{y} \cseq \cassume{y \leq 9} \cseq \cassign{l}{h + y})$,
where the first two statements model a non-deterministic choice of $y$ smaller or equal to $9$,
leaks information:
Observing for example $l = 20$ at the end of an execution, one can deduce that $h \geq 11$ (because $y \leq 9$).
We can formally express that $C_4$ violates GNI using the following hyper-triple:\footnote{Still assuming that $h$ is not modified.}
\\
$\simpleHoare{   
    \low{l}
    \land
    (\exists \inSet{\varphi_1}, \inSet{\varphi_2} \ldotp 
\varphi_1(h) \neq \varphi_2(h))
}{C_4}{
    \exists \inSet{\varphi'_1}, \inSet{\varphi'_2} \ldotp
    \forall \inSet{\varphi'} \ldotp
    \varphi'(h) = \varphi'_1(h)
    \Rightarrow
    \varphi'(l) \neq \varphi'_2(l)
}$
\\
The postcondition implies the negation of the postcondition we used previously to express GNI\@.
As before,
we had to strengthen the precondition to prove this violation.

GNI is a $\forall \forall \exists$-hyperproperty, whereas its negation is an $\exists \exists \forall$-hyperproperty.
To the best of our knowledge, \logic{} is the only Hoare logic that can \thibault{both} prove and disprove GNI\@. In fact, we will see in \secref{subsec:expressivity} that all hyperproperties over terminating program executions can be proven or disproven with \logic{}.

%% file: logic.tex
In this section, we present the programming language used in this paper~(\secref{subsec:language}),
formalize hyper-triples~(\secref{subsec:hyper-triples}),
present the core rules of \logic{}~(\secref{subsec:rules}),
prove soundness and completeness of the logic \wrt{} hyper-triples~(\secref{subsec:sound-complete}),
formally characterize the expressivity of hyper-triples~(\secref{subsec:expressivity}),
and discuss additional rules for composing proofs~(\secref{subsec:compositionality}).
All technical results presented in this section have been formalized in Isabelle/HOL.

\subsection{Language and Semantics}\label{subsec:language}

We present \logic{} for the following imperative programming language:

\begin{definition}\textbf{Program states and programming language.}
    A \emph{program state} (ranged over by $\sigma$) is a mapping from local variables (in the set $\vars$) to values (in the set $\vals$):
    The set of program states $\states$ is defined as the set of total functions from $\vars$ to $\vals$:
    $\states \triangleq \vars \rightarrow \vals$.

Program commands $C$ are defined by the following syntax, where $x$ ranges
    over variables in the set $\vars$, $e$ over expressions (modeled as total functions from $\states$ to $\vals$), and $b$ over predicates over states (total functions from $\states$ to Booleans):
$$
C \Coloneqq
\cskip \mid
\cassign{x}{e} \mid
\chavoc{x} \mid
\cassume{b} \mid
C \cseq C \mid
\cnif{C}{C} \mid
\cnwhile{C}
$$
\end{definition}

The $\cskip{}$, assignment, and sequential composition commands are standard. The command $\cassume{b}$ acts like $\cskip{}$ if $b$ holds
and otherwise stops the execution.
Instead of including \emph{deterministic} if-statements and while loops, we consider a \emph{non-deterministic} choice $\cnif{C_1}{C_2}$ and a \emph{non-deterministic} iteration $\cnwhile{C}$, which are more expressive.
Combined with the $\mathbf{assume}$ command, they can express deterministic if-statements and while loops as follows:
\begin{align*}
\cif{b}{C_1}{C_2} &\triangleq \cnif{ (\cassume{b} \cseq C_1)}{ (\cassume{\lnot b} \cseq C_2) } \\
\cifonly{b}{C} &\triangleq \cnif{ (\cassume{b} \cseq C)}{ (\cassume{\lnot b}) } \\
\cwhile{b}{C} &\triangleq \cnwhile{ (\cassume{b} \cseq C) } \cseq \cassume{\lnot b}
\end{align*}

Our language also includes a non-deterministic assignment $\chavoc{y}$ (also called \emph{havoc}),
which allows us to model unbounded non-determinism.
Together with $\mathbf{assume}$, it can for instance model the generation of random numbers between bounds:
$\cassign{y}{\mathit{randIntBounded}(a, b)}$ can be modeled as
$\chavoc{y} \cseq \cassume{a \leq y \leq b}$.

The big-step semantics of our language is standard, and formally defined in \ifextended{\appref{app:definitions} (\figref{fig:semantics})}{our extended version~\cite{hhlExtended}}.
The rule for $\chavoc{x}$ allows $x$ to be updated with any value $v$. $\cassume{b}$ leaves the state unchanged if $b$ holds; otherwise, the semantics gets stuck to indicate that there is no execution in which $b$ does \emph{not} hold. The command $\cnif{C_1}{C_2}$ non-deterministically executes either $C_1$ or $C_2$. $\cnwhile{C}$ non-deterministically either performs another loop iteration or terminates.

Note that our language does not contain any command that could fail (in particular, expression evaluation is total, such that division-by-zero and other errors cannot occur). Runtime failures could easily be modeled by instrumenting the program with a special Boolean variable $\mathit{err}$ that tracks whether a runtime error has occurred and skips the rest of the execution if this is the case.

\subsection{Hyper-Triples, Formally}\label{subsec:hyper-triples}

As explained in \secref{sec:hyper-triples}, the key idea behind Hyper Hoare Logic is to use \emph{properties of sets of states} as pre- and postconditions, whereas traditional Hoare logics use properties of individual states (or of a given number $k$ of states in logics for hyperproperties). Considering arbitrary sets of states increases the expressivity of triples substantially; for instance, universal and existential quantification over these sets corresponds to over- and underapproximate reasoning, respectively.
Moreover, combining both forms of quantification allows one to express advanced hyperproperties, such as generalized non-interference (see \secref{subsec:beyond-k-safety}).

To allow the assertions of Hyper Hoare Logic to refer to logical variables (motivated in \secref{subsec:k-safety}), we include them in our notion of state.

\begin{definition}\label{def:extended_states}
    \textbf{Extended states.}
    An \emph{extended state} (ranged over by $\varphi$) is a pair of a \emph{logical state} (a total mapping from logical variables to logical values)
    and a program state:
    $$
    \extstates \triangleq (\lvars \rightarrow \lvals) \times \states
    $$

\noindent
Given an extended state $\varphi$, we write $\lproj{\varphi}$ to refer to the logical state and $\pproj{\varphi}$ to refer to the program state, that is,
    $\varphi = (\lproj{\varphi}, \pproj{\varphi})$.
\end{definition}

We use the same meta variables ($x$, $y$, $z$) for program and logical variables. When it is clear from the context that $x \in \vars$ (resp.\ $x \in \lvars$), we often write $\varphi(x)$ to denote $\pproj{\varphi}(x)$ (resp.\ $\lproj{\varphi}(x)$).

The assertions of Hyper Hoare Logic are predicates over sets of extended states: 

\begin{definition}\label{def:hyperassertions}\textbf{Hyper-assertions.}
    A \emph{\hyperassertion{}} (ranged over by $P$, $Q$, $R$) is a total function from $\powerset{ \extstates{} }$ to Booleans.

    \noindent A \hyperassertion{} $P$ \emph{entails} a \hyperassertion{} $Q$, written $P \models Q$, iff all sets that satisfy $P$ also satisfy $Q$:
    $$(P \models Q) \triangleq (\forall S \ldotp \sat{S}{P} \Rightarrow \sat{S}{Q})$$
\end{definition}

We formalize hyper-assertions as semantic properties, which allows us to focus on the key ideas of the logic.
In \secref{sec:syntactic}, we will define a syntax for hyper-assertions, which will allow us to derive simpler rules than the ones presented in this section.

To formalize the meaning of hyper-triples, we need to relate them formally to the semantics of our programming language. Since hyper-triples are defined over extended states, we first define a semantic function $\mathit{sem}$ that lifts the operational semantics to extended states; it yields the set of extended states that can be reached by executing a command $C$ from a set of extended states~$S$:

\begin{definition}\textbf{Extended semantics.}\label{def:extended_semantics}
$$
\mathit{sem}(C, S) \triangleq \{ \varphi \mid \exists \sigma \ldotp (\lproj{\varphi}, \sigma) \in S \land
\bigstep{C}{\sigma}{\pproj{\varphi}} \}
$$
\end{definition}

The following lemma states several useful properties of the extended semantics.%

\begin{lemma}\textbf{Properties of the extended semantics.}\label{lem:sem}
    \begin{enumerate}
        \item $\sem(C, S_1 \cup S_2) = \sem(C, S_1) \cup \sem(C, S_2)$
        \item $S \subseteq S' \Rightarrow \sem(C, S) \subseteq \sem(C, S')$
        \item $\sem(C, \bigcup_x f(x)) = \bigcup_x \sem(C, f(x))$
        \item $\sem(\cskip, S) = S$
        \item $\sem(C_1 \cseq C_2, S) = \sem(C_2, \sem(C_1, S))$
        \item $\sem(\cnif{C_1}{C_2}, S) = \sem(C_1, S) \cup \sem(C_2, S)$
        \item $\sem(\cnwhile{C}, S) = \bigcup_{n \in \mathbb{N}} \sem(C^n, S)$
        where $C^n \triangleq \underbrace{C \cseq \ldots \cseq C}_{\text{n times}}$
    \end{enumerate}
\end{lemma}

Using the extended semantics, we can now define the meaning of hyper-triples.

\begin{definition}\label{def:hoare_triples}\textbf{Hyper-triples.}
    Given two \hyperassertions{} $P$ and $Q$, and a command $C$,
    the \emph{hyper-triple} $\simpleHoare{P}{C}{Q}$ is \emph{valid},
    written $\hoare{P}{C}{Q}$, iff for any set $S$ of initial extended states that satisfies $P$,
    the set $\mathit{sem}(C, S)$ of extended states reachable by executing $C$ in some state from $S$
    satisfies $Q$:
    $$\hoare{P}{C}{Q} \triangleq \left( \forall S \ldotp \sat{S}{P} \Rightarrow \sat{\mathit{sem}(C, S)}{Q} \right)$$
\end{definition}

This definition is similar to classical Hoare logic, where the initial and final states have been replaced by \emph{sets} of extended states.
As we have seen in \secref{sec:hyper-triples}, hyper-assertions over sets of states allow our hyper-triples to express properties of single executions
and of multiple executions (hyperproperties), as well as to perform overapproximate reasoning (like \eg Hoare Logic)
and underapproximate reasoning (like \eg Incorrectness Logic).

\subsection{Core Rules}\label{subsec:rules}

\begin{figure*}[t!]
\begin{center}
    \footnotesize
    \[
    \begin{array}{c}

    \Inf[\mathit{Skip}]{\shoare{P}{\cskip}{P}}
    \hspace{4mm}

    \Inf[\mathit{Seq}]{\shoare{P}{C_1}{R}}{\shoare{R}{C_2}{Q}}{\shoare{P}{C_1 \cseq C_2}{Q}}
    \hspace{4mm}
    
    \Inf[\mathit{Choice}]{\shoare{P}{C_1}{Q_1}}{\shoare{P}{C_2}{Q_2}}{\shoare{P}{\cnif{C_1}{C_2}}{Q_1 \otimes Q_2}}

    \\[2em]

    \Inf[\mathit{Cons}]{P \models P'}{Q' \models Q}{\shoare{P'}{C}{Q'}}{\shoare{P}{C}{Q}}
    \hspace{4mm}

    \Inf[\mathit{Assume}]
    {\shoare{\lambda S \ldotp \sat{\{ \varphi \mid \varphi \in S \land b(\pproj{\varphi}) \} }{P} }{\cassume{b}}{P}}

    \\[2em]

    \Inf[\mathit{Exist}]{\forall x  \ldotp (\shoare{P_x}{C}{Q_x})}{\shoare{\exists x \ldotp P_x}{C}{\exists x \ldotp Q_x}}

    \hspace{4mm}

    \Inf[\mathit{Assign}]{\shoare{\lambda S \ldotp \sat{
        \{ \varphi \mid \exists \alpha \in S \ldotp \lproj{\varphi} = \lproj{\alpha} \land \pproj{\varphi} = \pproj{\alpha}[x \mapsto e(\pproj{\varphi})] \}
    }{P}}{\cassign{x}{e}}{P}}

    \\[2em]

    \Inf[\mathit{Iter}]{\shoare{I_n}{C}{I_{n+1}}}{\shoare{I_0}{\cnwhile{C}}{\otimesn I_n}}

    \hspace{4mm}
    \Inf[\mathit{Havoc}]{\shoare{\lambda S \ldotp \sat{\{ \varphi \mid
    \exists \alpha \in S \ldotp \exists v \ldotp \lproj{\varphi} = \lproj{\alpha} \land \pproj{\varphi} = \pproj{\alpha}[x \mapsto v] \}}{P} }{\chavoc{x}}{P}}

    \end{array}
    \]
\end{center}

\caption{Core rules of \logic{}. The meaning of the operators $\otimes$ and $\otimesn$ are defined in \defref{def:otimes} and \defref{def:bigotimes}, respectively.}
\label{fig:rules}
\end{figure*}

\figref{fig:rules} shows the core rules of Hyper Hoare Logic. 
\namerule{Skip}, \namerule{Seq}, \namerule{Cons}, and \namerule{Exist} are analogous to traditional Hoare logic.
\namerule{Assume}, \namerule{Assign}, and \namerule{Havoc} are straightforward given the semantics of these commands.
All three rules work backward.
In particular, the precondition of  \namerule{Assume} applies the postcondition $P$ only to those states that satisfy the assumption $b$.
By leaving the value $v$ unconstrained, \namerule{Havoc} considers as precondition the postcondition $P$ for all possible values for $x$.
The three rules \namerule{Assume}, \namerule{Assign}, and \namerule{Havoc} are optimized for expressivity;
we will derive in \secref{sec:syntactic} syntactic versions of these rules, which are less expressive, but easier to apply.

The rule \namerule{Choice} (for non-deterministic choice) is more involved. Most standard Hoare logics use the same assertion $Q$ as postcondition of all three triples. However, such a rule would not be sound in \logic. Consider for instance an application of this hypothetical \namerule{Choice} rule where both $P$ and $Q$ are defined as 
$\lambda S \ldotp \lvert S \rvert = 1$, expressing that there is a single pre- and post-state. If commands $C_1$ and $C_2$ are deterministic, the antecedents of the rule can be \proven{} because a single pre-state leads to a single post-state. However, the non-deterministic choice will in general produce \emph{two} post-states, such that the  postcondition is violated. 

To account for the effects of non-determinism on the sets of states described by hyper-assertions, we obtain the postcondition of the non-deterministic choice by combining the postconditions of its branches. As shown by \lemref{lem:sem} (point 6), executing the non-deterministic choice $\cnif{C_1}{C_2}$ in the set of states $S$ amounts to executing $C_1$ in $S$ and $C_2$ in $S$, and taking the union of the two resulting sets of states. Thus, if $\sat{\sem(C_1, S)}{Q_1}$ and $\sat{\sem(C_2, S)}{Q_2}$ hold
then the postcondition of $\cnif{C_1}{C_2}$ must characterize the union $\sem(C_1, S) \cup \sem(C_2, S)$
The postcondition of the rule \namerule{Choice}, $Q_1 \otimes Q_2$, achieves that:

\begin{definition}\label{def:otimes}
A set $S$ satisfies $Q_1 \otimes Q_2$ iff it can be split into two (potentially overlapping) sets $S_1$ and $S_2$ (the sets of post-states of the branches), such that $S_1$ satisfies $Q_1$ and $S_2$ satisfies ${Q_2}$:
$$(Q_1 \otimes Q_2)(S) \triangleq \left( \exists S_1, S_2 \ldotp S = S_1 \cup S_2 \land Q_1(S_1) \land Q_2(S_2) \right)$$
\end{definition}

The rule \namerule{Iter} for non-deterministic iterations generalizes our treatment of non-deterministic choice. It employs an indexed loop invariant $I$, which maps a natural number $n$ to a \hyperassertion{} $I_n$.
$I_n$ characterizes the set of states reached after executing $n$ times the command $C$ in a set of initial states that satisfies $I_0$. Analogously to the rule \namerule{Choice}, the indexed invariant avoids using the same hyper-assertion for all non-deterministic choices. The precondition of the rule's conclusion and its premise prove (inductively) that the triple $\simpleHoare{I_0}{C^n}{I_n}$ holds for all $n$.
$I_n$ thus characterizes the set of reachable states after exactly $n$ iterations of the loop.
Since our loop is non-deterministic (\ie has no loop condition), the set
of reachable states after the loop is the union of the sets of reachable states after each iteration.
The postcondition of the conclusion captures this intuition,
by using the generalized version of the $\otimes$ operator to an indexed family of hyper-assertions:

\begin{definition}\label{def:bigotimes}
A set $S$ satisfies $\bigotimes_{n \in \mathbb{N}} I_n$ iff it can be split into
$\bigcup_i f(i) = f(0) \cup \ldots \cup f(i) \cup \ldots$, where $f(i)$ (the set of reachable states after exactly $i$ iterations) satisfies $I_i$ (for each $i\in \mathbb{N}$):
$$(\otimesn I_n)(S) \triangleq \left(\exists f \ldotp (S = \bigcup_{n \in \mathbb{N}} f(n)) \land (\forall n \in \mathbb{N} \ldotp I_n(f(n))) \right)$$
\end{definition}

Note that this rule makes \logic{} a partial correctness logic:
it only considers an unbounded, but finite number $n$ of loop iterations.
\thibault{%
In \ifextended{\appref{app:termination}}{our extended version~\cite{hhlExtended}}, we present an alternative, stronger definition of hyper-triples,
which express the existence of a terminating execution from any initial state,
and corresponding inference rules (proven sound in our \isabelle{} mechanization)
such as a rule for while loops based on a loop variant.
}
We also discuss a possible extension of \logic{} to prove non-termination, \ie the existence of non-terminating executions.

\subsection{Soundness and Completeness}\label{subsec:sound-complete}

We have \proven{} in Isabelle/HOL that \logic{} is sound and complete. That is, every hyper-triple that can be derived in the logic is valid, and vice versa. 
Note that \figref{fig:rules} contains only the \emph{core rules} of \logic{}. These are sufficient to prove completeness; all rules presented later in this paper are only useful to make proofs more succinct and natural.

\begin{theorem}\textbf{Soundness.}\label{thm:sound}
Hyper Hoare Logic is sound:
    \begin{center}
        If $\shoare{P}{C}{Q}$ then $\hoare{P}{C}{Q}$.
    \end{center}
\end{theorem}

\thibault{
\begin{proof}
    We prove $\forall P, Q \ldotp \shoare{P}{C}{Q} \Rightarrow \hoare{P}{C}{Q}$ by straightforward structural induction on $C$.
    The cases for $\cskip$, $C_1 \cseq C_2$, $\cnif{C_1}{C_2}$, and $\cnwhile{C}$,
    directly follow from \lemref{lem:sem}.
\end{proof}
}

\begin{theorem}\textbf{Completeness.}\label{thm:complete}
Hyper Hoare Logic is complete:
    \begin{center}
        If $\hoare{P}{C}{Q}$ then $\shoare{P}{C}{Q}$.
    \end{center}
\end{theorem}

\thibault{
\proofsketch{
    We prove $H(C) \triangleq (\forall P, Q \ldotp \hoare{P}{C}{Q} \Rightarrow \shoare{P}{C}{Q})$
    by structural induction over $C$. We show the case for $C \triangleq \cnif{C_1}{C_2}$; the proof for the non-deterministic iteration is analogous,
    and the other cases are standard or straightforward.

    We assume $H(C_1)$ and $H(C_2)$, and want to prove $H(C)$ where $C \triangleq \cnif{C_1}{C_2}$.
    As we illustrate after this proof sketch, we need to consider each possible value $V$ of the set of extended states $S$ separately.
    For an arbitrary value $V$, we define $P_V \triangleq (\lambda S \ldotp P(S) \land S = V)$,
    $R^1_V \triangleq (\lambda S \ldotp S = \sem(C_1, V) \land P(V))$, and $R^2_V \triangleq (\lambda S \ldotp S = \sem(C_2, V))$.
    We get $\shoare{P_V}{C_1}{R^1_V}$ from $H(C_1)$ and $\shoare{P_V}{C_2}{R^2_V}$ from $H(C_2)$.
    By applying the rule \namerule{Choice}, we get $\shoare{P_V}{C}{R^1_V \otimes R^2_V}$. Since we prove this triple for an arbitrary value $V$ (that is, for all $V$), we can apply the rule \namerule{Exist}, to obtain
    $\shoare{ \exists V \ldotp P_V }{C}{\exists V \ldotp R^1_V \otimes R^2_V }$.
    $P$ clearly entails $\exists V \ldotp P_V$,
    and
    the postcondition $\exists V \ldotp R^1_V \otimes R^2_V$ entails $(\lambda S \ldotp
    \exists V \ldotp P(V) \land S = \sem(C_1, V) \cup \sem(C_2, V))$,
    which precisely describes the sets of states $\sem(\cnif{C_1}{C_2}, V)$ (see \lemref{lem:sem}(6)) where $V$ satisfies $P$,
    and thus entails $Q$.
    By rule \namerule{Cons}, we get $\shoare{P}{C}{Q}$, which concludes the case.
}
}

Note that our completeness theorem is not concerned with the expressivity of the assertion language because we use \emph{semantic} hyper-assertions (\ie functions, see \defref{def:hyperassertions}). Similarly, by using semantic entailments in the rule \namerule{Cons},
we decouple the completeness of Hyper Hoare Logic from the completeness of the logic used to derive entailments.

Interestingly, the logic would \emph{not} be complete without the core rule \namerule{Exist},
as we illustrate with the following simple example:

\begin{example}
    Let $\varphi_v$ be the state that maps $x$ to $v$ and all other variables to $0$.
    Let $P_v \triangleq (\lambda S \ldotp S = \{ \varphi_v \})$.
    Clearly, the hyper-triples
    $\simpleHoare{P_0}{\cskip }{P_0}$,
    $\simpleHoare{P_2}{\cskip }{P_2}$,
    $\simpleHoare{P_0}{\cassign{x}{x+1}}{P_1}$,
    and
    $\simpleHoare{P_2}{\cassign{x}{x+1}}{P_3}$
    are all valid.
We would like to prove the hyper-triple
    $
    \simpleHoare{P_0 \lor P_2}{ \cnif{\cskip}{(\cassign{x}{x+1})} }{ \lambda S \ldotp S = \{ \varphi_0, \varphi_1 \} \lor S = \{ \varphi_2, \varphi_3 \}}
    $.
    That is, either $P_0$ holds before, and then we have $S = \{ \varphi_0, \varphi_1 \}$ afterwards, or $P_2$ holds before, and then we have $S = \{ \varphi_2, \varphi_3 \}$ afterwards.
    However, using the rule \namerule{Choice} only, the most precise triple we can prove is
    $$
    \Inf[\mathit{Choice}]{
        \shoare{P_0 \lor P_2}{\cskip }{ P_0 \lor P_2 }
    }{
        \shoare{P_0 \lor P_2}{\cassign{x}{x+1}}{ P_1 \lor P_3 }
    }{
        \shoare{P_0 \lor P_2}{ \cnif{ \cskip}{ (\cassign{x}{x+1}) } }{ (P_0 \lor P_2) \otimes (P_1 \lor P_3) }
    }
    $$
    The postcondition $(P_0 \lor P_2) \otimes (P_1 \lor P_3)$ is equivalent to
    $(P_0 \otimes P_1) \lor (P_0 \otimes P_3) \lor (P_2 \otimes P_1) \lor (P_2 \otimes P_3)$,
    \ie $\lambda S \ldotp S = \{ \varphi_0, \varphi_1 \} \lor S = \{ \varphi_0, \varphi_3 \} \lor S = \{ \varphi_2, \varphi_1 \} \lor S = \{ \varphi_2, \varphi_3 \}$.
    We thus have two spurious disjuncts, $P_0 \otimes P_3$ (\ie $S = \{ \varphi_0, \varphi_3 \}$) and $P_2 \otimes P_1$ (\ie $S = \{ \varphi_2, \varphi_1 \}$).
\end{example}

This example shows that the rule \namerule{Choice} on its own is not precise enough for the logic to be complete; we need at least a \emph{disjunction} rule to distinguish the two cases $P_0$ and $P_2$.
In general, however, there might be an infinite number of cases to consider,
which is why we need the rule \namerule{Exist}.
The premise of this rule allows us to \emph{fix} a set of states $S$ that satisfies some precondition $P$
and to prove the most precise postcondition for the precondition $\lambda S' \ldotp S = S'$;
combining these precise postconditions with an existential quantifier in the conclusion of the rule
allows us to obtain the most precise postcondition for the precondition $P$.

\thibault{For our example, we can use the rule \namerule{Exist} with a Boolean $b$ that records whether $P_0$ or $P_2$ is satisfied initially, as follows:}
$$
    \footnotesize
    \Inf[\mathit{Exist}]{
    \Inf[\mathit{Choice}]{
        \shoare{(b \simp P_0) \sand (\lnot b \simp P_2)}{\cskip }{ (b \simp P_0) \sand (\lnot b \simp P_2)}
    }
    {
        \shoare{(b \simp P_0) \sand (\lnot b \simp P_2)}{\cassign{x}{x+1} }{ (b \simp P_1) \sand (\lnot b \simp P_3)}
    }{
        \shoare{(b \simp P_0) \sand (\lnot b \simp P_2)}{\cnif{\cskip}{(\cassign{x}{x+1})} }{
        ((b \simp P_0) \sand (\lnot b \simp P_2)) \otimes ((b \simp P_1) \sand (\lnot b \simp P_3))
    }
    }}
    {
        \shoare{\underbrace{\exists b \ldotp (b \simp P_0) \sand (\lnot b \simp P_2)}_{ = P_0 \lor P_2  }}{
            \cnif{\cskip}{(\cassign{x}{x+1})} }{
        \underbracket{\exists b \ldotp ((b \simp P_0) \sand (\lnot b \simp P_2)) \otimes ((b \simp P_1) \sand (\lnot b \simp P_3))}_{
            = (P_0 \otimes P_2) \lor (P_1 \otimes P_3)
        } }
    }
$$

\subsection{Expressivity of Hyper-Triples}
\label{subsec:expressivity}

In the previous subsection, we have shown that \logic{} is sound and complete to establish the validity of hyper-triples,
and, thus, \logic{} is as expressive as hyper-triples.
We now show that hyper-triples are expressive enough to capture arbitrary hyperproperties over finite program executions.
A \emph{hyperproperty}~\cite{hyperproperties} is traditionally defined as a property of sets of \emph{traces} of a system, that is, of sequences of system states. Since Hoare logics typically consider only the initial and final state of a program execution, we use a slightly adapted definition here:

\begin{definition}\label{def:hyperproperties}\textbf{Program hyperproperties.}
    A \emph{program hyperproperty} is a set of
    sets of pairs of program states, 
    \ie an element of $\powerset{\powerset{\states \times \states}}$.

    A command $C$ satisfies the program hyperproperty $\phyperpropertyfont{H}$ iff
    the set of all pairs of pre- and post-states of $C$ is an element of $\phyperpropertyfont{H}$:
    $
    \hypersat{\{ (\sigma, \sigma') \mid \bigstep{C}{\sigma}{\sigma'} \}}{H}
    $.
\end{definition}

Equivalently, a program hyperproperty can be thought of as a predicate over $\powerset{\states \times \states}$. Note that this definition subsumes
\thibault{properties of single executions},
such as functional correctness properties. 

In contrast to traditional hyperproperties, our program hyperproperties describe only the \emph{finite} executions of a program, that is, those that reach a final state.
An extension of \logic{} to infinite executions might be possible by defining hyper-assertions over sets of traces rather than sets of states; we leave this as future work.
In the rest of this paper, when the context is clear, we use \emph{hyperproperties} to refer to \emph{program hyperproperties}.

Any program hyperproperty can be expressed as a hyper-triple in \logic{}:

\begin{theorem}\label{thm:expressing_hyperprop}\textbf{Expressing hyperproperties as hyper-triples.}
    Let $\phyperpropertyfont{H}$ be a program hyperproperty.
    Assume that the cardinality of $\lvars$ is at least the cardinality of $\vars$,
    and that the cardinality of $\lvals$ is at least the cardinality of $\vals$.

    Then there exist \hyperassertions{} $P$ and $Q$ such that,
    for any command $C$,
    $\hypersat{C}{H}$
    iff $\hoare{P}{C}{Q}$.
\end{theorem}

\proofsketch{
We define the precondition $P$ such that the set of initial states contains all program states, and the values of all program variables in these states are recorded in logical variables (which is possible due to the cardinality assumptions).
Since the logical variables are not affected by the execution of $C$, they allow $Q$ to refer to the initial values of any program variable, in addition to their values in the final state. Consequently, $Q$ can describe all possible pairs of pre- and post-states.
We simply define $Q$ to be true iff the set of these pairs is contained in $\phyperpropertyfont{H}$.
}

\thibault{We also proved the converse: every hyper-triple describes a program hyperproperty. That is, hyper-triples capture exactly the hyperproperties over finite executions.}

\thibault{
\begin{theorem}\textbf{Expressing hyper-triples as hyperproperties.}
    \label{thm:expressing-hyper-triples}
    For any hyper-assertions $P$ and $Q$, there exists a hyperproperty $\phyperpropertyfont{H}$
    such that, for any command $C$,
    $\hypersat{C}{H}$
    iff $\hoare{P}{C}{Q}$.
\end{theorem}
}

\thibault{
\begin{proof}
    We define $\phyperpropertyfont{H} \triangleq \left\{ \Sigma \mid \forall S \ldotp
    P(S) \Rightarrow Q( \{  (l, \sigma') \mid \exists \sigma \ldotp  (l, \sigma) \in S \land (\sigma, \sigma') \in \Sigma \} )
    \right\}$.
\end{proof}
}

Combined with our completeness result (\thmref{thm:complete}),
this theorem implies that, if a command $C$ satisfies a hyperproperty $\phyperpropertyfont{H}$ then there exists a proof of it in \logic{}.
More surprisingly, our logic also allows us to \emph{disprove} any hyperproperty: If $C$ does \emph{not} satisfy $\phyperpropertyfont{H}$ then $C$ satisfies the \emph{complement} of 
$\phyperpropertyfont{H}$, which is also a hyperproperty, and thus can also be \proven{}. Consequently, \logic{} can prove or disprove any \emph{program hyperproperty} as defined in \defref{def:hyperproperties}.

Since hyper-triples express hyperproperties (\thmref{thm:expressing_hyperprop} and \thmref{thm:expressing-hyper-triples}), the ability to disprove any hyperproperty implies that \logic{} can also disprove any \emph{hyper-triple}.
More precisely, one can \emph{always} use \logic{} to prove that some hyper-triple $\simpleHoare{P}{C}{Q}$ is \emph{invalid},
by proving the validity of another hyper-triple $\simpleHoare{P'}{C}{\lnot Q}$, where $P'$ is a satisfiable \hyperassertion{} that entails $P$.
Conversely, the validity of such a hyper-triple $\simpleHoare{P'}{C}{\lnot Q}$ implies
that all hyper-triples $\simpleHoare{P}{C}{Q}$ (with $P$ weaker than $P'$) are \emph{invalid}.
The following theorem precisely expresses this observation:

\begin{theorem}\label{thm:disproving}\textbf{Disproving hyper-triples.}
    Given $P$, $C$, and $Q$,
    the following two propositions are equivalent:
    \begin{enumerate}
        \item $\hoare{P}{C}{Q}$ does not hold.
        \item There exists a \hyperassertion{} $P'$ that is satisfiable, entails $P$,
        and $\hoare{P'}{C}{ \lnot Q }$.
    \end{enumerate}
\end{theorem}

\thibault{
\begin{proof}
    By negating \defref{def:hoare_triples},
    we get that point (1) is equivalent to the existence of a set of extended states $S$ such that
    $\sat{S}{P}$ holds but $\sat{\sem(C, S)}{Q}$ does not, \ie{} $\sat{\sem(C, S)}{\lnot Q}$ holds.
    Let $P' \triangleq (\lambda S' \ldotp S = S')$. $P'$ is clearly satisfiable.
    Moreover, point (1) implies that $P'$ entails $P$,
    and that $\hoare{P'}{C}{\lnot Q}$ holds. Thus, (1) implies (2).

    Assuming (2), we get that there exists a set of extended states $S$ such that $\sat{S}{P'}$ (since $P'$ is satisfiable)
    and $\sat{\sem(C, S)}{\lnot Q}$ hold.
    Since $P'$ entails $P$, $\sat{S}{P}$ holds, which implies (1).
\end{proof}
}

We need to strengthen $P$ to $P'$ in point~(2), because there might be some sets $S$, $S'$ that both satisfy $P$, such that $Q(\sem(C, S))$ holds, but $Q(\sem(C, S'))$ does not.
This was the case for our examples in \secref{subsec:k-safety} and \secref{subsec:beyond-k-safety}; 
for instance, one of the preconditions there was strengthened to 
include $\exists \inSet{\varphi_1}, \inSet{\varphi_2} \ldotp \varphi_1(h) \neq \varphi_2(h)$.

\thmref{thm:disproving} is another illustration of the expressivity of \logic{}. The corresponding result does \emph{not} hold in traditional Hoare logics.
For example, \thibault{(1)}~the classical Hoare triple $\normalHoare{\top}{\chavoc{x}}{x \geq 5}$ does not hold,
but \thibault{(2)}~there is \emph{no} satisfiable $P$ such that $\normalHoare{P}{\chavoc{x}}{ \lnot (x \geq 5)}$ holds.
\thibault{Moreover, to disprove this triple, one needs to show the \emph{existence} of an execution that satisfies the negated postcondition, which is not expressible in HL.}
In contrast, \logic{} can disprove the classical Hoare triple by proving the hyper-triple
$\simpleHoare{\exists \inSet{\varphi} \ldotp \top}{\chavoc{x}}{ \neg(\forall \inSet{\varphi} \ldotp \varphi(x) \geq 5) }$.

The correspondence between hyper-triples and program hyperproperties (\thmref{thm:expressing_hyperprop} and \thmref{thm:expressing-hyper-triples}), together with our completeness result (\thmref{thm:complete}) precisely characterizes the expressivity of \logic{}.
In \ifextended{\appref{app:expressivity}}{our extended version~\cite{hhlExtended}},
we show systematic ways to express the judgments of existing over- and underapproximating Hoare logics as hyper-triples.

\subsection{Compositionality}
\label{subsec:compositionality}

The core rules of \logic{} allow one to prove any valid hyper-triple, but not necessarily \emph{compositionally}.
As an example, consider the sequential composition of a command $C_1$ that satisfies \emph{generalized} non-interference (GNI) with a command $C_2$ that satisfies non-interference (NI).
We would like to prove that $C_1 \cseq C_2$ satisfies GNI (the weaker property).
As discussed in \secref{subsec:beyond-k-safety},
a possible postcondition for $C_1$ is
$\gni{l}{h} \triangleq (\forall \inSet{\varphi_1}, \inSet{\varphi_2} \ldotp \exists \inSet{\varphi} \ldotp \varphi_1(h) = \varphi(h) \land \varphi(l) = \varphi_2(l))$,
while a possible precondition for $C_2$ is
$\low{l} \triangleq (\forall \inSet{\varphi_1}, \inSet{\varphi_2} \ldotp \varphi_1(l) = \varphi_2(l))$.
Unfortunately, the corresponding hyper-triples for $C_1$ and $C_2$ cannot be composed using the core rules.
In particular, rule \namerule{Seq} cannot be applied (even in combination with \namerule{Cons}), since the postcondition of $C_1$ does not imply the precondition of $C_2$.
Note that this observation does \emph{not} contradict completeness: By
\thmref{thm:complete}, it is possible to prove \emph{more precise} triples for $C_1$ and $C_2$, such that the postcondition of $C_1$ matches the precondition of $C_2$.
However, to enable modular reasoning, our goal is to construct the proof by composing the given triples for the individual commands rather than deriving new ones.

We have \proven{} \thibault{(in \isabelle)} a number of useful \emph{compositionality rules} for hyper-triples,
which are presented in \ifextended{\appref{app:compositionality}}{%
~\cite{hhlExtended}}.
These rules are \emph{admissible} in \logic{}, in the sense that they do not modify the set of valid hyper-triples that can be \proven{}.
Rather, they enable flexible compositions of hyper-triples.
\thibault{As an example, we have \proven{} the following compositionality rule
\begin{center}
$ \small
    \Inf[\mathit{BigUnion}]{ \shoare{P}{C}{Q} }{ \shoare{\bigotimes P}{C}{\bigotimes Q} }
$
\end{center}
where $\bigotimes P \triangleq \left( \lambda S \ldotp \exists F \ldotp (S = \bigcup_{S' \in F} S') \land (\forall S' \in F \ldotp P(S')) \right)$.
This rule is useful when we have a set of states $S$ that does not directly satisfy the precondition $P$,
but that can be decomposed as a \emph{union} of smaller sets $S'$ that individually satisfy $P$.
After executing $C$, we get a set of states that is a union of smaller sets that individually satisfy $Q$.
}

\thibault{In our aforementioned example, we have a set of states $S$ that satisfies the postcondition $\forall \inSet{\varphi_1}, \inSet{\varphi_2} \ldotp \exists \inSet{\varphi} \ldotp \varphi_1(h) = \varphi(h) \land \varphi(l) = \varphi_2(l)$ of $C_1$.
While $S$ does not necessarily satisfy the precondition $\low{l}$ of $C_2$,
$S$ can be seen as a union of smaller sets of states $\{ \varphi, \varphi_2 \}$ that individually satisfy the precondition of $C_2$.
Thus, using the compositionality rule \namerule{BigUnion},
we obtain after $C_2$ a set of states that is a union of smaller sets that all satisfy the postcondition of $C_2$,
from which we can then prove the desired postcondition $\gni{l}{h}$ for $C_1 \cseq C_2$.%
\footnote{\ifextended{As we show in \appref{subsect:comp-examples}, we}{We}
also need to prove that $C_2$ does not drop executions depending on the value of $h$.}
The full proof of this example is
\ifextended{presented in \appref{subsect:comp-examples}}{shown in our extended version~\cite{hhlExtended}},
along with another challenging example.
}

%% file: syntactic.tex
The core rules presented in \secref{sec:logic} are optimized for expressivity: They are
sufficient to prove \emph{any} valid hyper-triple (\thmref{thm:complete}), but not necessarily in the simplest way. %
In particular, the rules for atomic statements \namerule{Assume}, \namerule{Assign}, and \namerule{Havoc} require a set comprehension in the precondition,
which is necessary when dealing with arbitrary semantic hyper-assertions.
However, by imposing syntactic restrictions on hyper-assertions,
we can derive simpler rules,
as we show in this section.
In \secref{subsec:syntactic-assertions}, we define a syntax for hyper-assertions, in which the set of states occurs only as range of  universal and existential quantifiers.
As we have seen in \secref{sec:hyper-triples} and formally show in
\ifextended{\appref{app:expressivity}}{our extended version~\cite{hhlExtended}}, this syntax is sufficient to capture many useful hyperproperties.
Moreover, it allows us to derive simple rules for assignments (\secref{subsec:assignment-rules})
and assume statements (\secref{subsec:assume-rule}).
All rules presented in this section have been proven sound in \isabelle.

\subsection{Syntactic Hyper-Assertions}
\label{subsec:syntactic-assertions}

We define a restricted class of syntactic hyper-assertions, which can interact with the set of states only through universal and existential quantification over its states:
\begin{definition}
\textbf{Syntactic hyper-expressions and hyper-assertions.}\\
\emph{Hyper-expressions} $e$ are defined by the following syntax,
where $\varphi$ ranges over states, $x$ over (program or logical) variables, $y$ over quantified variables, $c$ over literals,
$\oplus$  over binary operators (such as $+, -, *$ for integers, $\concat{}$ for lists, etc.),
and $f$ denotes functions from values to values (such as $\mathit{len}$ for lists):
$$\mathit{e} \Coloneqq
c \mid
y \mid
\pproj{\varphi}(x) \mid
\lproj{\varphi}(x) \mid
e \oplus e \mid
f(e)
$$
\emph{Syntactic hyper-assertions} $A$ are defined by the following syntax,
where $e$ ranges over hyper-expressions,
$b$ over boolean literals, and $\succeq$ over binary operators (such as $=, \neq, <, >, \leq, \geq, \ldots$):
$$A \Coloneqq
b \mid
e \succeq e \mid
A \lor A \mid
A \land A \mid
\forall y \ldotp A \mid
\exists y \ldotp A \mid
\forall \inSet{\varphi} \ldotp A \mid
\exists \inSet{\varphi} \ldotp A
$$
\end{definition}

Note that \emph{hyper-expressions} are different from \emph{program} expressions,
since the latter can only refer to program variables of a \emph{single} implicit state (\eg $x = y + z$),
while the former can explicitly refer to different states (\eg{} $\varphi(x) = \varphi'(x)$).
Negation $\lnot A$ is defined recursively in the standard way.
We also define
$(A \Rightarrow B) \triangleq (\lnot A \lor B)$,
$\emp{} \triangleq (\forall \inSet{\varphi} \ldotp \bot)$,
and $\always p \triangleq (\forall \inSet{\varphi} \ldotp p(\varphi))$,
where $p$ is a \emph{state}\footnote{\emph{State} expressions refer to a single (implicit) state. In contrast to program expressions, they may additionally refer to logical variables and use quantifiers over values.}
expression.
The evaluation of hyper-expressions and satisfiability of hyper-assertions
are formally defined in \ifextended{\appref{app:definitions}~(\defref{def:sat})}{our extended version~\cite{hhlExtended}}.

\begin{figure*}[t!]
\begin{center}
    \footnotesize
    \[
    \begin{array}{c}

    \Inf[\mathit{AssignS}]{\shoare{\transformAssign{e}{x}{P}}{\cassign{x}{e}}{P}}

    \hspace{4mm}
    \Inf[\mathit{HavocS}]{\shoare{\transformHavoc{x}{P}}{\chavoc{x}}{P}}

    \hspace{4mm}
    \Inf[\mathit{AssumeS}]{\shoare{\transformAssume{b}{P}}{\cassume{b}}{P}}

    \end{array}
    \]
\end{center}

\caption{Some syntactic rules of \logic{}.
The syntactic transformations $\transformAssign{e}{x}{A}$ and $\transformHavoc{x}{A}$ are defined in \defref{def:transformAssigns},
and the syntactic transformation $\transformAssume{b}{\_}$ is defined in \defref{def:transformAssume}.
}
\label{fig:syntactic-rules}
\end{figure*}

\subsection{Syntactic Rules for Deterministic and Non-Deterministic Assignments}
\label{subsec:assignment-rules}

In classical Hoare logic, we obtain the precondition of the rule for the assignment $\cassign{x}{e}$ by substituting $x$ by $e$ in the postcondition.
The \logic{} syntactic rule for assignments \namerule{AssignS}
(\figref{fig:syntactic-rules})
generalizes this idea by repeatedly applying this substitution for \emph{every quantified state}.
This syntactic transformation, written $\transformAssign{e}{x}{\_}$ is defined below.
As an example,
for the assignment $\cassign{x}{y+z}$ and
postcondition $\exists \inSet{\varphi} \ldotp \forall \inSet{\varphi'} \ldotp \varphi(x) \leq \varphi'(x)$,
we obtain the precondition
$\transformAssign{y + z}{x}{\exists \inSet{\varphi} \ldotp \forall \inSet{\varphi'} \ldotp \varphi(x) \leq \varphi'(x) }
= \left( \exists \inSet{\varphi} \ldotp \forall \inSet{\varphi'} \ldotp \varphi(y) + \varphi(z) \leq \varphi'(y) + \varphi'(z) \right) $.

Similarly, our syntactic rule for non-deterministic assignments \namerule{HavocS} substitutes every occurrence of $\varphi(x)$,
for every quantified state $\varphi$,
by a fresh quantified variable $v$.
This variable is universally quantified for universally-quantified states, capturing the intuition that %
we must consider all possible assigned values. %
In contrast, $v$ is existentially quantified for existentially-quantified states, 
because it is sufficient to find one suitable behavior of the non-deterministic assignment.
As an example, for the non-deterministic assignment $\chavoc{x}$ and the aforementioned postcondition, we obtain the precondition
$\transformHavoc{x}{\exists \inSet{\varphi} \ldotp \forall \inSet{\varphi'} \ldotp \varphi(x) \leq \varphi'(x) }
= \left( \exists \inSet{\varphi} \ldotp \exists v \ldotp \forall \inSet{\varphi'} \ldotp \forall v' \ldotp v \leq v' \right) $.

\begin{definition}\textbf{Syntactic transformations for assignments.}%
\\
\label{def:transformAssigns}
$\transformAssign{e}{x}{A}$ yields the \hyperassertion{} $A$, where $\varphi(x)$ is syntactically substituted by $\expApplied{e}{\varphi}$ for all (existentially or universally) quantified states $\varphi$.
The two main cases are:
$$\transformAssign{e}{x}{\forall \inSet{\varphi} \ldotp A} \triangleq \left( \forall \inSet{\varphi} \ldotp \transformAssign{e}{x}{A[\expApplied{e}{\varphi}/\varphi(x)]} \right)
\hspace{6mm}
\transformAssign{e}{x}{\exists \inSet{\varphi} \ldotp A} \triangleq \left( \exists \inSet{\varphi} \ldotp \transformAssign{e}{x}{A[\expApplied{e}{\varphi}/\varphi(x)]} \right)
$$
where $A[y/x]$ refers to the standard syntactic substitution of $x$ by $y$.
Other cases 
apply $\transformAssignEmpty{e}{x}$ recursively
(\eg{} $\transformAssign{e}{x}{A \land B} \triangleq \transformAssign{e}{x}{A} \land \transformAssign{e}{x}{B}$).
\ifextended{The full definition is in \appref{app:definitions}}{%
See~\citet{hhlExtended} for the full definition}.

$\transformHavoc{x}{A}$ yields the \hyperassertion{} $A$ where $\varphi(x)$ is syntactically substituted by a fresh quantified variable $v$,
universally (resp. existentially) quantified for universally (resp. existentially) quantified states.
The two main cases are (where $v$ is fresh):
$$\transformHavoc{x}{\forall \inSet{\varphi} \ldotp A} \triangleq \left( \forall \inSet{\varphi} \ldotp \forall v \ldotp \transformHavoc{x}{A[v/\varphi(x)]} \right)
\hspace{6mm}
\transformHavoc{x}{\exists \inSet{\varphi} \ldotp A} \triangleq \left( \exists \inSet{\varphi} \ldotp \exists v \ldotp \transformHavoc{x}{A[v/\varphi(x)]} \right)$$
Other cases apply $\transformHavocEmpty{x}$ recursively.
\ifextended{The full definition is in \appref{app:definitions}}{%
See~\citet{hhlExtended} for the full definition}.
\end{definition}

\subsection{Syntactic Rules for Assume Statements}
\label{subsec:assume-rule}

Intuitively, $\cassume{b}$ provides additional information when proving properties \emph{for all} states, but imposes an additional requirement when proving \emph{the existence} of a state.
This intuition is captured by the rule \namerule{AssumeS} shown in \figref{fig:syntactic-rules}. The syntactic transformation $\transformAssumeEmpty{b}$
adds the state expression $b$ as an assumption for universally-quantified states,
and as a proof obligation for existentially-quantified states.
As an example, for the statement $\cassume{x \geq 0}$ and the postcondition $\forall \inSet{\varphi} \ldotp \exists \inSet{\varphi'} \ldotp \varphi(x) \leq \varphi'(x)$,
we obtain the precondition
$\transformAssume{x \geq 0}{\forall \inSet{\varphi} \ldotp \exists \inSet{\varphi'} \ldotp \varphi(x) \leq \varphi'(x) }
=
(\forall \inSet{\varphi} \ldotp \varphi(x) \geq 0 \Rightarrow (\exists \inSet{\varphi'} \ldotp \varphi'(x) \geq 0 \land \varphi(x) \leq \varphi'(x) ))$.

\begin{definition}\textbf{Syntactic transformation for assume statements.}\\
\label{def:transformAssume}
The two main cases of $\transformAssumeEmpty{b}$ are
$$\transformAssume{b}{\forall \inSet{\varphi} \ldotp A} \triangleq \left( \forall \inSet{\varphi} \ldotp \expApplied{b}{\varphi} \Rightarrow \transformAssume{b}{A} \right)
\hspace{6mm}
\transformAssume{b}{\exists \inSet{\varphi} \ldotp A} \triangleq \left( \exists \inSet{\varphi} \ldotp \expApplied{b}{\varphi} \land \transformAssume{b}{A} \right)$$
Other cases apply $\transformAssumeEmpty{b}$ recursively.
\ifextended{The full definition is in \appref{app:definitions}}{%
See~\citet{hhlExtended} for the full definition}.
\end{definition}

\begin{figure}[t]
    \scriptsize
    \begin{align*}
    &\outline{
    \exists \inSet{\varphi_1}, \inSet{\varphi_2} \ldotp
\varphi_1(h) \neq \varphi_2(h)} \\
    &\outline{ \exists \inSet{\varphi_1} \ldotp (\exists \inSet{\varphi_2} \ldotp
    (\forall \inSet{\varphi} \ldotp
    \forall v \ldotp
    v \leq 9
    \Rightarrow
    (\varphi(h) = \varphi_1(h)
    \Rightarrow
    \varphi_2(h) + 9 >
    \varphi(h) + v)))
    }
    \tag{Cons}
    \\
    &\outline{ \exists \inSet{\varphi_1} \ldotp \exists v_1 \ldotp v_1 \leq 9 \land (\exists \inSet{\varphi_2} \ldotp
    \exists v_2 \ldotp
    v_2 \leq 9 \land
    (\forall \inSet{\varphi} \ldotp
    \forall v \ldotp
    v \leq 9
    \Rightarrow
    ((\varphi(h) \neq \varphi_1(h))
    \lor (\varphi(h) + v \neq \varphi_2(h) + v_2))))
    }
    \tag{Cons}
    \\
    &\chavoc{y} \cseq \\
    &\outline{ \exists \inSet{\varphi_1} \ldotp \varphi_1(y) \leq 9 \land (\exists \inSet{\varphi_2} \ldotp
    \varphi_2(y) \leq 9 \land
    (\forall \inSet{\varphi} \ldotp
    \varphi(y) \leq 9
    \Rightarrow
    (\varphi(h) \neq \varphi_1(h)
    \lor \varphi(h) + \varphi(y) \neq \varphi_2(h) + \varphi_2(y))))
    }
    \tag{HavocS}
    \\
    &\cassume{y \leq 9} \cseq \\
    &\outline{ \exists \inSet{\varphi_1}, \inSet{\varphi_2} \ldotp
    \forall \inSet{\varphi} \ldotp
    \varphi(h) \neq \varphi_1(h)
    \lor \varphi(h) + \varphi(y) \neq \varphi_2(h) + \varphi_2(y)
    }
    \tag{AssumeS}
    \\
    &\cassign{l}{h + y} \\
    &\outline{ \exists \inSet{\varphi_1}, \inSet{\varphi_2} \ldotp
    \forall \inSet{\varphi} \ldotp
    \varphi(h) \neq \varphi_1(h)
    \lor \varphi(l) \neq \varphi_2(l)
    }
    \tag{AssignS}
    \end{align*}
    \caption{Proof outline showing that the program \emph{violates} generalized non-interference.
    The rules used at each step of the derivation are shown on the right (the use of rule \namerule{Seq} is implicit).}
    \label{fig:GNIviolated}
\end{figure}

\paragraph{Example}
We now illustrate the use of our three syntactic rules for atomic statements in \figref{fig:GNIviolated},
to prove that the program 
$C_4 \triangleq (\chavoc{y} \cseq \cassume{y \leq 9} \cseq \cassign{l}{h + y})$ from \secref{subsec:k-safety}
violates GNI\@.
This program leaks information about the secret $h$ through its public output $l$ because the pad it uses (variable $y$) is upper bounded.
From the output $l$, we can derive a lower bound for the secret value of $h$, namely $h \geq l - 9$.

To see why $C_4$ violates GNI, consider two executions with different secret values for $h$, and where the execution for the larger secret value sets $y$ to exactly $9$. This execution will produce a larger public output $l$ (since the other execution adds at most $9$ to its smaller secret). Hence, these executions can be \emph{distinguished} by their public outputs.

Our proof outline in \figref{fig:GNIviolated} captures this intuitive reasoning in a natural way.
We start with the postcondition that corresponds to the negation of GNI,
and work our way backward, by successively applying our syntactic rules \namerule{AssignS}, \namerule{AssumeS}, and \namerule{HavocS}.
We conclude using the rule \namerule{Cons}:
Since the precondition implies the existence of two states with different values for $h$,
we first instantiate $\varphi_1$ and $\varphi_2$ such that $\varphi_1$ and $\varphi_2$ are both members of the set of initial states,
and $\varphi_2(h) > \varphi_1(h)$.\footnote{Note
that the quantified states $\varphi_1$, $\varphi_2$ and $\varphi$ from different hyper-assertions
can be unrelated.
That is, the witnesses for $\varphi_1$ and $\varphi_2$ in the first hyper-assertion $\color{royalblue} [\exists \inSet{\varphi_1}, \inSet{\varphi_2} \ldotp \varphi_1(h) \neq \varphi_2(h)]$
are not necessarily the same as the ones in the second hyper-assertion $\color{royalblue} [\exists \inSet{\varphi_1} \ldotp \exists \inSet{\varphi_2} \ldotp
\varphi_2(h) > \varphi_1(h)]$, which is why the entailment holds.}
We then instantiate $v_2 = 9$, such that, for any $v \leq 9$,
$\varphi_2(h) + v_2 > \varphi(h) + v$,
which concludes the proof.

%% file: loops.tex
To reason about standard while loops,
we can derive from the core rule \namerule{Iter} in \figref{fig:rules}
the rule \namerule{WhileDesugared}, shown in \figref{fig:loop-rules}
(recall that $\cwhile{b}{C} \triangleq \cnwhile{(\cassume{b} \cseq C)} \cseq \cassume{\lnot b}$).
While this derived rule is expressive, it has two main drawbacks for usability:
(1)~Because of the use of the infinitary $\otimesn$, it requires non-trivial \emph{semantic} reasoning (via the consequence rule),
and
(2)~the invariant $I_n$ relates only the executions that perform \emph{at least} $n$ iterations, but ignores executions that perform fewer.

To illustrate problem (2), imagine that we want to prove that the hyper-assertion
$\low{l} \triangleq (\forall \inSet{\varphi} \ldotp \forall \inSet{\varphi'} \ldotp \varphi(l) = \varphi'(l))$ holds after a while loop.
A natural choice for our loop invariant $I_n$ would be $I_n \triangleq \low{l}$ (independent of $n$).
However, this invariant does \emph{not} entail our desired postcondition $\low{l}$.
Indeed, $\otimesn \low{l}$ holds for a set of states iff it is a \emph{union of} sets of states that all \emph{individually} satisfy $\low{l}$. This property holds trivially in our example (simply choose one set per possible value of $l$) and, in particular, does not express that the entire set of states after the loop satisfies $\low{l}$.
Note that this does not contradict completeness (\thmref{thm:complete}),
but simply means that a stronger invariant $I_n$ is needed.

In this section, we thus present three more convenient loop rules, shown in \figref{fig:loop-rules},
which capture powerful reasoning principles, and overcome those limitations:
The rule \namerule{WhileSync} (\secref{subsec:sync}) is the easiest to use,
and can be applied whenever all executions of the loop have the same control flow.
Two additional rules for while loops can be applied whenever the
control flow differs.
The rule \namerule{While-}$\forall^* \exists^*$ (\secref{subsec:forall-exists})
supports $\forall^* \exists^*$ postconditions,
while the rule \namerule{While-}$\exists$ (\secref{subsec:proving-exists}) handles postconditions with a top-level existential quantifier.
In our experience, these loop rules cover all practical hyper-assertions that can be expressed in our syntax. We are not aware of any practical program hyperproperty that requires multiple quantifier alternations.

\begin{figure*}[t]
\begin{center}
    \scriptsize
    \[
    \begin{array}{c}

    \Inf[\mathit{WhileDesugared}]{\shoare{I_n}{\cassume{b} \cseq C}{I_{n+1}}}
    {\shoare{ \otimesn I_n }{\cassume{\lnot b}}{Q}}
    {\shoare{I_0}{\cwhile{b}{C}}{Q}}

    \\[2em]

    \Inf[\mathit{WhileSync}]{I \models \low{b}}{\shoare{I \land \square b}{C}{I}}{
        \shoare{I}{\cwhile{b}{C}}{ (I \lor \emp{}) \land \square (\lnot b) }}

    \hspace{3mm}

    \Inf[\mathit{IfSync}]{P \models \low{b}}{\shoare{P \land \square b}{C_1}{Q}}{\shoare{P \land \square (\lnot b)}{C_2}{Q}}{\shoare{P}{\cif{b}{C_1}{C_2}}{Q}}

    \\[2em]

    \Inf[\mathit{While-}\forall^* \exists^*]
    {\shoare{I}{\cifonly{b}{C}}{I}}
    {\shoare{I}{\cassume{\lnot b}}{Q}}
    { \text{no } \forall \inSet{\_} \text{ after any } \exists \text{ in } Q}
    {\shoare{I}{\cwhile{b}{C}}{Q}}

    \\[2em]

    \Inf[\mathit{While-}\exists]
    {\forall v \ldotp
        \shoare{\exists \inSet{\varphi} \ldotp P_\varphi \land \expApplied{b}{\varphi} \land v = \expApplied{e}{\varphi}}
                { \cifonly{b}{C} }
                { \exists \inSet{\varphi} \ldotp P_\varphi \land \expApplied{e}{\varphi} \prec v }}
    {\forall \varphi \ldotp \shoare{P_\varphi}{ \cwhile{b}{C} }{ Q_\varphi }}
    {\prec \text{wf}}
    {\shoare{\exists \inSet{\varphi} \ldotp P_\varphi }{ \cwhile{b}{C} }{ \exists \inSet{\varphi} \ldotp Q_\varphi }}

    \end{array}
    \]
\end{center}

\caption{\logic{} rules for while loops (and branching).
Recall that $\low{b} \triangleq (\forall \inSet{\varphi}, \inSet{\varphi'} \ldotp \expApplied{b}{\varphi} = \expApplied{b}{\varphi'})$
and $\square b \triangleq ( \forall \inSet{\varphi} \ldotp b(\varphi) )$.
In the rule \namerule{WhileSync}, $\prec$ must be \emph{well-founded} ($\text{wf}$).
}
\label{fig:loop-rules}
\end{figure*}

\subsection{Synchronized Control Flow}\label{subsec:sync}

Standard loop invariants are sound in relational logics if all executions exit the loop \thibault{\emph{simultaneously}~\cite{Benton04,terauchi2005secure}}.
In our logic, this synchronized control flow can be enforced by requiring that the loop guard $b$ has the same value in all states (1)~before the loop and (2)~after every loop iteration,
as shown by the rule \namerule{WhileSync} in \figref{fig:loop-rules}.
After the loop, we get to assume $(I \lor \emp{}) \land \always (\lnot b)$.
That is, the loop guard $b$ is false in all executions, and the invariant $I$ holds, or the set of states is empty.
The $\emp{}$ disjunct corresponds to the case where the loop does not terminate (\ie \emph{no} execution terminates).
Going back to our motivating example, the natural invariant $I \triangleq \low{l}$ with the rule \namerule{WhileSync} is now sufficient for our example, since we get the postcondition $(\low{l} \lor \emp{}) \land \always (\lnot b)$, which implies our desired (universally-quantified) postcondition $\low{l}$.
In the case where the desired postcondition quantifies existentially over states at the top-level, it is necessary to prove that
\thibault{at least one execution terminates}.
We show the corresponding rule in \ifextended{\appref{app:termination}}{our extended version~\cite{hhlExtended}}.

We also provide a rule for if statements with synchronized control flow (rule \namerule{IfSync} in \figref{fig:loop-rules}),
which can be applied when all executions take the same branch.
This rule is simpler to apply than the core rule \namerule{Choice}, since it avoids the $\otimes$ operator,
which usually requires semantic reasoning.

\begin{figure}
    \tiny
   \begin{align*}
    &\outline{\forall \inSet{\varphi_1}, \inSet{\varphi_2} \ldotp \mathit{len}(\varphi_1(h)) = \mathit{len}(\varphi_2(h)) } \\
    &\outline{ \forall \inSet{\varphi_1}, \inSet{\varphi_2} \ldotp
    0 = 0 \land \mathit{len}(\varphi_1(h)) = \mathit{len}(\varphi_2(h))
    \land (\exists \inSet{\varphi} \ldotp \varphi(h) = \varphi_1(h) \land [] = []) } \tag{Cons} \\
    &\cassign{s}{0} \\
    &\cassign{l}{[]} \\
    &\cassign{i}{0} \\
    &\outline{ \forall \inSet{\varphi_1}, \inSet{\varphi_2} \ldotp
    \varphi_1(i) = \varphi_2(i) \land \mathit{len}(\varphi_1(h)) = \mathit{len}(\varphi_2(h))
    \land (\exists \inSet{\varphi} \ldotp \varphi(h) = \varphi_1(h) \land \varphi(l) = \varphi_2(l)) } \tag{AssignS} \\
    &\mathbf{while} \; (i < \mathit{len}(h)) \; \{ \\
    &\quad \outline{ (\forall \inSet{\varphi_1}, \inSet{\varphi_2} \ldotp
    \varphi_1(i) = \varphi_2(i) \land \mathit{len}(\varphi_1(h)) = \mathit{len}(\varphi_2(h))
    \land (\exists \inSet{\varphi} \ldotp \varphi(h) = \varphi_1(h) \land \varphi(l) = \varphi_2(l))) \land \always(i < \mathit{len}(h)) } \\
    &\quad \outlineStart{ \forall \inSet{\varphi_1} \ldotp \forall v_1 \ldotp \forall \inSet{\varphi_2} \ldotp \forall v_2 \ldotp
    \varphi_1(i) + 1 = \varphi_2(i) + 1 \land \mathit{len}(\varphi_1(h)) = \mathit{len}(\varphi_2(h))
    \land
    } \\ &\quad \outlineEnd{
     (\exists \inSet{\varphi} \ldotp \exists v \ldotp \varphi(h) = \varphi_1(h) \land \varphi(l) \concat [ (\varphi(s) + \varphi(h)[\varphi(i)]) \oplus v ] = \varphi_2(l)
        \concat [ (\varphi_2(s) + \varphi_2(h)[\varphi_2(i)])  \oplus v_2 ]
    ) } \tag{Cons} \\
    &\quad \cassign{s}{s + h[i]} \cseq \\
    &\quad \chavoc{k} \cseq \\
    &\quad \cassign{l}{l \concat [s \oplus k]} \cseq \\
    &\quad \cassign{i}{i + 1} \cseq \\
    &\quad \outline{ \forall \inSet{\varphi_1}, \inSet{\varphi_2} \ldotp
    \varphi_1(i) = \varphi_2(i) \land \mathit{len}(\varphi_1(h)) = \mathit{len}(\varphi_2(h))
    \land (\exists \inSet{\varphi} \ldotp \varphi(h) = \varphi_1(h) \land \varphi(l) = \varphi_2(l)) } \tag{HavocS, AssignS} \\
    &\} \\
    &\outline{((
        \forall \inSet{\varphi_1}, \inSet{\varphi_2} \ldotp
        \varphi_1(i) = \varphi_2(i) \land \mathit{len}(\varphi_1(h)) = \mathit{len}(\varphi_2(h))
        \land (\exists \inSet{\varphi} \ldotp \varphi(h) = \varphi_1(h) \land \varphi(l) = \varphi_2(l))
    )
         \lor \emp{}) \land \always(i \geq \mathit{len}(h))} \tag{WhileSync} \\
    &\outline{ \forall \inSet{\varphi_1}, \inSet{\varphi_2} \ldotp
    \exists \inSet{\varphi} \ldotp \varphi(h) = \varphi_1(h) \land \varphi(l) = \varphi_2(l) } \tag{Cons}
    \end{align*}
    \caption{A proof that the program in black satisfies generalized non-interference (where the elements of list $h$ are secret, but its length is public), using the rule \namerule{WhileSync}.
    $[]$ represents the empty list, $\concat$ represents list concatenation, $h[i]$ represents the i-th element of list $h$, and $\oplus$ represents the XOR operator.}
    \label{fig:provingGNIloop}
\end{figure}

\paragraph{Example}
The program in \figref{fig:provingGNIloop}
takes as input a list $h$ of secret values (but whose length is public),
computes its prefix sum $[h[0], h[0] + h[1], \ldots]$, and encrypts the result by performing a one-time pad on each element of this prefix sum,
resulting in the output $[h[0] \oplus k_0, (h[0] + h[1]) \oplus k_1, \ldots]$.
The keys $k_0, k_1, \ldots$ are chosen non-deterministically at each iteration, via the variable $k$.\footnote{In practice,
the keys used in this program should be stored somewhere, so that one is later able to decrypt the output.}

Our goal is to prove that the encrypted output $l$ does not leak information about the secret elements of $h$,
provided that the attacker does not have any information about the non-deterministically chosen keys. 
We achieve this by formally proving that this program satisfies GNI\@.
Since the length of the list $h$ is public, we start with the precondition $\forall \inSet{\varphi_1}, \inSet{\varphi_2} \ldotp \mathit{len}(\varphi_1(h)) = \mathit{len}(\varphi_2(h))$.
This implies that all our executions will perform the same number of loop iterations.
Thus, we use the rule \namerule{WhileSync},
with the natural loop invariant $I \triangleq
    (\forall \inSet{\varphi_1}, \inSet{\varphi_2} \ldotp
    \varphi_1(i) = \varphi_2(i) \land \mathit{len}(\varphi_1(h)) = \mathit{len}(\varphi_2(h))
    \land (\exists \inSet{\varphi} \ldotp \varphi(h) = \varphi_1(h) \land \varphi(l) = \varphi_2(l)))$.
The last conjunct corresponds to the postcondition we want to prove, while the former entails $\low{i < \mathit{len}(h)}$,
as required by the rule \namerule{WhileSync}.

The proof of the loop body starts at the end with the loop invariant $I$, and works backward, using the syntactic rules \namerule{HavocS} and \namerule{AssignS}.
From $I \land \always(i < \mathit{len}(h))$,
we have to prove that there exists a value $v$ such that
$\varphi(l) \concat [ (\varphi(s) + \varphi(h)[\varphi(i)]) \oplus v ] = \varphi_2(l)
\concat [ (\varphi_2(s) + \varphi_2(h)[\varphi_2(i)])  \oplus v_2 ]$.
Since $\varphi(l) = \varphi_2(l)$,
this boils down to proving that
$(\varphi(s) + \varphi(h)[\varphi(i)]) \oplus v = (\varphi_2(s) + \varphi_2(h)[\varphi_2(i)]) \oplus v_2$,
which we achieve by choosing
$v \triangleq (\varphi_2(s) + \varphi_2(h)[\varphi_2(i)]) \oplus v_2 \oplus (\varphi(s) + \varphi(h)[\varphi(i)])$.

\subsection{$\forall^* \exists^*$-Hyperproperties}\label{subsec:forall-exists}

Let us now turn to the more general case, where different executions might exit the loop at different iterations.
As explained at the start of this section, the main usability issue of the rule \namerule{WhileDesugared}
is the precondition $\otimesn I_n$ in the second premise, which requires non-trivial semantic reasoning.
The $\otimesn$ operator is required, because $I_n$ ignores executions that exited the loop earlier;
it relates only the executions that have performed \emph{at least} $n$ iterations.
In particular, it would be unsound to replace the precondition $\otimesn I_n$ by $\exists n \ldotp I_n$.

\setlength\intextsep{0pt}
\begin{wrapfigure}{r}{0.23\textwidth}
    \begin{minipage}{0.23\textwidth}
    \begin{align*}
    &\cassign{a}{0} \cseq \sline
    &\cassign{b}{1} \cseq \sline
    &\cassign{i}{0} \cseq \sline
    &\mathbf{while} \; (i < n) \; \{ \sline
    &\quad \cassign{tmp}{b} \cseq \sline
    &\quad \cassign{b}{a+b} \cseq \sline
    &\quad \cassign{a}{tmp} \cseq \sline
    &\quad \cassign{i}{i + 1} \sline
    &\}
    \end{align*}
    \end{minipage}
    \caption{The program $C_{\mathit{fib}}$, which computes the $n$-th Fibonacci number.}
    \label{fig:fib}
\end{wrapfigure}

The rule \namerule{While-}$\forall^* \exists^*$ in \figref{fig:loop-rules} solves this problem for the general case of $\forall^* \exists^*$-postconditions.
The key insight is to reason about the successive \emph{unrollings} of the while loop:
The rule requires to prove an invariant $I$ for the conditional statement $\cifonly{b}{C}$, in contrast to $\cassume{b} \cseq C$ in the rule \namerule{WhileDesugared}.
This allows the invariant $I$ to refer to \emph{all} executions,
\ie executions that are still running the loop (which will execute $C$), and executions that have already exited the loop (which will not execute $C$).

\paragraph{Example}
The program $C_{\mathit{fib}}$ in \figref{fig:fib} takes as input an integer $n \geq 0$
and computes the $n$-th Fibonacci number (in variable $a$).
We want to prove that $C_{\mathit{fib}}$ is monotonic, \ie that the $n$-th Fibonacci number is greater than or equal to the $m$-th Fibonacci number whenever $n \geq m$,
without making explicit what $C_{\mathit{fib}}$ computes.
Formally, we want to prove the hyper-triple\\
$\simpleHoare{ \forall \inSet{\varphi_1}{,}\inSet{\varphi_2} \ldotp \varphi_1(t) \seq 1 \sand \varphi_2(t) \seq 2 {\Rightarrow} \varphi_1(n) \sgeq \varphi_2(n) }{C_{\mathit{fib}}}{ \forall \inSet{\varphi_1}{,}\inSet{\varphi_2} \ldotp \varphi_1(t) \seq 1 \sand \varphi_2(t) \seq 2 {\Rightarrow} \varphi_1(a) \sgeq \varphi_2(a) }$,
where $t$ is a logical variable used to track the execution (as explained in \secref{subsec:k-safety}).

Intuitively, this program is monotonic because
both executions will perform at least $\varphi_2(n)$ iterations, during which they will have the same values for $a$ and $b$.
The first execution will then perform $\varphi_1(n) - \varphi_2(n)$ additional iterations, during which $a$ and $b$ will increase,
thus resulting in larger values for $a$ and $b$.%

We cannot use the rule \namerule{WhileSync} to make this intuitive argument formal, since both executions might perform a different number of iterations.
Moreover, we cannot express this intuitive argument with the rule \namerule{WhileDesugared} either,
since the invariant $I_k$ only relates executions that perform \emph{at least $k$ iterations}, as explained earlier:
After the first $\varphi_2(n)$ iterations, the loop invariant $I_k$ cannot refer to the values of $a$ and $b$ in the second execution,
since this execution has already exited the loop.

However, we can use the rule \namerule{While-}$\forall^* \exists^*$ to prove that $C_{\mathit{fib}}$ is monotonic,
with the intuitive loop invariant
$I \triangleq (\forall \inSet{\varphi_1} {,} \inSet{\varphi_2} \ldotp \varphi_1(t) \seq 1 \sand \varphi_2(t) \seq 2$ \\ $\Rightarrow
    (\varphi_1(n) \sminus \varphi_1(i) \sgeq \varphi_2(n) \sminus \varphi_2(i) \sand \varphi_1(a) \sgeq \varphi_2(a) %
    \sand \varphi_1(b) \sgeq \varphi_2(b) %
    )
    \sand \always( b \sgeq a \sgeq 0 ))$.
The first part captures the relation between the two executions:
$a$ and $b$ are larger in the first execution than in the second one, and the first execution does at least as many iterations as the second one.
The second part $\always( b \geq a \geq 0 )$ is needed to prove that the additional iterations lead to larger values for $a$ and $b$.
The proof of this example is in
\ifextended{the appendix~(\appref{app:fibonacci}).}{%
our extended version~\cite{hhlExtended}.}

\paragraph{Restriction to $\forall^* \exists^*$-hyperproperties.}
The rule \namerule{While-}$\forall^* \exists^*$ is quite general and powerful,
since it can be applied to prove any postcondition of the shape $\forall^* \exists^*$,
which includes \emph{all} safety hyperproperties, as well as \thibault{some} liveness hyperproperties such as GNI\@.
However, it cannot be applied for postconditions with a top-level existential quantification over states, because this would be unsound.
Indeed, a triple such as $\shoare{\exists \inSet{\varphi} \ldotp \forall \inSet{\varphi'} \ldotp I}{ \cifonly{b}{C} }{\exists \inSet{\varphi} \ldotp \forall \inSet{\varphi'} \ldotp I}$
implies that, for any $n$, there exists a state $\varphi$ such that $I$ holds for all states $\varphi'$ reached after \emph{unrolling the loop $n$ times}.
The key issue is that $\varphi$ might not be a valid witness for states $\varphi'$ reached after \emph{more than $n$ loop unrollings},
and therefore we might have different witnesses for $\varphi$ for each value of $n$.
We thus have no guarantee that there is a \emph{global} witness that works for all states $\varphi'$ after any \emph{number} of loop unrollings. To handle such examples, we present a rule for
$\exists^* \forall^*$-hyperproperties
next.

\subsection{$\exists^* \forall^*$-Hyperproperties}
\label{subsec:proving-exists}

\begin{wrapfigure}{r}{0.22\textwidth}
    \begin{minipage}{0.22\textwidth}
    \begin{align*}
    &\cassign{x}{0} \cseq \\[-5pt]
    &\cassign{y}{0} \cseq \\[-5pt]
    &\cassign{i}{0} \cseq \\[-5pt]
    &\mathbf{while} \; (i < k) \; \{ \\[-5pt]
    &\quad \chavoc{r} \cseq \\[-5pt]
    &\quad \cassume{r \geq 2} \cseq \\[-5pt]
    &\quad \cassign{t}{x} \cseq \\[-5pt]
    &\quad \cassign{x}{2*x + r} \cseq \\[-5pt]
    &\quad \cassign{y}{y + t * r} \cseq \\[-5pt]
    &\quad \cassign{i}{i + 1} \\[-5pt]
    &\}
    \end{align*}
    \end{minipage}
    \caption{A program with a final state with minimal values for $x$ and $y$.}
    \label{fig:minimum}
    \vspace{1mm}
\end{wrapfigure}

The rule \namerule{While-}$\forall^* \exists^*$ can be applied for any postcondition of the form $\forall^* \exists^*$,
which includes all safety hyperproperties as well as \thibault{some} liveness hyperproperties such as GNI,
but cannot be applied to prove postconditions with a top-level existential quantifier,
such as postconditions of the shape $\exists^* \forall^*$ (\eg to prove the existence of minimal executions, or to prove that a $\forall^* \exists^*$-hyperproperty is violated).
In this case, we can apply
the rule \namerule{While-}$\exists$ in \figref{fig:loop-rules}.
To the best of our knowledge, this is the first program logic rule that can deal with $\exists^* \forall^*$-hyperproperties for loops.
This rule splits the reasoning into two parts:
First, we prove that there is a \emph{terminating} state $\varphi$ such that the hyper-assertion $P_\varphi$ holds after some number of loop unrollings.
This is achieved via the first premise of the rule, which requires a well-founded relation $\prec$, and a variant $e(\varphi)$ that strictly decreases at each iteration,
until $b(\varphi)$ becomes false and $\varphi$ exits the loop.\footnote{Note
that the existentially-quantified state $\varphi$ in the postcondition of the first premise of the rule \namerule{While-}$\exists$
does \emph{not} have to be from the same execution as the one in the precondition.}
In a second step, we fix the state $\varphi$ (since it has exited the loop),
which corresponds to our global witness,
and prove $\shoare{P_\varphi}{\cwhile{b}{C}}{Q_\varphi}$ using any loop rule.
For example, if $P_\varphi$ has another top-level existential quantifier, we can apply the rule \namerule{While-}$\exists$ once more;
if $P_\varphi$ is a $\forall^* \exists^*$-hyper-assertion,
we can apply the rule \namerule{While-}$\forall^* \exists^*$.

As an example, consider proving that the program $C_m$ in \figref{fig:minimum} has a final state with a minimal value for $x$ and $y$.
Formally, we want to prove the triple\\
$\simpleHoare{\lnot \emp{} \land \always(k \geq 0)}{C_m}{ \exists \inSet{\varphi} \ldotp \forall \inSet{\alpha} \ldotp \varphi(x) \leq \alpha(x) \land \varphi(y) \leq \alpha(y) }$.
Since the set of initial states is not empty and $k$ is always non-negative, we know that
there is an initial state with a minimal value for $k$.
We prove that this state leads to a final state with minimal values for $x$ and $y$, using the rule \namerule{While-}$\exists$.
For the first premise,
we choose the variant\footnote{We interpret $\prec$ as $<$ between natural numbers, \ie $a \prec b$ iff $0 \leq a$ and $a < b$, which is well-founded.} $k - i$,
and the invariant
$P_\varphi \triangleq
(
    \forall \inSet{\alpha} \ldotp 0 \leq \getvar{\varphi}{x} \leq \getvar{\alpha}{x} \land 0 \leq \getvar{\varphi}{y} \leq \getvar{\alpha}{y} \land \getvar{\varphi}{k} \leq \getvar{\alpha}{k} \land \getvar{\varphi}{i} = \getvar{\alpha}{i}
)$,
capturing both that $\varphi$ has minimal values for $x$ and $y$, but also that $\varphi$ will be the first state to exit the loop.
We prove that this is indeed an invariant for the loop, by choosing $r = 2$ for the non-deterministic assignment for $\varphi$.
Finally, we prove the second premise with $Q_\varphi \triangleq (\forall \inSet{\alpha} \ldotp 0 \leq \varphi(x) \leq \alpha(x) \land 0 \leq \varphi(y) \leq \alpha(y))$
and the rule \namerule{While-}$\forall^* \exists^*$.
The proof of this example is in
\ifextended{the appendix~(\appref{app:minimum}).}{%
our extended version.}

%% file: related_work.tex
\mypar{Overapproximate (relational) Hoare logics}
Hoare Logic originated with the seminal works of \citet{FloydLogic} and \citet{HoareLogic},
with the goal of proving programs functionally correct.
Relational Hoare Logic~\cite{Benton04} (RHL) extends Hoare Logic to reason about (2-safety) hyperproperties of a single program as well as properties relating the executions of two different programs (\eg semantic equivalence). RHL's ability to relate the executions of two different programs is also useful in the context of proving $2$-safety hyperproperties of a single program,
in particular, when the two executions take different branches of a conditional statement. In comparison, \logic{} can prove and disprove hyperproperties of a single program (\secref{subsec:expressivity}), but requires a program transformation to express relational properties
(see \ifextended{end of \appref{subsec:beyond}}{our extended version~\cite{hhlExtended}}).
Extending \logic{} to multiple programs is interesting future work.

RHL has been extended in many ways, for example to
deal with heap-manipulating~\cite{RelationalSL}
and higher-order~\cite{aguirre2017relational}
programs.
A family of Hoare and separation logics~\cite{AmtoftSIF,Costanzo2014,SecCSL,CommCSL} designed to prove non-interference~\cite{volpano1996nonInterference} specializes RHL by considering triples with
a single
program, similar to \logic{}.
\citet{SurveyRHL} provides an overview of the principles underlying relational Hoare logics.
Cartesian Hoare Logic~\cite{CHL16} (CHL) extends RHL to reason about 
hyperproperties of $k$ executions, with a focus on automation and scalability.
CHL has recently been reframed~\cite{HypersafetyCompositionally}
as a weakest-precondition calculus,
increasing its support for proof compositionality.
\logic{} can express the properties supported by CHL, in addition to many other properties; automating \logic{} is future work.

\mypar{Underapproximate program logics}
Reverse Hoare Logic~\cite{ReverseHL} is an underapproximate variant of Hoare Logic, designed to prove the existence of good executions.
The recent Incorrectness Logic~\cite{IncorrectnessLogic} adapts this idea to prove the presence of bugs.
Incorrectness Logic has been extended with concepts from separation logic to reason about heap-manipulating sequential~\cite{ISL} and concurrent~\cite{CISL} programs.
It has also been extended to prove the presence of insecurity in a program (\ie to disprove $2$-safety hyperproperties)~\cite{InsecurityLogic}.
Underapproximate logics have been successfully used as foundation of industrial bug-finding tools~\cite{Blackshear2018,Gorogiannis2019,Distefano2019,Le2022}.
\logic{} enables proving and disproving hyperproperties within the same logic.

Several recent works have proposed approaches to unify over- and underapproximate reasoning.
Exact Separation Logic~\cite{maksimovic2022exact}
can establish both overapproximate and (backward) underapproximate properties over single executions of heap-manipulating programs,
by employing triples that describe \emph{exactly} the set of reachable states.
Local Completeness Logic~\cite{LCL-logic,correctnessIncorrectnessAI}
unifies over- and underapproximate reasoning in the context of abstract interpretation,
by building on Incorrectness Logic, and enforcing a notion of \emph{local completeness} (no false alarm should be produced relative to some fixed input).
HL and IL have been both embedded in a Kleene algebra with diamond operators and countable joins of tests~\cite{UnifyingHLandIL}.
Dynamic Logic~\cite{harel1979first} is an extension of modal logic that can express both overapproximate and underapproximate guarantees over single executions of a program.

Outcome Logic~\cite{OutcomeLogic} (OL) unifies overapproximate and (forward) underapproximate reasoning for heap-manipulating and probabilistic programs,
by combining and generalizing the standard overapproximate Hoare triples with forward underapproximate triples\ifextended{ (see \appref{subsec:liveness})}{}.
OL (instantiated to the powerset monad) uses a semantic model similar to our extended semantics (\defref{def:extended_semantics}),
and a similar definition for triples (\defref{def:hoare_triples}).
Moreover, a theorem similar to our \thmref{thm:disproving} holds in OL, \ie invalid OL triples can be disproven within OL\@.
The key difference with \logic{} is that OL does not support reasoning about hyperproperties.
OL assertions are composed of atomic unary assertions, which can express the existence and the absence of certain states,
but not relate states with each other, which is key to expressing hyperproperties.
OL does not provide logical variables, on which we rely to express certain hyperproperties (see \secref{subsec:k-safety}).

\mypar{Logics for $\forall^* \exists^*$-hyperproperties}
\citet{next700RHL} present a general framework for defining relational program logics for arbitrary monadic effects (such as state, input-output, nondeterminism, and discrete probabilities),
for two executions of two (potentially different) programs.
Their key idea is to map \emph{pairs} of (monadic) computations to relational specifications, using relational \emph{effect observations}.
In particular, they discuss instantiations for $\forall \forall$-, $\forall \exists$-, and $\exists \exists$-hyperproperties.
RHLE~\cite{RHLE} supports overapproximate and (a limited form of) underapproximate reasoning,
as it can establish  $\forall^* \exists^*$-hyperproperties, such as generalized non-interference (\secref{subsec:beyond-k-safety}) and program refinement.
\thibault{BiKAT~\cite{BiKat},
an algebra of alignment for relational verification, can be used directly to prove $\forall \forall$-properties.
Moreover, $\forall \exists$-properties between two programs $C_1$ and $C_2$ can also be proved with BiKAT,
by proving that a corresponding $\forall \forall$-property holds for some \emph{alignment witness},
\ie a program that overapproximates the behavior of $C_1$
while underapproximating the behavior of $C_2$.}
\thibault{All three frameworks can be used to}
reason about relational properties of multiple programs, whereas \logic{} requires a program transformation to handle such properties.
On the other hand, our logic supports a wider range of underapproximate reasoning and can express properties not handled by any of them,
\eg $\exists^* \forall^*$-hyperproperties \thibault{and hyperproperties relating an unbounded or infinite number of executions}.
Moreover, even for $\forall^* \exists^*$-hyperproperties, \logic{}
provides while loop rules that have no equivalent in these logics, such as the rules
\namerule{While-}$\exists$ (useful in this context for $\exists^*$-hyperproperties)
and
\namerule{While-}$\forall^* \exists^*$~(\secref{sec:loops}). 

\thibault{Note that one can in principle use BiKAT to prove $\exists \forall$-properties,
by essentially proving the negation of $\forall \exists$-properties:
To prove that an $\exists \forall$-property between two programs $C_1$ and $C_2$ holds,
one needs to consider all programs $W$ that overapproximate the behavior of $C_1$ and underapproximate the behavior of $C_2$,
and prove that $W$ does \emph{not} satisfy a $\forall \forall$-property.}

\mypar{Probabilistic Hoare logics}
Many assertion-based logics for probabilistic programs have been proposed~\cite{Ramshaw79, probaHL,corin2006probabilistic,Ellora,probabilisticSL, Rand2015}.
These logics typically employ assertions over \emph{probability (sub-)distributions} of states,
which bear some similarities to hyper-assertions:
Asserting the existence (resp.\ absence) of a state is analogous to asserting that the probability of this state is strictly positive (resp.\ zero).
\thibault{Taking the union of two sets of states is analogous to taking the sum of two sub-distributions.
Our operator $\otimes$ (\defref{def:otimes}) used in the rule \namerule{Choice}
is thus similar to the operator $\oplus$ from~\citet{Ellora}}.
Notably, our loop rule \namerule{While}-$\forall^* \exists^*$ draws some inspiration from the rule \namerule{While} of~\citet{Ellora},
which also requires an invariant that holds for all \emph{unrollings} of the loop.
These probabilistic logics have also been extended to the relational setting~\cite{barthe2009formal},
for instance to reason about the equivalence of probabilistic programs.

\mypar{Verification of hyperproperties}
The concept of hyperproperties has been formalized by \citet{hyperproperties}.
Verifying that a program satisfies a $k$-safety hyperproperty can be reduced to verifying a
\thibault{safety}
property of
the \emph{self-composition} of the program~\cite{selfComposition,terauchi2005secure}
(\eg by sequentially composing the program with renamed copies of itself).
Self-composition has been generalized to product programs~\cite{productPrograms,eilers2019modular}.
(Extensions of) product programs have also been used to verify relational properties such as program refinement~\cite{barthe2013beyond}
and probabilistic relational properties such as differential privacy~\cite{diffPrivacyHL}.
The temporal logics LTL, CTL, and CTL*, have been extended to HyperLTL and HyperCTL*~\cite{clarkson2014temporal}
to specify hyperproperties,
and model-checking algorithms~\cite{coenen2019verifying, Hsu21,beyondkSafety,AutoHyperQ} have been proposed to verify hyperproperties expressed in these logics,
including hyperproperties outside the safety class.
\citet{Unno2021} propose an approach to automate relational verification (\thibault{including $\forall^* \exists^*$-properties such as GNI}) based on an extension of constrained Horn-clauses.
Relational properties of imperative programs can be verified by reducing them to validity problems in trace logic~\cite{relationalUsingTraceLogic}.
Finally, the notion of hypercollecting semantics~\cite{assaf2017hypercollecting} (similar to our extended semantics) has been proposed to statically analyze information flow using abstract interpretation~\cite{cousot1977abstract}.
\thibault{One major difference between our extended semantics and
this hypercollecting semantics is the treatment of loops.
The former is defined directly on top of the big-step semantics (\defref{def:extended_semantics}), whereas the latter is defined inductively, and, in the case of loops, as a fixpoint over sets of sets of traces,
which is more suitable for abstract interpretation, but less precise than the extended semantics.
This difference in precision matters
for hyperproperties that are not subset-closed (such as GNI)~\cite{PasquaPhD2019,Naumann2019}.
}

%% file: conclusion.tex
We have presented \logic{}, a novel, sound, and complete program logic that supports reasoning about a wide range of hyperproperties.
It is based on a simple but powerful idea: reasoning directly about the \emph{set} of states at a given program point, instead of a fixed number of states.
We have demonstrated that \logic{} is very expressive: It can be used to prove or disprove \emph{any} program hyperproperty over terminating executions,
including $\exists^* \forall^*$-hyperproperties and hyperproperties relating an unbounded or infinite number of executions,
which goes beyond the properties handled by existing Hoare logics.
Moreover, we have presented syntactic rules, compositionality rules, and rules for loops that capture important proof principles naturally.

We believe that \logic{} is a powerful foundation for reasoning about the correctness and incorrectness of program hyperproperties. We plan to build on this foundation in our future work.
First, we will explore automation for \logic{} by developing an encoding into an SMT-based verification system such as Boogie~\cite{Boogie}. 
Second, we will extend the language supported by the logic, in particular, to include a heap. The main challenge will be to adapt concepts from separation logic to hyper-assertions,
\eg to find a suitable definition for the separating conjunction of two hyper-assertions.
Third, we will explore an extension of \logic{} that can relate multiple programs.

%% file: appendix/definitions.tex
\begin{figure*}[h]
    \footnotesize
\begin{center}
\[
\begin{array}{c}
\Inf{\bigstep{\cskip}{\sigma}{\sigma}}
\hspace{3mm}
\Inf{\bigstep{\cassign{x}{e}}{\sigma}{\sigma[x \mapsto e(\sigma)]}}
\hspace{3mm}
\Inf{\bigstep{\chavoc{x}}{\sigma}{\sigma[x \mapsto v]}}
\hspace{3mm}
\Inf{\bigstep{C_1}{\sigma}{\sigma'}}{\bigstep{C_2}{\sigma'}{\sigma''}}{\bigstep{C_1 \cseq C_2}{\sigma}{\sigma''}}
\\[2em]
\Inf{\bigstep{C_1}{\sigma}{\sigma'}}{\bigstep{\cnif{C_1}{C_2}}{\sigma}{\sigma'}}
\hspace{3mm}
\Inf{\bigstep{C_2}{\sigma}{\sigma'}}{\bigstep{\cnif{C_1}{C_2}}{\sigma}{\sigma'}}
\hspace{3mm}
\Inf{b(\sigma)}{\bigstep{\cassume{b}}{\sigma}{\sigma}}
\hspace{3mm}
\Inf{\bigstep{C}{\sigma}{\sigma'}}{\bigstep{\cnwhile{C}}{\sigma'}{\sigma''}}{\bigstep{\cnwhile{C}}{\sigma}{\sigma''}}
\hspace{3mm}
\Inf{\bigstep{\cnwhile{C}}{\sigma}{\sigma}}
\end{array}
\]
\end{center}
\caption{Big-step semantics. Since expressions are functions from states to values, $e(\sigma)$ denotes the evaluation of expression $e$ in state $\sigma$. $\sigma[x \mapsto v]$ is the state that yields $v$ for $x$ and the value in $\sigma$ for all other variables.
}
\label{fig:semantics}
\end{figure*}

\begin{definition}\textbf{Evaluation of syntactic hyper-expressions and satisfiability of hyper-assertions.}\\
    \label{def:sat}
    Let $\Sigma$ a mapping from variables (such as $\varphi$) to states, and $\Delta$ a mapping from variables (such as $x$) to values.\footnote{In our Isabelle formalization, these mappings are actually lists, since we use De Bruijn indices~\cite{DeBruijn}.}
    The evaluation of hyper-expressions is defined as follows:
\begin{align*}
    &\evalExp{c}{\Sigma}{\Delta} \triangleq c \\
    &\evalExp{y}{\Sigma}{\Delta} \triangleq \Delta(y) \\
    &\evalExp{ \pproj{\varphi}(x) }{\Sigma}{\Delta} \triangleq \pproj{(\Sigma(\varphi))}(x) \\
    &\evalExp{ \lproj{\varphi}(x) }{\Sigma}{\Delta} \triangleq \lproj{(\Sigma(\varphi))}(x) \\
    &\evalExp{ e_1 \oplus e_2 }{\Sigma}{\Delta} \triangleq \evalExp{e_1}{\Sigma}{\Delta} \oplus \evalExp{e_2}{\Sigma}{\Delta} \\
    &\evalExp{f(e)}{\Sigma}{\Delta} \triangleq f(\evalExp{e}{\Sigma}{\Delta})
\end{align*}

Let $S$ be a set of states. The satisfiability of hyper-assertions is defined as follows:
\begin{align*}
    &\satSyntactic{S}{\Sigma}{\Delta}{b} \triangleq b \\
    &\satSyntactic{S}{\Sigma}{\Delta}{e_1 \succeq e_2} \triangleq \left( \evalExp{e_1}{\Sigma}{\Delta} \succeq \evalExp{e_2}{\Sigma}{\Delta} \right) \\
    &\satSyntactic{S}{\Sigma}{\Delta}{A \land B} \triangleq \left( \satSyntactic{S}{\Sigma}{\Delta}{A} \land \satSyntactic{S}{\Sigma}{\Delta}{B} \right) \\
    &\satSyntactic{S}{\Sigma}{\Delta}{A \lor B} \triangleq \left( \satSyntactic{S}{\Sigma}{\Delta}{A} \lor \satSyntactic{S}{\Sigma}{\Delta}{B} \right) \\
    &\satSyntactic{S}{\Sigma}{\Delta}{\forall x \ldotp A} \triangleq \left( \forall v \ldotp \satSyntactic{S}{\Sigma}{\Delta[x \mapsto v]}{A} \right) \\
    &\satSyntactic{S}{\Sigma}{\Delta}{\exists x \ldotp A} \triangleq \left( \exists v \ldotp \satSyntactic{S}{\Sigma}{\Delta[x \mapsto v]}{A} \right) \\
    &\satSyntactic{S}{\Sigma}{\Delta}{\forall \varphi \ldotp A} \triangleq \left(
        \forall \alpha \ldotp \satSyntactic{S}{\Sigma[\varphi \mapsto \alpha]}{\Delta}{A} \right) \\
    &\satSyntactic{S}{\Sigma}{\Delta}{\exists \varphi \ldotp A} \triangleq \left(
        \exists \alpha \ldotp \satSyntactic{S}{\Sigma[\varphi \mapsto \alpha]}{\Delta}{A} \right)
\end{align*}

When interpreting hyper-assertions in hyper-triples, we start with $\Delta$ and $\Sigma$ being the empty mappings,
except when there is an explicit quantifier around the triple,
such as in the premises for the rule \namerule{While-}$\exists$ from \figref{fig:loop-rules}.

\end{definition}

\begin{definition}\textbf{Syntactic transformation for deterministic assignments.}\\
$\transformAssign{e}{x}{A}$ \peter{yields} the \hyperassertion{} $A$, where $\varphi(x)$ is syntactically substituted by $\expApplied{e}{\varphi}$, for all (existentially or universally) quantified states $\varphi$:
\begin{align*}
&\transformAssign{e}{x}{b} \triangleq b \\
&\transformAssign{e}{x}{e_1 \succeq e_2} \triangleq e_1 \succeq e_2 \\
&\transformAssign{e}{x}{A \land B} \triangleq \transformAssign{e}{x}{A} \land \transformAssign{e}{x}{B} \\
&\transformAssign{e}{x}{A \lor B} \triangleq \transformAssign{e}{x}{A} \lor \transformAssign{e}{x}{B} \\
&\transformAssign{e}{x}{\forall x \ldotp A} \triangleq \forall x \ldotp \transformAssign{e}{x}{A} \\
&\transformAssign{e}{x}{\exists x \ldotp A} \triangleq \exists x \ldotp \transformAssign{e}{x}{A} \\
&\transformAssign{e}{x}{\forall \inSet{\varphi} \ldotp A} \triangleq \left( \forall \inSet{\varphi} \ldotp \transformAssign{e}{x}{A[\expApplied{e}{\varphi}/\varphi(x)]} \right) \\
&\transformAssign{e}{x}{\exists \inSet{\varphi} \ldotp A} \triangleq \left( \exists \inSet{\varphi} \ldotp \transformAssign{e}{x}{A[\expApplied{e}{\varphi}/\varphi(x)]} \right)
\end{align*}
where $A[y/x]$ refers to the standard syntactic substitution of $x$ by $y$.

\end{definition}

\begin{definition}\textbf{Syntactic transformation for non-deterministic assignments.}\\
$\transformHavoc{x}{A}$ \peter{yields} the \hyperassertion{} $A$ where $\varphi(x)$ is syntactically substituted by a fresh quantified variable $v$,
universally (resp. existentially) quantified for universally (resp. existentially) quantified states:
\begin{align*}
&\transformHavoc{x}{b} \triangleq b \\
&\transformHavoc{x}{e_1 \succeq e_2} \triangleq e_1 \succeq e_2 \\
&\transformHavoc{x}{A \land B} \triangleq \transformHavoc{x}{A} \land \transformHavoc{x}{B} \\
&\transformHavoc{x}{A \lor B} \triangleq \transformHavoc{x}{A} \lor \transformHavoc{x}{B} \\
&\transformHavoc{x}{\forall x \ldotp A} \triangleq \forall x \ldotp \transformHavoc{x}{A} \\
&\transformHavoc{x}{\exists x \ldotp A} \triangleq \exists x \ldotp \transformHavoc{x}{A} \\
&\transformHavoc{x}{\forall \inSet{\varphi} \ldotp A} \triangleq \left( \forall \inSet{\varphi} \ldotp \forall v \ldotp \transformHavoc{x}{A[v/\varphi(x)]} \right) \\
&\transformHavoc{x}{\exists \inSet{\varphi} \ldotp A} \triangleq \left( \exists \inSet{\varphi} \ldotp \exists v \ldotp \transformHavoc{x}{A[v/\varphi(x)]} \right)
\end{align*}
\end{definition}

\begin{definition}\textbf{Syntactic transformation for assume statements.}
\begin{align*}
&\transformAssume{p}{b} \triangleq b \\
&\transformAssume{p}{e_1 \succeq e_2} \triangleq e_1 \succeq e_2 \\
&\transformAssume{p}{A \land B} \triangleq \transformAssume{p}{A} \land \transformAssume{p}{B} \\
&\transformAssume{p}{A \lor B} \triangleq \transformAssume{p}{A} \lor \transformAssume{p}{B} \\
&\transformAssume{p}{\forall x \ldotp A} \triangleq \forall x \ldotp \transformAssume{p}{A} \\
&\transformAssume{p}{\exists x \ldotp A} \triangleq \exists x \ldotp \transformAssume{p}{A} \\
&\transformAssume{p}{\forall \inSet{\varphi} \ldotp A} \triangleq \forall \inSet{\varphi} \ldotp \expApplied{p}{\varphi} \Rightarrow \transformAssume{p}{A} \\
&\transformAssume{p}{\exists \inSet{\varphi} \ldotp A} \triangleq \exists \inSet{\varphi} \ldotp \expApplied{p}{\varphi} \land \transformAssume{p}{A}
\end{align*}
\end{definition}

%% file: appendix/unbounded.tex
Given a program with a low-sensitivity (\emph{low} for short) input $l$, a high-sensitivity (\emph{high for short}) input $h$,
and output $o$,
an interesting problem is to \emph{quantify} how much information about $h$ is leaked through $o$.
This information flow can be quantified \thibault{in different ways,
for example using \emph{Shannon entropy}~\cite{Shannon48},
or \emph{min-capacity}~\cite{assaf2017hypercollecting,Smith2009}.
\citet{Yasuoka2010} have shown that, for both definitions,
deciding whether the quantitative information flow is smaller than some upper-bound $q$ (given as input)
is \emph{not} $k$-safety for any $k$, and thus requires reasoning about an \emph{unbounded} number of executions.
The authors also show that, when $q$ is fixed, upper-bounding the quantitative information flow becomes $k$-safety (for some $k$)
when based on min-capacity, but not when based on Shannon entropy.
}

\thibault{For simplicity, we focus here on the problem of upper-bounding the quantitative information flow based on min-capacity,}
which can be reduced to quantifying the number of different values that the output $o$ can have,
given that the initial value of $l$ is fixed (but the initial value of $h$ is not)~\cite{assaf2017hypercollecting}.
The problem (1) of \emph{upper-bounding} the number of possible values of $o$ is hypersafety, but not $k$-safety for any $k>0$,
and thus requires the ability to reason about an \emph{unbounded} number of executions,
which is not possible in any existing Hoare logic,
but is possible in \logic{}.
The harder problem (2) of both \emph{lower-bounding} (to show that there is some leakage)
and upper-bounding this quantity is not hypersafety anymore, and thus requires to be able to reason directly about properties of sets,
in this case cardinality.

\begin{figure}[h]
\begin{align*}
&\cassign{o}{0} \cseq \\
&\cassign{i}{0} \cseq \\
&\mathbf{while} \; (i < \mathit{max}(l, h)) \; \{ \\
&\quad \chavoc{r} \cseq \\
&\quad \cassume{0 \leq r \leq 1} \cseq \\
&\quad \cassign{o}{o + r} \\
&\quad \cassign{i}{i + 1} \\
&\}
\end{align*}
\caption{The program $C_l$ that leaks information about the high-sensitivity input $h$ via its output $o$.}
\label{fig:min-leakage}
\end{figure}

As an example, consider the program $C_l$ shown in \figref{fig:min-leakage}.
Assuming that we know $h \geq 0$,
the output $o$ of this program can be at most $h$.
Hence, $o$ leaks information about $h$: We learn that $h \geq o$.
With respect to problem (1),
we can express that this program can have \emph{at most} $v+1$ output values,
where $v$ is the initial value of $l$, with the hyper-triple
$$
\simpleHoare{\always(h \geq 0) \land \low{l}}{C_l}{
\lambda S \ldotp \exists v \ldotp (\forall \varphi \in S \ldotp \varphi(l) = v) \land
|\{ \varphi(o) \mid \varphi \in S \}| \leq v + 1}
$$

\thibault{Note that this hyper-triple is equivalent to the following:
$$
\forall v \ldotp
\simpleHoare{\always(h \geq 0 \land l = v)}{C_l}{
\lambda S \ldotp
|\{ \varphi(o) \mid \varphi \in S \}| \leq v + 1}
$$
Given that the initial value of $l$ is $v$, there can be at most $v+1$ output values for $o$.
}

Moreover, with respect to the harder problem (2),
we can express that this program can have \emph{exactly} $v+1$ output values, with the hyper-triple
$$
\simpleHoare{\always(h \geq 0) \land \low{l}}{C_l}{
\lambda S \ldotp \exists v \ldotp (\forall \varphi \in S \ldotp \varphi(l) = v) \land
|\{ \varphi(o) \mid \varphi \in S \}| = v + 1}
$$
\thibault{or, equivalently:
$
\forall v \ldotp
\simpleHoare{\always(h \geq 0 \land l = v)}{C_l}{
\lambda S \ldotp
|\{ \varphi(o) \mid \varphi \in S \}| = v + 1}
$.}

%% file: appendix/expressivity.tex
In this section, we demonstrate the expressivity of the logic by showing that hyper-triples
can express the judgments of existing over- and underapproximating Hoare logics (\appref{subsec:safety} and \appref{subsec:liveness}) and enable reasoning about useful properties that go beyond over- and underapproximation (\appref{subsec:beyond}). All theorems and propositions in this section have been \proven{} in Isabelle/HOL.

\subsection{Overapproximate Hoare Logics}\label{subsec:safety}

The vast majority of existing Hoare logics prove the absence of bad (combinations of) program executions. To achieve that, they prove properties \emph{for all} (combinations of) executions, that is, they overapproximate the set of possible (combinations of) executions. In this subsection, we discuss overapproximate logics that prove properties of single executions or of $k$ executions (for a fixed number $k$), and show that \logic{} goes beyond them by also supporting properties of unboundedly or infinitely many executions.

\mypar{Single executions}

Classical Hoare Logic~\cite{HoareLogic} is an overapproximate logic for properties of single executions.
The meaning of triples can be defined as follows:
\begin{definition}\textbf{Hoare Logic (HL)}.
   Let $P$ and $Q$ be sets of extended states. Then
   \begin{align*}
\nhoare{HL}{P}{C}{Q}
 \triangleq
(\forall \varphi \in P \ldotp
\forall \sigma' \ldotp
\bigstep{C}{\pproj{\varphi}}{\sigma'}
\Rightarrow (\lproj{\varphi}, \sigma') \in Q
)
   \end{align*}
\end{definition}

This definition reflects the standard partial-correctness meaning of Hoare triples: executing $C$ in some initial state that satisfies $P$ can only lead to a final state that satisfies $Q$. This meaning can be expressed as a program hyperproperty as defined in \defref{def:hyperproperties}:

\begin{proposition}\textbf{HL triples express hyperproperties.}
    Given sets of extended states $P$ and $Q$, there exists a hyperproperty $\phyperpropertyfont{H}$ such that, for all commands $C$,
    $
        \hypersat{C}{H}$ iff $\nhoare{HL}{P}{C}{Q}
    $.
\end{proposition}

\begin{proof}
    We define
    $$\phyperpropertyfont{H} \triangleq
    \{ C \mid
        \forall \varphi \in P \ldotp
        \forall \sigma' \ldotp
        (\pproj{\varphi}, \sigma') \in \Sigma(C)
        \Rightarrow
        (\lproj{\varphi}, \sigma') \in Q
    \}
    $$
    and prove
    $\forall C \ldotp \hypersat{C}{H} \Longleftrightarrow \nhoare{HL}{P}{C}{Q}$.
\end{proof}

This proposition together with completeness of our logic implies the \emph{existence} of a proof in \logic{} for every valid classical Hoare triple. But there is an even stronger connection: we can map any assertion in classical Hoare logic to a hyper-assertion in \logic{}, which suggests a direct translation from classical Hoare logic to our \logic. 

The assertions $P$ and $Q$ of a valid Hoare triple characterize \emph{all} initial and \emph{all} final states of executing a command $C$. Consequently, they represent \emph{upper bounds} on the possible initial and final states. We can use this observation to map  classical Hoare triples to hyper-triples by interpreting their pre- and postconditions as upper bounds on sets of states.

\begin{proposition}\textbf{Expressing HL in \logic{}.}\label{prop:expressing_hl}
Let $\overline{P} \triangleq (\lambda S \ldotp S \subseteq P)$. \\ Then
$
\nhoare{HL}{P}{C}{Q}$
iff
$\hoare{\overline{P}}{C}{\overline{Q}}$.

Equivalently,
$
\nhoare{HL}{P}{C}{Q}$
iff
$\hoare{\forall \inSet{\varphi} \ldotp \varphi \in P}{C}{\forall \inSet{\varphi} \ldotp \varphi \in Q}$.
\end{proposition}

This proposition implies that some rules of \logic{} have a direct correspondence in HL\@.
For example, the rule \namerule{Seq} instantiated with $\overline{P}$, $\overline{R}$, and $\overline{Q}$ directly corresponds
to the sequential composition rule from HL\@.
Moreover, the upper-bound operator distributes over $\otimes$ and $\bigotimes$,
since
$\overline{A} \otimes \overline{B} = \overline{A \cup B}$,
and
$\bigotimes_i \overline{F_i} = \overline{\bigcup_i F(i)}$.
Consequently, we can for example easily derive in \logic{} the classic while-rule from HL, using the rule \namerule{Iter} from \figref{fig:rules}.

\mypar{$k$ executions}

Many extensions of HL have been proposed to deal with hyperproperties of $k$ executions. As a representative of this class of logics, we relate Cartesian Hoare Logic~\cite{CHL16} to our \logic{}.
To define the meaning of Cartesian Hoare Logic triples, we first lift our semantic relation $\rightarrow$ from one execution on states to $k$ executions on extended states.
    Let $k \in \mathbb{N^+}$.
    We write $\vecto{\varphi}$ to represent the $k$-tuple of extended states
    ($\varphi_1, \ldots, \varphi_k)$,
    and $\forall \vecto{\varphi}$ (resp.\ $\exists \vecto{\varphi}$)
    as a shorthand for $\forall \varphi_1, \ldots, \varphi_k$
    (resp.\ $\exists \varphi_1, \ldots, \varphi_k$).
    Moreover, we define the relation $\overset{k}{\rightarrow}$
    as
    $\bigstepK{C}{\varphi}{\varphi'}{k} \triangleq
    (\forall i \in \range{1}{k} \ldotp
   \bigstep{C}{\pproj{{\varphi_i}}}{\pproj{{\varphi'_i}}}
\;\land\;
\lproj{{\varphi_i}}=\lproj{{\varphi'_i}})$.

\begin{definition}
\label{def:CHL}
\textbf{Cartesian Hoare Logic (CHL).}
Let $k \in \mathbb{N^+}$,
and
let $P$ and $Q$ be sets of $k$-tuples of extended states. Then
\begin{align*}
\nhoare{CHL(k)}{P}{C}{Q}
\triangleq
&(
\forall \vecto{\varphi} \in P \ldotp
\forall \vecto{\varphi'} \ldotp
\bigstepK{C}{\varphi}{\varphi'}{k}
\Rightarrow
\vecto{\varphi'} \in Q
)
\end{align*}
\end{definition}

$\nhoare{CHL(k)}{P}{C}{Q}$ is valid iff
executing $C$ $k$ times in $k$ initial states that \emph{together} satisfy $P$
can only lead to $k$ final states that together satisfy $Q$.
This meaning can be expressed as a program hyperproperty:

\begin{proposition}
\label{prop:trans-CHL}
\textbf{CHL triples express hyperproperties.}
    Given sets of $k$-tuples of extended states $P$ and $Q$, there exists a hyperproperty $\phyperpropertyfont{H}$
    such that, for all commands $C$,
    $
        \hypersat{C}{H} \Longleftrightarrow \nhoare{CHL(k)}{P}{C}{Q}
    $.
\end{proposition}

\begin{proof}
    We define
    \begin{align*}
        \phyperpropertyfont{H} \triangleq
    \{ C \mid
        \forall \vecto{\varphi} \in P \ldotp
        \forall \vecto{\varphi'} \ldotp
        (\forall i \in \range{1}{k} \ldotp
        \lproj{\varphi}_i = \lproj{{\varphi'}}_i
       \land
        (\pproj{\varphi_i}, \pproj{{\varphi'}_i}) \in \Sigma(C))
        \Rightarrow \vecto{\varphi'} \in Q
    \}
    \end{align*}
    and prove
    $\forall C \ldotp \hypersat{C}{H} \Longleftrightarrow \nhoare{CHL(k)}{P}{C}{Q}$.
\end{proof}

Like we did for Hoare Logic, we can provide a direct translation from CHL triples to hyper-triples in our logic. Similarly to HL, CHL assertions express upper bounds, here on sets of $k$-tuples. However, simply using upper bounds as in \propref{prop:expressing_hl} does not capture the full expressiveness of CHL because executions in CHL are \emph{distinguishable}. For example, one can
express monotonicity from $x$ to $y$ as
$\nhoare{CHL(k)}{x(1) \geq x(2)}{\cassign{y}{x}}{y(1) \geq y(2)}$. When going from (ordered) tuples of states in CHL to (unordered) sets of states in \logic{}, we need to identify which state in the set of final states corresponds to execution $1$, and which state corresponds to execution $2$.
As we did in \appref{subsect:comp-examples} to express monotonicity,  we use a logical variable $t$ to tag a state with the number $i$ of the execution it corresponds to.

\begin{proposition}\label{prop:expressing_CHL}\textbf{Expressing CHL in \logic{}.}
    Let
    \begin{align*}
    &P' \triangleq (
        \forall \vecto{\varphi} \ldotp
    (\forall i \in \range{1}{k} \ldotp
    \inSet{\varphi_i}
    \land \lproj{{\varphi}_i}(t) = i)
	\Rightarrow \vecto{\varphi} \in P) \\
    &Q' \triangleq (
        \forall \vecto{\varphi} \ldotp
    (\forall i \in \range{1}{k} \ldotp
    \inSet{\varphi_i}
    \land \lproj{{\varphi}_i}(t) = i)
    \Rightarrow \vecto{\varphi} \in Q)
    \end{align*}
where $t$ does not occur free in $P$ or $Q$. Then
    $
    \nhoare{CHL(k)}{P}{C}{Q} \;\Longleftrightarrow\; \hoare{P'}{C}{Q'}
    $.
\end{proposition}

Recall that $\inSet{\varphi}$ is equivalent to $\lambda S \ldotp \varphi \in S$.
As an example, we can express the CHL assertion $y(1) \geq y(2)$ as the hyper-assertion
$\forall \inSet{\varphi_1}, \inSet{\varphi_2}
\ldotp \lproj{\varphi_1}(t) = 1 \land \lproj{\varphi_2}(t) = 2 \Rightarrow \pproj{\varphi_1}(y) \geq \pproj{\varphi_2}(y)$. Such translations provide a direct way of representing CHL proofs in \logic{}.

CHL, like \logic{}, can reason about multiple executions of a single command $C$, which is sufficient for many practically-relevant hyperproperties such as non-interference or determinism. Other logics, such as Relational Hoare Logic~\cite{Benton04}, relate the executions of multiple (potentially different) commands, for instance, to prove program equivalence. In case these commands are all the same, triples of relational logics can be translated to \logic{} analogously to CHL\@. We explain how to encode relational properties relating different commands to \logic{} in \appref{subsec:beyond}.

\mypar{Unboundedly many executions}

To the best of our knowledge, all existing overapproximate Hoare logics consider a fixed number $k$ of executions. In contrast, \logic{} can reason about an unbounded number of executions, as 
illustrated in \appref{app:unbounded}.

Moreover, since our hyper-assertions are functions of potentially-infinite sets of states, \logic{} can even express properties of infinitely-many executions, as we illustrate in \appref{subsec:beyond}.

\subsection{Underapproximate Hoare Logics}\label{subsec:liveness}

Several recent Hoare logics allow proving the \emph{existence} of certain (combinations of) program executions, which is useful, for instance, to disprove a specification, that is, to demonstrate that a program definitely has a bug. These logics underapproximate the set of possible (combinations of) executions. In this subsection, we discuss two forms of underapproximate logics, \emph{backward} and \emph{forward}, and show that both can be expressed in \logic{}.

\mypar{Backward underapproximation}

Reverse Hoare Logic~\cite{ReverseHL} and Incorrectness Logic~\cite{IncorrectnessLogic} are both underapproximate logics. Reverse Hoare Logic is designed to reason about the reachability of good final states. Incorrectness Logic uses the same ideas to prove the presence of bugs in programs. We focus on Incorrectness Logic in the following, but our results also apply to Reverse Hoare Logic.
Incorrectness Logic reasons about single program executions:

\begin{definition}\textbf{Incorrectness Logic (IL).}
   Let $P$ and $Q$ be sets of extended states. Then
$$
\nhoare{IL}{P}{C}{Q}
\triangleq
(\forall \varphi \in Q \ldotp \exists \sigma \ldotp
(\lproj{\varphi}, \sigma) \in P
\land \bigstep{C}{\sigma}{\pproj{\varphi}})
$$
\end{definition}

The meaning of IL triples is defined \emph{backward} from the postcondition: any state that satisfies the postcondition $Q$ can be reached by executing $C$ in an initial state that satisfies the precondition $P$. This meaning can be expressed as a program hyperproperty:

\begin{proposition}
\label{prop:trans-IL}
\textbf{IL triples express hyperproperties.}
    Given sets of extended states $P$ and $Q$, there exists a hyperproperty $\phyperpropertyfont{H}$
    such that, for all commands $C$,
    $\hypersat{C}{H}$
    iff $\nhoare{IL}{P}{C}{Q}$.
\end{proposition}

\begin{proof}
    We define
    $$\phyperpropertyfont{H} \triangleq
    \{ C \mid
        \forall \varphi \in Q \ldotp
        \exists \sigma \ldotp (\lproj{\varphi}, \sigma) \in P \land (\sigma, \pproj{\varphi}) \in \Sigma(C)
    \}
    $$
    and prove
    $\forall C \ldotp \hypersat{C}{H} \Longleftrightarrow \nhoare{IL}{P}{C}{Q}$.
\end{proof}

Hoare Logic shows the absence of executions by overapproximating the set of possible executions, whereas Incorrectness Logic shows the existence of executions by underapproximating it. This duality also leads to an analogous translation of IL judgments into \logic, which uses lower bounds on the set of executions instead of the upper bounds used in \propref{prop:expressing_hl}.

\begin{proposition}\textbf{Expressing IL in \logic.}
Let $\underline{P} \triangleq (\lambda S \ldotp P \subseteq S)$.
Then
$\nhoare{IL}{P}{C}{Q}$
iff
$\hoare{\underline{P}}{C}{\underline{Q}}$.

Equivalently,
$\nhoare{IL}{P}{C}{Q}$
iff
$\hoare{\forall \varphi \in P \ldotp \inSet{\varphi}}{C}{\forall \varphi \in Q \ldotp \inSet{\varphi} }$.
\end{proposition}

Analogous to the upper bounds for HL, the lower-bound operator
distributes over $\otimes$ and $\bigotimes$:
$\underline{A} \otimes \underline{B} = \underline{A \cup B}$
and
$\bigotimes_i \underline{F_i} = \underline{\bigcup_i F(i)}$.
Using the latter equality with the rules \namerule{While}
and \namerule{Cons}, it is easy to derive the loop rules from both 
Incorrectness Logic and Reverse Hoare Logic.

\medskip
\citet{InsecurityLogic} has recently proposed an underapproximate logic based on IL that can reason about two executions of two (potentially different) programs, for instance, to prove that a program violates a hyperproperty such as non-interference.
We use the name \emph{k-Incorrectness Logic} for the restricted version of this logic where the two programs are the same
(and discuss relational properties between different programs in \appref{subsec:beyond}).
The meaning of triples in k-Incorrectness Logic is also defined backward.
They express that, for any pair of final states $(\varphi'_1, \varphi'_2)$
that together satisfy a relational postcondition, there exist two initial states $\varphi_1$ and $\varphi_2$ that together
satisfy the relational precondition, and executing command $C$ in $\varphi_1$ (resp.\ $\varphi_2$) leads to $\varphi'_1$ (resp.\ $\varphi'_2$). Our formalization lifts this meaning from 2 to $k$ executions:

\begin{definition}\textbf{k-Incorrectness Logic (k-IL).}
Let $k \in \mathbb{N^+}$,
and $P$ and $Q$ be sets of $k$-tuples of extended states. Then
$\nhoare{k-IL}{P}{C}{Q}
\triangleq
(
\forall \vecto{\varphi'} \in Q \ldotp
\exists \vecto{\varphi} \in P \ldotp
\bigstepK{C}{\varphi}{\varphi'}{k}
)$.
\end{definition}

Again, this meaning is a hyperproperty:
\begin{proposition}
\label{prop:RIL-is-hyper}
\textbf{k-IL triples express hyperproperties.}
    Given sets of $k$-tuples of extended states $P$ and $Q$, there exists a hyperproperty $\phyperpropertyfont{H}$
    such that, for all commands $C$,
    $\hypersat{C}{H} \Longleftrightarrow \nhoare{k-IL}{P}{C}{Q}$.
\end{proposition}

\begin{proof}
    We define
    \begin{align*}
        \phyperpropertyfont{H} \triangleq 
    \{ C \mid
        \forall \vecto{\varphi'} \in Q \ldotp
        \exists \vecto{\varphi'} \in P \ldotp
        (\forall i \in \range{1}{k} \ldotp
        \lproj{\varphi_i} = \lproj{{\varphi'}_i}
       \land
        (\pproj{\varphi_i}, \pproj{{\varphi'}_i}) \in \Sigma(C))
    \}
    \end{align*}
    and prove
    $\forall C \ldotp \hypersat{C}{H} \Longleftrightarrow \nhoare{k-IL}{P}{C}{Q}$.
\end{proof}

Together with \thmref{thm:expressing_hyperprop}, this implies that we can express any k-IL triple as hyper-triple in \logic. However, defining a direct translation of k-IL triples to hyper-triples is surprisingly tricky. In particular, it is \emph{not} sufficient to apply the transformation from \propref{prop:expressing_CHL}, which uses a logical variable $t$ to tag each state with the number of the execution it belongs to. This approach works for Cartesian Hoare Logic because CHL and \logic{} are both forward logics (see \defref{def:hoare_triples} and \defref{def:CHL}). Intuitively, this commonality allows us to identify corresponding tuples from the preconditions in the two logics and relate them to corresponding tuples in the postconditions.

However, since k-IL is a \emph{backward} logic, the same approach is not sufficient to identify corresponding tuples. For two states
$\varphi'_1$ and $\varphi'_2$ from the set of final states, we know through the tag variable $t$ to which execution they belong, but not whether they originated from one tuple 
$(\varphi_1, \varphi_2)\in P$, or from two \emph{unrelated} tuples.

To solve this problem, we use another logical variable $u$, which records the ``identity'' of the initial $k$-tuple that satisfies $P$. To avoid cardinality issues, we define the encoding under the assumption that $P$ depends only on program variables. Consequently,
there are at most $|\states^k|$ such $k$-tuples, which we can represent as logical values if the cardinality of $\lvals$ is at least the cardinality of $\states^k$, as shown by the following result:

\begin{proposition}\textbf{Expressing k-IL in \logic.}
    Let $t,u$ be distinct variables in $\lvars{}$ and
    \begin{align*}
    P' \triangleq (
    \forall \vecto{\varphi} \in P \ldotp
            &(\forall i \in \range{1}{k} \ldotp
            \lproj{{\varphi}_i}(t) = i)
        \Rightarrow
            (\exists v \ldotp \forall i \in \range{1}{k}
            \ldotp
            \inSet{{\varphi}_i[u \mapsto v]}
            )
    )
    \\
    Q' \triangleq (
    \forall \vecto{\varphi'} \in Q \ldotp
            &(\forall i \in \range{1}{k} \ldotp
            \lproj{{\varphi'}_i}(t) = i)
        \Rightarrow
            (\exists v \ldotp \forall i \in \range{1}{k}
            \ldotp
            \inSet{{\varphi'_i}[u \mapsto v]}
            )
    )
    \end{align*}
    If
    (1) $P$ depends only on program variables,
    (2) the cardinality of $\lvals{}$ is at least the cardinality of $\states{}^k$, and
    (3) $t,u$ do not occur free in $P$ or $Q$,
    then
    $
    \nhoare{k-IL}{P}{C}{Q} \;\Longleftrightarrow\; \hoare{P'}{C}{Q'}
    $.
\end{proposition}

This proposition provides a direct translation for some k-IL triples into hyper-triples. Those that cannot be translated directly can still be verified with \logic{}, according to \propref{prop:RIL-is-hyper}.

\mypar{Forward underapproximation}

Underapproximate logics can also be formulated in a forward way:
Executing command $C$ in any state that satisfies the precondition
reaches at least one final state that satisfies the postcondition.
Forward underapproximation has recently been explored in Outcome Logic~\cite{OutcomeLogic}, a Hoare logic whose goal is to unify correctness (in the sense of classical Hoare logic) and incorrectness reasoning (in the sense of forward underapproximation) for single program executions. We focus on the underapproximation aspect of Outcome Logic here; overapproximation can be handled analogously to Hoare Logic (see \appref{subsec:safety}).
Moreover, we restrict the discussion to the programming language defined in \secref{subsec:language}; Outcome Logic also supports 
heap-manipulating and probabilistic programs, which we do not consider here.

Forward underapproximation for single executions can be formalized as follows:

\begin{definition}\textbf{Forward Underapproximation (FU).}\label{def:forward_underapproximation}
Let $P$ and $Q$ be sets of extended states. Then
    $\nhoare{FU}{P}{C}{Q} \triangleq
    \left(\forall \varphi \in P \ldotp 
    \exists \sigma' \ldotp
    \bigstep{C}{\pproj{\varphi}}{\sigma'} \land (\lproj{\varphi}, \sigma') \in Q
    \right)$
\end{definition}

This meaning can be expressed in \logic{} as follows: If we execute $C$ in a set of initial states that contains at least one state from $P$ then the set of final states will contain at least one state in $Q$.

\begin{proposition}\textbf{Expressing FU in \logic.}
$$
\nhoareSmall{FU}{P}{C}{Q}
\;\Longleftrightarrow\;
\hoareSmall{\lambda S \ldotp P \cap S \neq \varnothing}{C}{\lambda S \ldotp Q \cap S \neq \varnothing}
$$

Equivalently,
$\nhoare{FU}{P}{C}{Q}$
iff
$\hoareSmall{\exists \inSet{\varphi} \ldotp \varphi \in P}{C}{\exists \inSet{\varphi} \ldotp \varphi \in Q}$.

\end{proposition}

The precondition (resp.\ postcondition) states that the intersection between $S$ and $P$ (resp.\ $Q$) is non-empty.
If instead it required that $S$ is a \emph{non-empty subset} of $P$ (resp.\ $Q$), it would express the meaning of Outcome Logic triples, \ie the conjunction of classical Hoare Logic and forward underapproximation.

While Outcome Logic reasons about single executions only,
it is straightforward to generalize forward underapproximation to multiple executions:

\begin{definition}\textbf{k-Forward Underapproximation (k-FU).}
Let $k \in \mathbb{N^+}$,
and
let $P$ and $Q$ be sets of $k$-tuples of extended states. Then
$
\nhoare{k-FU}{P}{C}{Q}
\triangleq
(
\forall \vecto{\varphi} \in P \ldotp
\exists \vecto{\varphi'} \in Q \ldotp
\bigstepK{C}{\varphi}{\varphi'}{k}
)$.
\end{definition}

Again, this meaning can be expressed as a hyperproperty:

\begin{proposition}
\label{prop:trans-RFU}
\textbf{k-FU triples express hyperproperties.}
Given sets of $k$-tuples of extended states $P$ and $Q$, there exists a hyperproperty $\phyperpropertyfont{H}$
    such that, for all commands $C$,
    $
        \hypersat{C}{H} \Longleftrightarrow \nhoare{k-FU}{P}{C}{Q}
    $.
\end{proposition}

\begin{proof}
    We define
    \begin{align*}
        \phyperpropertyfont{H} \triangleq
    \{ C \mid
        \forall \vecto{\varphi} \in P \ldotp
        \exists \vecto{\varphi'} \in Q \ldotp
        (\forall i \in \range{1}{k} \ldotp
        \lproj{\varphi_i} = \lproj{{\varphi'}_i}
       \land
        (\pproj{\varphi_i}, \pproj{{\varphi'}_i}) \in \Sigma(C))
    \}
    \end{align*}
    and prove
    $\forall C \ldotp \hypersat{C}{H} \Longleftrightarrow \nhoare{k-FU}{P}{C}{Q}$.
\end{proof}

Since FU corresponds exactly to k-FU for $k=1$, this proposition applies also to FU\@.

Because k-FU is \emph{forward} underapproximate, we can use the tagging from \propref{prop:expressing_CHL} to translate k-FU triples into hyper-triples.
The following encoding intuitively corresponds to
the precondition $(S_1 \times \ldots \times S_k) \cap P \neq \varnothing$
and the postcondition
$(S_1 \times \ldots \times S_k) \cap Q \neq \varnothing$,
where $S_i$ corresponds to the set of states with $t = i$:

\begin{proposition}\textbf{Expressing k-FU in \logic.}
    \\ Let
        $P' \triangleq (
        \exists \vecto{\varphi} \in P \ldotp
            \forall i \in \range{1}{k} \ldotp
                \inSet{{\varphi}_i} \land \lproj{{\varphi}_i}(t) = i
        )$
        and
        $Q' \triangleq (
        \exists \vecto{\varphi'} \in Q \ldotp
            \forall i \in \range{1}{k} \ldotp
                \inSet{{\varphi'_i}} \land \lproj{{\varphi'_i}}(t) = i
        )$.

\noindent
    If $t$ does not occur free in $P$ or $Q$, then
    $
    \nhoare{k-FU}{P}{C}{Q} \;\Longleftrightarrow\; \hoare{P'}{C}{Q'}
    $.
\end{proposition}

\subsection{Beyond Over- and Underapproximation}\label{subsec:beyond}

In the previous subsections, we have discussed overapproximate logics, which reason about \emph{all} executions, and underapproximate logics, which reason about the \emph{existence} of executions. In this subsection, we explore program hyperproperties that combine universal and existential quantification, as well as properties that apply other comprehensions to the set of executions. We also discuss relational properties about multiple programs (such as program equivalence).

\mypar{$\forall^* \exists^*$-hyperproperties}

Generalized non-interference (see \secref{subsec:beyond-k-safety}) intuitively expresses that for each execution that produces a given observable output, there exists another execution that produces the same output using any other secret. That is, observing the output does not reveal any information about the secret. GNI is a 
hyperproperty that cannot be expressed in existing over- or underapproximate Hoare logics. It mandates the existence of an execution \emph{based on other possible executions}, whereas underapproximate logics can show only the existence of (combinations of) executions that satisfy some properties,
\emph{independently of the other possible executions}.
Generalized non-interference belongs to a broader class of $\forall^*\exists^*$-hyperproperties.

RHLE~\cite{RHLE} is a Hoare-style relational logic that has been recently proposed to verify $\forall^*\exists^*$-relational properties, such as program refinement~\cite{refinement}.
We call the special case of RHLE where triples specify properties of multiple executions of the same command \emph{k-Universal Existential}; we can formalize its triples as follows:

\begin{definition}\textbf{k-Universal Existential (k-UE).}
Let $k_1, k_2 \in \mathbb{N^+}$,
and
let $P$ and $Q$ be sets of $(k_1+k_2)$-tuples of extended states. Then
    \begin{align*}
        &\nhoare{k-UE(k_1,k_2)}{P}{C}{Q}
        \triangleq
        (
            \forall (\vecto{\varphi}, \vecto{\gamma}) \in P \ldotp
                \forall \vecto{\varphi'} \ldotp
                        \bigstepK{C}{\varphi}{\varphi'}{k_1}
                    \Rightarrow
                        (\exists \vecto{\gamma'} \ldotp
                            \bigstepK{C}{\gamma}{\gamma'}{k_2} \land
                            (\vecto{\varphi'}, \vecto{\gamma'}) \in Q
                        )
        )
    \end{align*}
\end{definition}

Given $k_1+k_2$ initial states $\varphi_1, \ldots, \varphi_{k_1}$ and $\gamma_1, \ldots, \gamma_{k_2}$ that together satisfy the precondition $P$,
for any final states $\varphi'_1, \ldots, \varphi'_{k_1}$ that can be reached by executing $C$ in the initial states
$\varphi_1, \ldots, \varphi_{k_1}$,
there exist $k_2$ final states $\gamma'_1, \ldots, \gamma'_{k_2}$ that can be reached by executing $C$ in the initial states $\gamma_1, \ldots, \gamma_{k_2}$,
such that $\varphi'_1, \ldots, \varphi'_{k_1}, \gamma'_1, \ldots, \gamma'_{k_2}$ together satisfy the postcondition $Q$.

The properties expressed by k-UE assertions are hyperproperties:

\begin{proposition}
\label{prop:trans-RHLE}
\textbf{k-UE triples express hyperproperties.}
Given sets of $(k_1 + k_2)$-tuples of extended states $P$ and $Q$, there exists a hyperproperty $\phyperpropertyfont{H}$
    such that, for all commands $C$,
    $
        \hypersat{C}{H} \Longleftrightarrow \nhoare{k-UE(k_1, k_2)}{P}{C}{Q}
    $.
\end{proposition}

\begin{proof}
    We define
    \begin{align*}
        \phyperpropertyfont{H} \triangleq
    \{ C \mid
        \forall (\vecto{\varphi}, \vecto{\gamma}) \in P \ldotp
        \forall \vecto{\varphi'} \ldotp
        &\left( \forall i \in \range{1}{k_1}
        \ldotp
            (\pproj{\varphi_i}, \pproj{{\varphi'}_i}) \in \Sigma(C) 
            \land
            \lproj{\varphi_i} = \lproj{{\varphi'}_i}
        \right) \\
        &\Rightarrow
        \exists \vecto{\gamma'} \ldotp
        (\vecto{\varphi'}, \vecto{\gamma'}) \in Q
        \land
        ( \forall i \in \range{1}{k_2}
        \ldotp
            (\pproj{\gamma_i}, \pproj{{\gamma'}_i}) \in \Sigma(C) 
            \land
            \lproj{\gamma_i} = \lproj{{\gamma'}_i}
        )
    \}
    \end{align*}
    and prove
    $\forall C \ldotp \hypersat{C}{H} \Longleftrightarrow \nhoare{k-UE(k_1, k_2)}{P}{C}{Q}$.
\end{proof}

They can be directly expressed in \logic{}, as follows:

\begin{proposition}\textbf{Expressing k-UE in \logic.}
    Let $t,u$ be distinct variables in $\lvars{}$, and
    \begin{align*}
    T_n &\triangleq (\lambda \vecto{\varphi} \ldotp
        \forall i \in \range{1}{k_n} \ldotp
            \inSet{{\varphi}_i} \land
            {\varphi}_i(t) = i \land
            {\varphi}_i(u) = n
    ) \\
        P' &\triangleq (
            \forall i \ldotp
                \exists \inSet{\varphi} \ldotp
                    \lproj{\varphi}(t) = i
                    \land
                    \lproj{\varphi}(u) = 2
            )
            \land
            (\forall \vecto{\varphi}, \vecto{\gamma} \ldotp
                T_1(\vecto{\varphi}) \land T_2(\vecto{\gamma})
                \Rightarrow (\vecto{\varphi}, \vecto{\gamma}) \in P
        ) \\
        Q' &\triangleq (
            \forall \vecto{\varphi'} \ldotp
                T_1(\varphi') \Rightarrow
                (\exists \vecto{\gamma'} \ldotp
                    T_2(\vecto{\gamma'}) \land (\vecto{\varphi'}, \vecto{\gamma'}) \in Q
                )
        )
    \end{align*}
where $t,u$ do not occur free in $P$ or $Q$. Then
    $
        \nhoare{k-UE(k_1,k_2)}{P}{C}{Q} \;\Longleftrightarrow\;
        \hoare{P'}{C}{Q'}
    $.
\end{proposition}

This proposition  borrows ideas from the translations of other logics we saw earlier. In particular, we use a logical variable $t$ to tag the executions, and an additional logical variable $u$ that indicates whether a state is universally ($u=1$) or existentially ($u=2$) quantified.

\mypar{$\exists^* \forall^*$-hyperproperties}

To the best of our knowledge, no existing Hoare logic can express $\exists^*\forall^*$-hyperproperties,
\ie the \emph{existence} of executions in relation to \emph{all} other executions.
As shown by the example in \secref{sec:syntactic},
$\exists^*\forall^*$-hyperproperties naturally arise when
disproving a $\forall^*\exists^*$-hyperproperty (such as GNI),
where the existential part can be thought of as a counter-example,
and the universal part as the proof that this is indeed a counter-example.
The existence of a minimum for a function computed by a command $C$ is another simple example of an $\exists^*\forall^*$-property,
as illustrated in \secref{subsec:proving-exists}.

\mypar{Properties using other comprehensions}

Some interesting program hyperproperties cannot be expressed by quantifying over states, but require other comprehensions over the set of states, such as counting or summation. As an example, the hyperproperty ``there are exactly $n$ different possible outputs for any given input'' cannot be expressed by quantifying over the states, but requires counting (see \appref{app:unbounded}). Other examples of such hyperproperties include statistical properties about a program:

\begin{example}
    \textbf{Mean number of requests.}
    Consider a command $C$ that, given some input $x$, retrieves and returns information from a database.
    At the end of the execution of $C$, variable $n$ contains the number of database requests that were performed.
    If the input values are
    uniformly distributed (and with a suitable precondition $P$),
    then the following hyper-triple expresses that the average number of requests performed by $C$ is at most $2$:
    $$
    \simpleHoare{P}{C}{\lambda S \ldotp \mathit{mean}_n^x(\{ \pproj{\varphi} \mid \varphi \in S \}) \leq 2}
    $$
    where $\mathit{mean}_n^x$ computes the average (using a suitable definition for the average if the set is infinite) of the value of $n$.%
\end{example}

To the best of our knowledge, \logic{} is the only Hoare logic that can prove this property; existing logics neither support reasoning about mean-comprehensions over multiple execution states nor 
reasoning about infinitely many executions \emph{at the same time} (which is necessary if the domain of input $x$ is infinite).

\mypar{Relational program properties}

Relational program properties typically relate executions of several \emph{different} programs and,  thus, do not correspond to program hyperproperties as defined in \defref{def:hyperproperties}. However, it is possible to construct a single \emph{product} program~\cite{productPrograms,barthe2013beyond} that encodes the executions of several given programs, such that relational properties can be expressed as hyperproperties of the constructed program and \proven{} in \logic{}.

We illustrate this approach on program refinement~\cite{refinement}.
A command $C_2$ \emph{refines} a command $C_1$ iff the set of pairs of pre- and post-states of $C_2$ is a subset of the corresponding set of $C_1$. Program refinement is a $\forall\exists$-property,
where the $\forall$ and the $\exists$ apply to different programs. To encode refinement, we construct a new product program that non-deterministically executes either $C_1$ or $C_2$, and we track in a logical variable $t$ which command was executed. This encoding allows us to express and prove refinement in \logic{}:

\begin{example}\textbf{Expressing program refinement in \logic.}\\
    Let $C \triangleq \cnif{ ( \cassign{t}{1} \cseq C_1 ) }{ ( \cassign{t}{2} \cseq C_2 ) } $.
    If $t$ does not occur free in $C_1$ or $C_2$ then
    $C_2$ refines $C_1$ iff
    $$
    \hoare{\top}{C}{
        \forall \inSet{\varphi} \ldotp
        \pproj{\varphi}(t) = 2 \Rightarrow \inSet{(\lproj{\varphi}, \pproj{\varphi}[t \coloneqq 1])}
    }
    $$
\end{example}

This example illustrates the general methodology to transform a relational property over different programs into an equivalent hyperproperty for a new product program,
and thus to reason about relational program properties in \logic{}.
Relational logics typically provide rules that align and relate parts of the different program executions; we present such a rule for \logic{} in \appref{app:synchronous}.

\medskip
This section demonstrated that \logic{} is sufficiently expressive to prove and disprove arbitrary program hyperproperties as defined in \defref{def:hyperproperties}. Thereby, it captures hyperproperties that are beyond the reach of
existing Hoare logics.

%% file: appendix/compositionality.tex
The core rules of \logic{} allow one to prove any valid hyper-triple, but not necessarily \emph{compositionally},
as explained in \secref{subsec:compositionality}.
As an example, consider the sequential composition of a command $C_1$ that satisfies \emph{generalized} non-interference (GNI) with a command $C_2$ that satisfies non-interference (NI).
We would like to prove that $C_1 \cseq C_2$ satisfies GNI (the weaker property).
As discussed in \secref{subsec:beyond-k-safety},
a possible postcondition for $C_1$ is
$\gni{l}{h} \triangleq (\forall \inSet{\varphi_1}, \inSet{\varphi_2} \ldotp \exists \inSet{\varphi} \ldotp \lproj{\varphi_1}(h) = \lproj{\varphi}(h) \land \pproj{\varphi}(l) = \pproj{\varphi_2}(l))$,
while a possible precondition for $C_2$ is
$\low{l} \triangleq (\forall \inSet{\varphi_1}, \inSet{\varphi_2} \ldotp \varphi_1(l) = \varphi_2(l))$.
The corresponding hyper-triples for $C_1$ and $C_2$ cannot be composed using the core rules.
In particular, rule \namerule{Seq} cannot be applied (even in combination with \namerule{Cons}), since the postcondition of $C_1$ does not imply the precondition of $C_2$.
Note that this observation does \emph{not} contradict completeness: By
\thmref{thm:complete}, it is possible to prove \emph{more precise} triples for $C_1$ and $C_2$, such that the postcondition of $C_1$ matches the precondition of $C_2$.
However, to enable modular reasoning, our goal is to construct the proof by composing the given triples for the individual commands rather than deriving new ones.

In this section, we present \emph{compositionality rules} for hyper-triples~(\appref{subsec:compositionality-rules}).
These rules are \emph{admissible} in \logic{}, in the sense that they do not modify the set of valid hyper-triples that can be proved.
Rather, these rules enable flexible compositions of hyper-triples (such as those discussed above).
We illustrate these rules on two examples (\appref{subsect:comp-examples}): Composing minimality with monotonicity, and GNI with NI\@.
All technical results presented in this section (soundness of the rules shown in \figref{fig:compositional_rules} and validity of the examples)
have been formalized and proved in \isabelle{}.

\subsection{Compositionality Rules}
\label{subsec:compositionality-rules}

\begin{figure}[h]
    \[
    \begin{array}{c}

    \Inf[\mathit{Linking}]{
        \forall \varphi_1, \varphi_2 \ldotp \left( \lproj{\varphi_1} = \lproj{\varphi_2} \land \shoare{ \inSet{\varphi_1} }{C}{ \inSet{ \varphi_2 } }
        \Longrightarrow \;
        \shoare{ P_{\varphi_1} }{C}{ Q_{\varphi_2} } \right)
    }{\shoare{ \forall \inSet{\varphi} \ldotp P_{\varphi} }{C}{  \forall \inSet{\varphi} \ldotp
    Q_{\varphi} }}

    \\[2em]

    \Inf[\mathit{And}]{\shoare{P_1}{C}{Q_1}}{\shoare{P_2}{C}{Q_2}}{\shoare{P_1 \land P_2}{C}{Q_1 \land Q_2}}

    \hspace{5mm}

    \Inf[\mathit{Or}]{\shoare{P_1}{C}{Q_1}}{\shoare{P_2}{C}{Q_2}}{\shoare{P_1 \lor P_2}{C}{Q_1 \lor Q_2}}

    \\[2em]

    \Inf[\mathit{FrameSafe}]{\shoare{P}{C}{Q}}{\text{no } \exists \inSet{\_} \text{ in } F}{\mathit{wr}(C) \cap \freevars{F} = \varnothing}{\shoare{P \land F}{C}{Q \land F}}

    \\[2em]

    \Inf[\mathit{Forall}]{\forall x \ldotp ( \shoare{P_x}{C}{Q_x} )}{\shoare{\forall x \ldotp P_x}{C}{\forall x \ldotp Q_x}}

    \hspace{5mm}

    \Inf[\mathit{IndexedUnion}]{\forall x \ldotp ( \shoare{P_x}{C}{Q_x} )}{\shoare{\bigotimes\nolimits_{x \in X} P_x}{C}{
        \bigotimes\nolimits_{x \in X} Q_x}}

    \\[2em]

    \Inf[\mathit{Union}]{\shoare{P_1}{C}{Q_1}}{\shoare{P_2}{C}{Q_2}}{\shoare{P_1 \otimes P_2}{C}{Q_1 \otimes Q_2}}

    \hspace{5mm}

    \Inf[\mathit{BigUnion}]{ \shoare{P}{C}{Q} }{ \shoare{\bigotimes P}{C}{\bigotimes Q} }

    \\[2em]

    \Inf[\mathit{Specialize}]{\shoare{P}{C}{Q}}{\mathit{wr}(C) \cap \freevars{b} = \varnothing}{\shoare{\transformAssume{b}{P}}{C}{\transformAssume{b}{Q}}}

    \\[2em]

    \Inf[\mathit{LUpdate}]{P \Rightarrow^V P'}{\shoare{P'}{C}{Q}}{\invariant{V}{Q}}{\shoare{P}{C}{Q}}

    \\[2em]

    \Inf[\mathit{LUpdateS}]{\shoare{P \land (\forall \inSet{\varphi} \ldotp \varphi(t) = \expApplied{e}{\varphi})}{C}{Q}}{ t \notin \freevars{P} \cup \freevars{Q} \cup \mathit{fv}(e) }{\shoare{P}{C}{Q}}

    \\[2em]

    \Inf[\mathit{AtMost}]{\shoare{P}{C}{Q}}{\shoare{\sqsubseteq P}{C}{\sqsubseteq Q}}
    \hspace{5mm}

    \Inf[\mathit{AtLeast}]{\shoare{P}{C}{Q}}{\shoare{\sqsupseteq P}{C}{\sqsupseteq Q}}

    \\[2em]

    \Inf[\mathit{True}]{\shoare{P}{C}{\top}}
    \hspace{5mm}

    \Inf[\mathit{False}]{\shoare{\bot}{C}{Q}}
    \hspace{5mm}

    \Inf[\mathit{Empty}]{\shoare{ \emp{} }{C}{ \emp{} }}
    \hspace{5mm}

    \end{array}
    \]

\caption{Compositionality rules of \logic{}. All these rules have been proven sound in \isabelle{}.
$\mathit{wr}(C)$ corresponds to the set of program variables that are potentially written by $C$ (\ie that appear on the left-hand side of an assignment),
while $\freevars{F}$ corresponds to the set of program variables that appear in look-up expressions for quantified states.
For example, $\freevars{\forall \inSet{\varphi} \ldotp \exists n \ldotp \pproj{\varphi}(x) = n^2} = \{x\}$.
The operators $\bigotimes$, $\sqsubseteq$, and $\sqsupseteq$ are defined as follows:
$\bigotimes P \triangleq \left( \lambda S \ldotp \exists F \ldotp (S = \bigcup_{S' \in F} S') \land (\forall S' \in F \ldotp P(S')) \right)$,
$\sqsubseteq P \triangleq \left( \lambda S \ldotp \exists S' \ldotp S \subseteq S' \land P(S') \right)$,
and
$\sqsupseteq P \triangleq \left( \lambda S \ldotp \exists S' \ldotp S' \subseteq S \Rightarrow P(S') \right)$.
}
\label{fig:compositional_rules}
\end{figure}

\figref{fig:compositional_rules} shows a (selection of) compositionality rules for \logic{}, which we discuss below.

\paragraph{Linking}
To prove hyper-triples of the form $\simpleHoare{\forall \inSet{\varphi_1} \ldotp P_{\varphi_1}}{C}{\forall \inSet{\varphi_2} \ldotp Q_{\varphi_2}}$, the rule \namerule{Linking} considers each pair of pre-state $\varphi_1$ and post-state $\varphi_2$ separately, and lets one assume that $\varphi_2$ can be reached by executing $C$ in the state $\varphi_1$, and that logical variables do not change during this execution.

\paragraph{Conjunctions and disjunctions}
\logic{} admits the usual rules for conjunction
(\namerule{And} and \namerule{Forall} in \figref{fig:compositional_rules})
and disjunction  (\namerule{Or} in \figref{fig:compositional_rules}
and the core rule \namerule{Exist} in \figref{fig:rules}).

\paragraph{Framing}
Similarly to the frame rules in Hoare logic and separation logic~\cite{SeparationLogic},
\logic{} admits rules that allow us to frame information about states that is not affected by the execution of $C$.
The rule \namerule{FrameSafe} allows us to frame
any hyper-assertion $F$ if
(1) it does not refer to variables that the program can modify, and (2) it does not existentially quantify over states.
While (1) is standard, (2) is specific to hyper-assertions:
Framing the existence of a state (\eg with $F \triangleq \exists \inSet{\varphi} \ldotp \top$)
would be unsound if the execution of the program in the state $\varphi$ does not terminate.
We show in \appref{app:termination} that restriction (2) can be lifted if $C$ terminates.
\depsOnFibonacci{We also show an example of how this rule is used in \appref{app:fibonacci}.}

\paragraph{Decompositions}
As explained at the beginning of this section, the
two triples $\simpleHoare{P}{C_1}{\gni{l}{h}}$
and $\simpleHoare{\low{l}}{C_2}{Q}$ cannot be composed
because $\gni{l}{h}$ does not entail $\low{l}$ (not all states in the set $S$ of final states of $C_1$ need to have the same value for $l$). 
However, we can prove GNI for the composed commands by decomposing $S$ into subsets that all satisfy $\low{l}$ and considering each subset separately.
The rule \namerule{BigUnion} allows us to perform this decomposition (formally expressed with the hyper-assertion $\bigotimes \low{l}$),
use the specification of $C_2$ on each of these subsets (since they all satisfy the precondition of $C_2$),
and eventually recompose the final set of states
(again with the operator $\bigotimes$) 
to prove our desired postcondition.
\logic{} also admits rules for binary unions (rule \namerule{Union})
and indexed unions (rule \namerule{IndexedUnion}).

Note that unions ($\otimes$ and $\bigotimes$) and
disjunctions in hyper-assertions are very \emph{different}:
$(P \otimes Q)(S)$ expresses that the set $S$ can be decomposed into two sets $S_P$ (satisfying $P$) and $S_Q$ (satisfying $Q$),
while $(P \lor Q)(S)$ expresses that the entire set $S$ satisfies $P$ or $Q$.
Similarly, intersections and conjunctions are very different:
While \logic{} admits conjunction rules,
rules based on intersections would be unsound,
as shown by the following example:
\begin{example} Let $P_1 \triangleq (\lambda S \ldotp \exists \varphi \ldotp S = \{ \varphi \} \land \varphi(x) = 1)$,
and $P_2 \triangleq (\lambda S \ldotp \exists \varphi \ldotp S = \{ \varphi \} \land \varphi(x) = 2)$.
Both triples $\simpleHoare{P_1}{\cassign{x}{1}}{P_1}$ and $\simpleHoare{P_2}{\cassign{x}{1}}{P_1}$ are valid,
but the triple
$$\simpleHoare{
    \lambda S \ldotp \exists S_1, S_2 \ldotp S = S_1 \cap S_2 \land P_1(S_1) \land P_2(S_2)}{\cassign{x}{1}}{
    \lambda S \ldotp \exists S_1, S_2 \ldotp S = S_1 \cap S_2 \land P_1(S_1) \land P_1(S_2)}$$
is invalid, as the precondition is equivalent to $\emp$, but the postcondition is satisfiable by a non-empty set (with states satisfying $x = 1$).
\end{example}

\paragraph{Specializing hyper-triples}
By definition, a hyper-triple can only be applied to a set of states that satisfies its precondition, which can be restrictive.
In cases where only a \emph{subset} of the current set of states satisfies the precondition, one can obtain a \emph{specialized}
triple using the rule \namerule{Specialize}.
This rule uses the syntactic transformation $\transformAssumeEmpty{b}$ defined in \secref{subsec:assume-rule}
to weaken both the precondition and the postcondition of the triple, which is sound as long as the validity of $b$ is not influenced by executing $C$.
Intuitively, $\transformAssume{b}{P}$ holds for a set $S$ iff $P$ holds for the subset of states from $S$ that satisfy $b$.
As an example, the triple\\
$\simpleHoare{\always(t \seq 1 {\Rightarrow} x {\geq} 0) \sand \always(t \seq 2 {\Rightarrow} x {<} 0)}{C}
{ \forall \inSet{\varphi_1} {,} \inSet{\varphi_2}
\ldotp \varphi_1(t) {=} 1 {\land} \varphi_2(t){=} 2 {\Rightarrow} \varphi_1(y) {\geq} \varphi_2(y) }$,
whose postcondition corresponds to $\mono{t}{y}$,
can be derived from the two triples
$\simpleHoare{\always(x \geq 0)}{C}{\always(y \geq 0)}$ and
$\simpleHoare{\always(x < 0)}{C}{\always(y < 0)}$,
by applying the rule \namerule{Specialize} twice,
using $b \triangleq (t {=} 1)$ and $b \triangleq (t {=} 2)$ respectively,
followed by the consequence rule.

\paragraph{Logical updates}
Logical variables play an important role in the expressivity of the logic:
As we have informally shown in \secref{subsec:k-safety}
and formally shown
in \appref{app:expressivity},
relational specifications are typically expressed in \logic{} by using logical variables to formally link the pre-state of an execution with the corresponding post-states. Since logical variables cannot be modified by the execution, these tags are preserved.

To apply this proof strategy with existing triples, it is often necessary to update logical variables to introduce such tags. The rule \namerule{LUpdate} allows us to update the logical variables in a set $V$, provided that (1)~from every set of states $S$ that satisfies $P$,
we can obtain a new set of states $S'$ that satisfies $P'$, by only updating (for each state) the logical variables in $V$, (2)~we can prove the triple with the updated set of initial states, and (3)~the postcondition $Q$ cannot distinguish between two sets of states that are equivalent up to logical variables in $V$.
We formalize this intuition  in the following:

\begin{definition}\textbf{Logical updates.}
    Let $V$ be a set of logical variable names.
    Two states $\varphi_1$ and $\varphi_2$ are equal up to logical variables $V$,
    written $\varphi_1 \overset{V}{=} \varphi_2$, iff $\forall i \ldotp i \not\in V \Rightarrow \lproj{\varphi_1}(i) = \lproj{\varphi_2}(i)$ and $\pproj{\varphi_1} = \pproj{\varphi_2}$.

    Two sets of states $S_1$ and $S_2$ are equivalent up to logical variables $V$, written $S_1 \overset{V}{=} S_2$,    iff
every state $\varphi_1 \in S_1$ has a corresponding state $\varphi_2 \in S_2$ with the same values for all variables except those in $V$, and vice-versa:
    $$
    (\forall \varphi_1 \in S_1 \ldotp \exists \varphi_2 \in S_2 \ldotp \varphi_1 \overset{V}{=} \varphi_2)
    \land
    (\forall \varphi_2 \in S_2 \ldotp \exists \varphi_1 \in S_1 \ldotp \varphi_1 \overset{V}{=} \varphi_2)
    $$

    A \hyperassertion{} $P$ entails a \hyperassertion{} $P'$ modulo logical variables $V$, written $P \logEntails{V} P'$, iff
    $$
    \forall S \ldotp P(S) \Longrightarrow (\exists S' \ldotp P'(S') \land S \overset{V}{=} S')
    $$

    Finally, a hyper-assertion $P$ is invariant with respect to logical updates in $V$, written $\invariant{V}{P}$, iff
$$
\forall S_1, S_2 \ldotp S_1 \overset{V}{=} S_2 \Longrightarrow (P(S_1) \Longleftrightarrow P(S_2))
$$
    
\end{definition}

Note that $\invariant{V}{Q}$ means that $Q$ cannot inspect the value of logical variables in $V$, but it usually also implies that $Q$ cannot check for \emph{equality} between states, and cannot inspect the cardinality of the set, since updating logical variables might collapse two states that were previously distinct (because of distinct values for logical variables in $V$).

Since this rule requires semantic reasoning,
we also derive a weaker syntactic version of this rule, \namerule{LUpdateS}, which is easier to use.
The rule \namerule{LUpdateS} allows us to %
strengthen a precondition $P$ to 
$P \land (\forall \inSet{\varphi} \ldotp \varphi(t) = \expApplied{e}{\varphi})$,
which corresponds to updating the logical variable $t$ with the expression $e$,
as long as the logical variable $t$ does not appear \emph{syntactically} in $P$, $Q$, and $e$
(and thus does not influence their validity).
For example, to connect the postcondition $\always(x = 0 \lor x = 1)$
to the precondition $\mono{t}{x} = \left( \forall \inSet{\varphi_1}, \inSet{\varphi_2} \ldotp \varphi_1(t) = 1 \land \varphi_2(t) = 2 \Rightarrow \varphi_1(x) \geq \varphi_2(x) \right)$
described in \secref{subsec:k-safety},
one can use this rule to assign $1$ to $t$ if $x = 1$, and $2$ otherwise.
\depsOnFibonacci{\appref{app:fibonacci} shows a detailed example.}

\subsection{Examples}
\label{subsect:comp-examples}

We now illustrate our compositionality rules on two examples:
Composing minimality and monotonicity, and composing non-interference with generalized non-interference.

\subsubsection{Composing Minimality and Monotonicity}\label{subsec:composing_mini_mono}

Consider a command $C_1$ that computes a function that has a minimum for $x$,
and a deterministic command $C_2$ that is monotonic from $x$ to $y$.
We want to prove \emph{compositionally} that $C_1 \cseq C_2$ has a minimum for $y$.

More precisely, we assume that 
$C_1$ satisfies the specification
$\simpleHoare{P}{C_1}{\hasMin{x}}$, where $\hasMin{x} \triangleq ( \exists \inSet{\varphi} \ldotp \forall \inSet{\varphi'} \ldotp \pproj{\varphi}(x) \leq \pproj{\varphi'}(x) )$,
and $C_2$ satisfies the two specifications $\simpleHoare{\mono{i}{x}}{C_2}{\mono{i}{y}}$ (monotonicity) and $\simpleHoare{\singleton{}}{C_2}{\singleton{}}$ (determinism\footnote{This triple ensures that $C_2$ does not map the initial state with the minimum value for $x$ to potentially  different states with incomparable values for $y$ (the order $\leq$ on values might be partial). Moreover, it ensures that $C_2$ does not drop any initial states because of an assume command or a non-terminating loop.}),
where $\mono{i}{x} \triangleq ( \forall \inSet{\varphi_1}, \inSet{\varphi_2} \ldotp \lproj{\varphi_1}(i) = 1 \land \lproj{\varphi_2}(i) = 2 \Rightarrow \pproj{\varphi_1}(x) \leq \pproj{\varphi_2}(x) )$,
and $\singleton \triangleq (\exists \inSet{\varphi} \ldotp \forall \inSet{\varphi'} \ldotp \varphi = \varphi')$.
With the core rules alone, we cannot compose the two triples to prove that $C_1 \cseq C_2$ has a minimum for $y$
since the postcondition of $C_1$ does not imply the precondition of $C_2$.

\begin{figure*}[t!]
\begin{center}
\scriptsize
\begin{equation}
    \Inf[\mathit{FrameSafe}]{
        \Inf[\mathit{Specialize}]{
            \shoare{\singleton{}}{C_2}{\singleton{}}
        }
        {
            \shoare{
                \transformAssume{i = 1}{\singleton{}}
             }{C_2}{
                \transformAssume{i = 1}{\singleton{}}
            }
        }
    }{
        \shoare{
            \underbrace{
                 \transformAssume{i = 1}{\singleton{}} \land (\forall \inSet{\varphi} \ldotp \lproj{\varphi}(i) \in \{1,2\})
            }_{P'}
        }{C_2}{
            \underbrace{
                \transformAssume{i = 1}{\singleton{}} \land (\forall \inSet{\varphi} \ldotp \lproj{\varphi}(i) \in \{1,2\})
            }_{Q'}
        }
    }
\end{equation}
\vspace{-4mm}
\begin{equation*}
        \Inf[\mathit{Seq}]{
            \shoare{P}{C_1}{\hasMin{x}}
        }
        {
            \Inf[\mathit{LUpdate}]{
                \hasMin{x} \logEntails{\{i\}}
                    \mono{i}{x} \land P'
            }{
        \Inf[\mathit{Cons}]
        {
            \Inf[\mathit{And}]{
                    \shoare{ \mono{i}{x} }{C_2}{ \mono{i}{y} }
            }{
                \Inf{(1)}{
                    \shoare{P'}{C_2}{Q'}
                }
            }{
                \shoare{ \mono{i}{x} \land P'}{C_2}{ \mono{i}{y} \land Q' }
            }
        }{
            \shoare{ \mono{i}{x} \land P'}{C_2}{ \hasMin{y} }
        }
            }{
                \invariant{\{i\}}{ \hasMin{y} }
            }{
                \shoare{\hasMin{x}}{C_2}{\hasMin{y}}
            }
        }
        {
            \shoare{P}{C_1 \cseq C_2}{\hasMin{y}}
        }
\end{equation*}
\end{center}
\vspace{-4mm}
\caption{A compositional proof that the sequential composition of a command that has a minimum and a monotonic, deterministic command in turn has a minimum. 
Recall that 
$\singleton \triangleq (\exists \inSet{\varphi} \ldotp \forall \inSet{\varphi'} \ldotp \varphi = \varphi')$,
and thus
$\transformAssume{i = 1}{\singleton{}}
= (\exists \inSet{\varphi} \ldotp \varphi(i) = 1 \land (\forall \inSet{\varphi'} \ldotp
                \varphi'(i) = 1 \Rightarrow \varphi = \varphi'))$
}
\label{fig:composing_mini_mono}
\end{figure*}

\figref{fig:composing_mini_mono} shows a valid derivation in \logic{} of
$\shoare{P}{C_1 \cseq C_2}{\hasMin{y}}$ (which we have proved in \isabelle{}).
The key idea is to use the rule \namerule{LUpdate} to mark the minimal state with $i = 1$, and all the other states with $i = 2$, in order to match $C_1$'s postcondition with $C_2$'s precondition.
Note that we \emph{had to} use the consequence rule to turn $C_2$'s postcondition
$\mono{i}{y} \land Q'$ into $\hasMin{y}$ \emph{before} applying the rule \namerule{LUpdate},
because the latter hyper-assertion is invariant \wrt\ logical updates in $\{i\}$ (as required by the rule \namerule{LUpdate}), whereas the former is not.

The upper part of \figref{fig:composing_mini_mono} shows the derivation of $\shoare{P'}{C_2}{Q'}$, which uses \namerule{Specialize} to restrict the triple $\simpleHoare{\singleton{}}{C_2}{\singleton{}}$ to the subset of states where $i = 1$, ensuring the existence of a unique state (the minimum) where $i = 1$ after executing $C_2$.
We also use the rule \namerule{FrameSafe} to ensure that our set only contains states with $i = 1$ or $i = 2$.

\subsubsection{Composing Non-Interference with Generalized Non-Interference}
\label{subsec:composing_gni_sni}

To illustrate additional compositionality rules, we re-visit the example introduced at the beginning of this section. Consider a command $C_1$ that satisfies GNI (for a public variable $l$ and a secret variable $h$) and a command $C_2$ that satisfies 
NI (for the public variable $l$).
We want to prove that $C_1 \cseq C_2$ satisfies GNI (for $l$ and $h$).

More precisely, we assume that $C_1$ satisfies the hyper-triple
$\shoare{\low{l}}{C_1}{\gni{l}{h}}$, where $\gni{l}{h} \triangleq (\forall \inSet{\varphi_1}, \inSet{\varphi_2} \ldotp \exists \inSet{\varphi} \ldotp \lproj{\varphi_1}(h) = \lproj{\varphi}(h) \land \pproj{\varphi}(l) = \pproj{\varphi_2}(l))$.
Moreover, we assume that $C_2$ satisfies the triples
$\shoare{\low{l}}{C_2}{\low{l}}$, where $\low{l} \triangleq (\forall \inSet{\varphi_1}, \inSet{\varphi_2} \ldotp \pproj{\varphi_1}(l) = \pproj{\varphi_2}(l))$,
and $\shoare{\lnot \emp{}}{C_2}{\lnot \emp{}}$.
The second triple is needed to ensure that $C_2$ does not drop executions depending on some values for $h$ (\eg because of secret-dependent non-termination), which might cause $C_1 \cseq C_2$ to violate GNI\@.

\begin{figure*}
    \footnotesize
    \begin{equation}
        \Inf[\mathit{Cons}]{
            \Inf[\mathit{BigUnion}]{
                \Inf[\mathit{And}]{
                    \shoare{\low{l}}{C_2}{\low{l}}
                }{
                    \Inf[\mathit{Cons}]{
                        \Inf[\mathit{Specialize}]{ \shoare{\lnot \emp{}}{C_2}{\lnot \emp{}} }{
                            \shoare{ 
                                \transformAssume{h = \lproj{\varphi_1}(h)}{\lnot \emp{}}
                            }{C_2}{
                                \transformAssume{h = \lproj{\varphi_1}(h)}{\lnot \emp{}}
                            }
                        }
                    }{
                        \shoare{\exists \inSet{\varphi} \ldotp \lproj{\varphi}(h) = \lproj{\varphi_1}(h)}{C_2}{
                        \exists \inSet{\varphi} \ldotp \lproj{\varphi}(h) = \lproj{\varphi_1}(h)}
                    }
                }{
                    \shoare{\low{l} \land (\exists \inSet{\varphi} \ldotp \lproj{\varphi}(h) = \lproj{\varphi_1}(h) )}{C_2}{
                    \low{l} \land (\exists \inSet{\varphi} \ldotp \lproj{\varphi}(h) = \lproj{\varphi_1}(h) )}
                }
            }
            {
                \shoare{
                    \bigotimes \left( \low{l} \land (\exists \inSet{\varphi} \ldotp \lproj{\varphi}(h) = \lproj{\varphi_1}(h) ) \right)
                }
                {C_2}{
                    \bigotimes \left( \low{l} \land (\exists \inSet{\varphi} \ldotp \lproj{\varphi}(h) = \lproj{\varphi_1}(h) ) \right)
                }
            }}{
                \shoare{
                    \underbrace{\forall \inSet{\varphi_2} \ldotp \exists \inSet{\varphi} \ldotp \lproj{\varphi_1}(h) = \lproj{\varphi}(h) \land \pproj{\varphi_2}(l) = \pproj{\varphi}(l)}_{P'_{\varphi_1}}
                }{C_2}{
                    \underbrace{\forall \inSet{\varphi_2} \ldotp \exists \inSet{\varphi} \ldotp \lproj{\varphi_1}(h) = \lproj{\varphi}(h) \land \pproj{\varphi_2}(l) = \pproj{\varphi}(l)}_{Q'_{\varphi_1}}
                }
            }
            \label{eq:proving_linking}
    \end{equation}
\vspace{-2mm}    
    \begin{equation*}
        \Inf[\mathit{Seq}]{
            \shoare{ \low{l} }{C_1}{ \gni{l}{h} }
        }{
            \Inf[\mathit{Linking}]{
                \Inf{\text{using } (\ref{eq:proving_linking})\text{ and }
                \lproj{\varphi_1} = \lproj{\varphi_1'} \Longrightarrow Q'_{\varphi_1} = Q'_{\varphi_1'}
                }{
                    \forall \varphi_1, \varphi_1' \ldotp ( \lproj{\varphi_1} = \lproj{\varphi_1'} \land \shoare{ \inSet{\varphi_1} }{C}{ \inSet{ \varphi_1' } }
                    \Longrightarrow
                    (\shoare{
                        P'_{\varphi_1}
                    }{C_2}{
                        Q'_{\varphi_1'}
                    })
                }
            }{
                \shoare{ \gni{l}{h} }{C_2}{ \gni{l}{h} }
            }
        }{
            \shoare{ \low{l} }{C_1 \cseq C_2}{ \gni{l}{h} }
        }
    \end{equation*}
\vspace{-4mm}
\caption{A compositional proof that the sequential composition of a command that satisfies GNI and a command that satisfies NI in turn satisfies GNI.}
    \label{fig:composing_gni_sni}
\end{figure*}

\figref{fig:composing_gni_sni} shows a valid derivation of the triple $\shoare{\low{l}}{C_1 \cseq C_2}{\gni{l}{h}}$
(which we have proved in \isabelle{}).
The first key idea of this derivation is to use the rule \namerule{Linking} to eliminate the $\forall \inSet{\varphi_1}$ in the pre- and postcondition of the triple $\simpleHoare{\gni{l}{h}}{C_2}{\gni{l}{h}}$,
while assuming that they have the same value for the logical variable $h$ (implied by the assumption $\lproj{\varphi_1} = \lproj{\varphi_1'}$).
The second key idea is to decompose any set of states $S$ that satisfies $P'_{\varphi_1}$
(defined as $\forall \inSet{\varphi_2} \ldotp \exists \inSet{\varphi} \ldotp \lproj{\varphi_1}(h) = \lproj{\varphi}(h) \land \pproj{\varphi_2}(l) = \pproj{\varphi}(l)$)
into a union of smaller sets that all satisfy $\low{l} \land (\exists \inSet{\varphi} \ldotp \lproj{\varphi_1}(h) = \lproj{\varphi}(h))$.
More precisely, we rewrite $S$ as the union of all sets $\{ \varphi, \varphi_2 \}$
for all $\varphi, \varphi_2 \in S$ such that
$\lproj{\varphi_1}(h) = \lproj{\varphi}(h) \land \pproj{\varphi_2}(l) = \pproj{\varphi}(l)$, using the rule \namerule{Cons}.
Unlike $S$, these smaller sets all satisfy the precondition $\low{l}$ of $C_2$,
which allows us to leverage the triple $\shoare{\low{l}}{C_2}{\low{l}}$.
Finally, we use the rule \namerule{Specialize} to prove that,
after executing $C_2$ in each of the smaller sets $\{ \varphi, \varphi_2 \}$,
there will exist at least one state $\varphi'$ with $\lproj{\varphi'}(h) = \lproj{\varphi_1}(h)$.

%% file: appendix/termination.tex
\subsection{Termination-Based Rules}

In \appref{app:compositionality}, we have introduced the rule \namerule{FrameSafe} (\figref{fig:compositional_rules}),
which is sound only for hyper-assertions that do not contain any $\exists \inSet{\_}$,
because the program $C$ around which we want to frame some hyper-assertion might not terminate.
Moreover, in \secref{subsec:sync}, we have introduced the synchronized while rule \namerule{WhileSync} (\figref{fig:loop-rules}).
This rule contains an $\emp{}$ disjunct in the postcondition of the conclusion,
which prevents it from being useful to prove hyperproperties of the form $\exists^+ \forall^*$,
\ie with a top-level existential quantifier over states.
This $\emp{}$ disjunct corresponds to the case where the loop does not terminate.

In this section, we show that we can overcome those two limitations by introducing \thibault{\emph{terminating}} hyper-triples,
which are stronger than normal hyper-triples, in that they also ensure
the existence of \thibault{(at least)} one terminating execution from any initial state:

\begin{definition}\textbf{\thibault{Terminating} hyper-triples.}
    \label{def:total-triples}
    $$\hoareTot{P}{C}{Q} \triangleq \left( 
        \forall S \ldotp P(S)
        \Rightarrow
        (Q(\sem(C, S))
        \land (\forall \varphi \in S \ldotp
        \exists \sigma' \ldotp
        \bigstep{C}{\pproj{\varphi}}{\sigma'})
        )
    \right)$$
\end{definition}

\begin{figure}
$$\Inf[\mathit{Frame}]{\mathit{wr}(C) \cap \mathit{fv}(F) = \varnothing}{\shoareTot{P}{C}{Q}}{F \text{ is a syntactic hyper-assertion}}{
        \shoareTot{ P \land F}{C}{
        Q \land F}}
$$
$$
    \Inf[\mathit{WhileSyncTerm}]
    { \shoareTot{I \land \square (b \land e = t^L)}{C}{I \land \low{b} \land \square (e \prec t^L)}}
    {\prec \text{well-founded}}
    {t^L \notin \freevars{I}}
    {\shoareTot{I \land \low{b}}{\cwhile{b}{C}}{ I \land \square (\lnot b) }}
$$

$$
    \Inf[\mathit{WhileDesugaredTerm}]
    { \shoare{P_n}{ \cassume{b} }{ Q_n } }
    { \shoareTot{Q_n \land \square (e = t^L)}{C}{P_{n+1} \land \square (e \prec t^L)}}.
    {\shoare{ \otimesn P_n }{\cassume{\lnot b}}{R}}
    {t^L \notin \freevars{P_n} \cup \freevars{Q_n}}
    {\prec \text{well-founded}}
    {\shoareTot{P_0}{\cwhile{b}{C}}{ R }}
$$

    \caption{\thibault{Rules based on terminating hyper-triples.}}
    \label{fig:terminating-rules}
\end{figure}

\thibault{Using terminating hyper-triples, we can now express and prove sound (which we have done in \isabelle{})
the two rules \namerule{Frame} and \namerule{WhileSyncTerm} in \figref{fig:terminating-rules}, which solve the aforementioned limitations.}
The rule \namerule{Frame} can be used for \emph{any} hyper-assertion expressed in the syntax
defined in~\secref{subsec:syntactic-assertions}.
Unlike the rule \namerule{WhileSync}, the rule \thibault{\namerule{WhileSyncTerm}}
does not have the $\emp{}$ disjunct in the postcondition of its conclusion anymore,
and thus can be used to prove hyperproperties of the form $\exists^+ \forall^*$!
It achieves this by requiring that (1) the loop body $C$ terminates (in the sense of \defref{def:total-triples}),
and (2) that the loop itself terminates \thibault{(also in the sense of \defref{def:total-triples})},
by requiring that a variant $e$ decreases in all executions.
The initial value of the variant $e$ is stored in the logical variable $t^L$, such that it can be referred to
in the postcondition.

\paragraph{Rules for terminating hyper-triples.}

\thibault{For any program statement $C$ that does not contain any while loop
or $\mathbf{assume}$ statement
(apart from those used in the desugaring of $\cifonly{b}{C_1}$ and $\cif{b}{C_1}{C_2}$),
both triples are equivalent, \ie{}
$\hoareTot{P}{C}{Q} \Longleftrightarrow \hoare{P}{C}{Q}$.
We have proved the corresponding inference rules sound in our \isabelle{} mechanization.
Moreover, on top of the rule \namerule{WhileSyncTerm},
we have proved sound in \isabelle{} the rule \namerule{WhileDesugaredTerm},
which is similar to the rule \namerule{WhileDesugared} from \secref{sec:loops},
but requires proving that a loop variant stricly decreases after every iteration.}

\paragraph{Relation to total correctness}

\thibault{Terminating hyper-triples ensure a property \emph{strictly weaker} than total correctness.
The terminating hyper-triple
$\shoareTot{\top}{ \chavoc{x} \cseq \cwhile{x > 0}{ \cskip} }{ \top }$
is for example valid according to \defref{def:total-triples},
even though not all executions terminate.
This weaker property is sufficient to prove sound the two powerful inference rules \namerule{Frame} and \namerule{WhileSyncTerm}, which lift the restrictions of the rules \namerule{FrameSafe} and \namerule{WhileSync}.}

\thibault{However, we believe that the rules presented in \figref{fig:terminating-rules}
would also be sound for an alternative definition of hyper-triples that ensures \emph{total correctness}
(\ie where all executions terminate).
To formally characterize total correctness, we would need to use a small-step semantics, instead of the big-step semantics
we presented in \secref{subsec:language}.
We do not expect any significant challenge in using a small-step semantics instead of a big-step semantics.}

\subsection{(Dis-)Proving Termination}

It seems that \logic{} could also be extended to disprove termination.
To prove \emph{non}-termination of a loop $\cwhile{b}{C}$,
one can express and prove that a set of states $R$, in which all states satisfy the loop guard $b$, is a \emph{recurrent set} \cite{provingNonTermination}.
$R$ is a recurrent set iff executing $C$ in any state from $R$ leads to at least another state in $R$,
which can easily be expressed as a hyper-triple:
$$
\simpleHoare{\exists \inSet{\varphi} \ldotp \varphi \in R}{\cassume{b} \cseq C}{\exists \inSet{\varphi} \ldotp \varphi \in R}
$$

Thus, if one state from $R$ reaches $\cwhile{b}{C}$, we know that there is at least one non-terminating execution.

Note that both extensions of \logic{} (to prove and disprove termination) would require modifying the underlying semantic model of the logic; in particular, the extended semantics in  \defref{def:extended_semantics} should be modified to also capture non-terminating executions.
We do not expect such a modification to pose any significant challenge.

%% file: appendix/fibonacci.tex
In this section, we show the proof that the program $C_{\mathit{fib}}$ from \figref{fig:fib} is monotonic.
Precisely, we prove the triple
$$\shoare{ \forall \inSet{\varphi_1}{,}\inSet{\varphi_2} \ldotp \varphi_1(t) \seq 1 \sand \varphi_2(t) \seq 2 {\Rightarrow} \varphi_1(n) \sgeq \varphi_2(n) }{C_{\mathit{fib}}}{ \forall \inSet{\varphi_1}{,}\inSet{\varphi_2} \ldotp \varphi_1(t) \seq 1 \sand \varphi_2(t) \seq 2 {\Rightarrow} \varphi_1(a) \sgeq \varphi_2(a) }$$
using the rule \namerule{While-}$\forall^* \exists^*$ with the loop invariant
$I \triangleq (
    (\forall \inSet{\varphi_1}, \inSet{\varphi_2} \ldotp \varphi_1(t) \seq 1 \sand \varphi_2(t) \seq 2 \Rightarrow
    (\varphi_1(n) \sminus \varphi_1(i) \geq \varphi_2(n) \sminus \varphi_2(i) \land \varphi_1(a) \geq \varphi_2(a)
    \land \varphi_1(b) \geq \varphi_2(b))) \land \always( b \geq a \geq 0 ))$.

\begin{figure}[h]
    \footnotesize
    \begin{align*}
    &\outline{\forall \inSet{\varphi_1}{,}\inSet{\varphi_2} \ldotp \varphi_1(t) \seq 1 \sand \varphi_2(t) \seq 2 {\Rightarrow} \varphi_1(n) \sgeq \varphi_2(n) } \\
    &\outline{(\forall \inSet{\varphi_1}, \inSet{\varphi_2} \ldotp \varphi_1(t) \seq 1 \sand \varphi_2(t) \seq 2 \Rightarrow
    (\varphi_1(n) \sminus 0 \geq \varphi_2(n) \sminus 0 \land 0 \geq 0
    \land 1 \geq 1)) \land \always( 1 \geq a \geq 0 )} \tag{Cons} \\
    &\cassign{a}{0} \cseq \\
    &\outline{(\forall \inSet{\varphi_1}, \inSet{\varphi_2} \ldotp \varphi_1(t) \seq 1 \sand \varphi_2(t) \seq 2 \Rightarrow
    (\varphi_1(n) \sminus 0 \geq \varphi_2(n) \sminus 0 \land \varphi_1(a) \geq \varphi_2(a)
    \land 1 \geq 1)) \land \always( 1 \geq a \geq 0 )} \tag{AssignS} \\
    &\cassign{b}{1} \cseq \\
    &\outline{(\forall \inSet{\varphi_1}, \inSet{\varphi_2} \ldotp \varphi_1(t) \seq 1 \sand \varphi_2(t) \seq 2 \Rightarrow
    (\varphi_1(n) \sminus 0 \geq \varphi_2(n) \sminus 0 \land \varphi_1(a) \geq \varphi_2(a)
    \land \varphi_1(b) \geq \varphi_2(b))) \land \always( b \geq a \geq 0 )} \tag{AssignS} \\
    &\cassign{i}{0} \cseq \\
    &\outline{(\forall \inSet{\varphi_1}, \inSet{\varphi_2} \ldotp \varphi_1(t) \seq 1 \sand \varphi_2(t) \seq 2 \Rightarrow
    (\varphi_1(n) \sminus \varphi_1(i) \geq \varphi_2(n) \sminus \varphi_2(i) \land \varphi_1(a) \geq \varphi_2(a)
    \land \varphi_1(b) \geq \varphi_2(b))) \land \always( b \geq a \geq 0 )} \tag{AssignS}
    \end{align*}
    \caption{First part of the proof, which proves that the loop invariant $I$ holds before the loop.}
    \label{fig:fib-part1}
\end{figure}

\begin{figure}[h]
    \tiny
   \begin{align*}
    &\outline{\forall \inSet{\varphi_1}, \inSet{\varphi_2} \ldotp \varphi_1(t) \seq 1 \sand \varphi_2(t) \seq 2 \Rightarrow
    (\varphi_1(n) \sminus \varphi_1(i) \geq \varphi_2(n) \sminus \varphi_2(i) \land \varphi_1(a) \geq \varphi_2(a) %
    \land \varphi_1(b) \geq \varphi_2(b) %
    )
    \land \always( b \geq a \geq 0 )
    } \\
    &\outline{\forall \inSet{\varphi_1}, \inSet{\varphi_2} \ldotp \varphi_1(t) \seq 1 \sand \varphi_2(t) \seq 2 \Rightarrow
    (\varphi_1(n) \sminus \varphi_1(i) \geq \varphi_2(n) \sminus \varphi_2(i) \land \varphi_1(a) \geq \varphi_2(a) %
    \land \varphi_1(b) \geq \varphi_2(b) %
    )
    \land \always( b \geq a \geq 0 )
    \land \always( v_a = a \land v_b = b \land v_i = i )
    }
    \tag{LUpdateS}
     \\
    &\outlineStart{ (\forall \inSet{\varphi} \ldotp
    \varphi(i) = \varphi(v_i)
    \land \varphi(a) = \varphi(v_a)
    \land \varphi(b) = \varphi(v_b)
    ) } \\
    &\blue{
    \land
        \underbrace{(\forall \inSet{\varphi_1}, \inSet{\varphi_2} \ldotp \varphi_1(t) = 1 \land \varphi_2(t) = 2 \Rightarrow
            \varphi_1(v_a) \geq \varphi_2(v_a) \geq 0
            \land
            \varphi_1(v_b) \geq \varphi_2(v_b) \geq 0
            \land
            \varphi_1(n) - \varphi_1(v_i) \geq \varphi_2(n) - \varphi_2(v_i))}_{F}
    \} } \tag{Cons} \\
    &\outline{ \forall \inSet{\varphi} \ldotp
    \varphi(i) = \varphi(v_i) \land \varphi(a) = \varphi(v_a) \land \varphi(b) = \varphi(v_b)} \\
    &\quad \mathbf{if} \; (*) \; \{ \\
        &\quad \quad \outline{ \forall \inSet{\varphi} \ldotp
        \varphi(i) = \varphi(v_i) \land \varphi(a) = \varphi(v_a) \land \varphi(b) = \varphi(v_b)} \\
        &\quad \quad \outline{
            \forall \inSet{\varphi} \ldotp \varphi(i) < \varphi(n) \Rightarrow \varphi(v_i) < \varphi(n) \land \varphi(i) + 1 = \varphi(v_i) + 1
            \land \varphi(b) = \varphi(v_b) \land \varphi(a) + \varphi(b) = \varphi(v_a) + \varphi(v_b)
        } \tag{Cons} \\
        &\quad \quad \cassume{i < n} \cseq \\
        &\quad \quad \outline{
            \forall \inSet{\varphi} \ldotp \varphi(v_i) < \varphi(n) \land \varphi(i) + 1 = \varphi(v_i) + 1
            \land \varphi(b) = \varphi(v_b) \land \varphi(a) + \varphi(b) = \varphi(v_a) + \varphi(v_b)
        } \tag{AssumeS} \\
        &\quad \quad \cassign{tmp}{b} \cseq \\
        &\quad \quad \cassign{b}{a+b} \cseq \\
        &\quad \quad \cassign{a}{tmp} \cseq \\
        &\quad \quad \cassign{i}{i + 1} \\
        &\quad \quad \blue{ \{
            \underbrace{\forall \inSet{\varphi} \ldotp \varphi(v_i) < \varphi(n) \land \varphi(i) = \varphi(v_i) + 1
            \land \varphi(a) = \varphi(v_b) \land \varphi(b) = \varphi(v_a) + \varphi(v_b)}_{Q_1}
        \} } \tag{AssignS} \\
    &\quad \} \\
    &\quad \mathbf{else} \; \{ \\
        &\quad \quad \outline{ \forall \inSet{\varphi} \ldotp
        \varphi(i) = \varphi(v_i) \land \varphi(a) = \varphi(v_a) \land \varphi(b) = \varphi(v_b)} \\
        &\quad \quad \outline{
            \forall \inSet{\varphi} \ldotp \varphi(i) \geq \varphi(n) \Rightarrow \varphi(v_i) \geq \varphi(n) \land \varphi(i) = \varphi(v_i)
            \land \varphi(a) = \varphi(v_a) \land \varphi(b) = \varphi(v_b)
        } \tag{Cons} \\
        &\quad \quad \cassume{ \lnot (i < n)} \\
        &\quad \quad \blue{ \{
            \underbrace{\forall \inSet{\varphi} \ldotp \varphi(v_i) \geq \varphi(n) \land \varphi(i) = \varphi(v_i)
            \land \varphi(a) = \varphi(v_a) \land \varphi(b) = \varphi(v_b)}_{Q_2}
        \} } \tag{AssumeS} \\
    &\quad \} \\
    &\quad \outline{Q_1 \otimes Q_2}
    \tag{Choice} \\
    &\outline{(Q_1 \otimes Q_2) \land F} \tag{FrameSafe} \\
    &\outline{\forall \inSet{\varphi_1}, \inSet{\varphi_2} \ldotp \varphi_1(t) \seq 1 \sand \varphi_2(t) \seq 2 \Rightarrow
    (\varphi_1(n) \sminus \varphi_1(i) \geq \varphi_2(n) \sminus \varphi_2(i) \land \varphi_1(a) \geq \varphi_2(a) %
    \land \varphi_1(b) \geq \varphi_2(b) %
    )
    \land \always( b \geq a \geq 0 )
    } \tag{Cons} \\
    \end{align*}
    \caption{Second part of the proof. This proof outline shows
    $\shoare{I}{\cifonly{i < n}{C_{body}}}{I}$, the first premise of the rule \namerule{While-}$\forall^* \exists$,
    where $C_{body}$ refers to the body of the loop.
    }
    \label{fig:fib-part2}
\end{figure}

\figref{fig:fib-part1} shows the (trivial) first part of the proof, which proves that the loop invariant $I$ holds before the loop,
and \figref{fig:fib-part2} shows the proof of $\shoare{I}{\cifonly{i < n}{C_{body}}}{I}$, the first premise of
the rule \namerule{While-}$\forall^* \exists$ (the second premise is trivial).
In \figref{fig:fib-part2}, we first record
the initial values of $a$, $b$, and $i$ in the logical variables $v_a$, $v_b$, and $v_i$, respectively,
using the rule \namerule{LUpdateS} presented in \appref{app:compositionality}.
We then split our new hyper-assertion into
a simple part, $\forall \inSet{\varphi} \ldotp
    \varphi(i) = \varphi(v_i)
    \land \varphi(a) = \varphi(v_a)
    \land \varphi(b) = \varphi(v_b)$,
and a frame $F$ which stores the relevant information from the invariant $I$ with the initial values.
The hyper-assertion $F$ is then framed around the if-statement, using the rule \namerule{FrameSafe} from \appref{app:compositionality}.
The proof of each branch is straightforward; the postconditions of the two branches are combined via the rule \namerule{Choice}.

We finally conclude with the consequence rule.
This last entailment is justified by a case distinction.
Let $\varphi_1$, $\varphi_2$ be two states such that $\varphi_1(t) = 1$, $\varphi_2(t) = 2$,
and $\inSet{\varphi_1}$ and $\inSet{\varphi_2}$ hold.
From the frame $F$, we know that
$\varphi_1(v_a) \geq \varphi_2(v_a)$, and $\varphi_1(v_b) \geq \varphi_2(v_b)$.
We conclude the proof by distinguishing the following three cases (the proof for each case is straightforward):
(1) Both $\varphi_1$ and $\varphi_2$ took the first branch of the if statement,
\ie $\varphi_1(v_i) < \varphi_1(n)$ and $\varphi_2(v_i) < \varphi_2(n)$,
and thus both are in the set characterized by $Q_1$.
(2) Both $\varphi_1$ and $\varphi_2$ took the second branch,
\ie $\varphi_1(v_i) \geq \varphi_1(n)$ and $\varphi_2(v_i) \geq \varphi_2(n)$,
and thus both are in the set characterized by $Q_2$.
(3) $\varphi_1$ took the first branch and $\varphi_2$ took the second branch,
\ie $\varphi_1(v_i) < \varphi_1(n)$ and $\varphi_2(v_i) \geq \varphi_2(n)$,
and thus $\varphi_1$ is in the set characterized by $Q_1$ and $\varphi_2$ is in the set characterized by $Q_2$.

Importantly, the fourth case is not possible, because
this would imply
$\varphi_2(n) - \varphi_2(v_i) > 0 \geq \varphi_1(n) - \varphi_1(v_i)$,
which contradicts the inequality $\varphi_1(n) - \varphi_1(v_i) \geq \varphi_2(n) - \varphi_2(v_i)$
from the frame $F$.

%% file: appendix/minimum.tex
This section contains the proof, using the rule \namerule{While-}$exists$, that the program $C_m$ from \figref{fig:minimum} satisfies the triple
$$\simpleHoare{\lnot \emp{} \land \always(k \geq 0)}{C_m}{ \exists \inSet{\varphi} \ldotp \forall \inSet{\alpha} \ldotp \varphi(x) \leq \alpha(x) \land \varphi(y) \leq \alpha(y) }$$

\figref{fig:min-part1}
contains the (trivial) first part of the proof, which justifies that the hyper-assertion
$\exists \inSet{\varphi} \ldotp P_\varphi$,
where $P_\varphi \triangleq (\forall \inSet{\alpha} \ldotp
0 \leq \varphi(x) \leq \alpha(x) \land 0 \leq \varphi(y) \leq \alpha(y) \land \varphi(k) \leq \alpha(k) \land \varphi(i) = \alpha(i))$,
holds before the loop,
as required by the precondition of the conclusion of the rule \namerule{While-}$\exists$.

\figref{fig:min-part2} shows the proof of the first premise of the rule \namerule{While-}$\exists$, namely
$$\forall v \ldotp \exists \inSet{\varphi} \ldotp \shoare{P_\varphi \land \varphi(i) < \varphi(k) \land v = \varphi(k) - \varphi(i)}{
    \cifonly{i < k}{C_{body}}}{ \exists \inSet{\varphi} \ldotp P_\varphi \land \varphi(k) - \varphi(i) \prec v}$$
where $C_{body}$ is the body of the loop.

Finally, \figref{fig:min-part3} shows the proof of the second premise of the rule \namerule{While-}$\exists$.
More precisely, it shows
$$\forall \varphi \ldotp \shoare{Q_\varphi}{ \cifonly{i < k}{C_{body}} }{ Q_\varphi }$$
where $Q_\varphi \triangleq \forall \inSet{\alpha} \ldotp 0 \leq \varphi(x) \leq \alpha(x) \land 0 \leq \varphi(y) \leq \alpha(y)$,
from which we easily derive the second premise of the rule \namerule{While-}$\exists$, using the consequence rule
(since $P_\varphi$ clearly entails $Q_\varphi$), and the rule \namerule{While-}$\forall^* \exists^*$ rule.

\begin{figure}[h]
    \small
\begin{align*}
&\outline{\lnot \emp{} \land \always(k \geq 0)} \\
&\outline{\exists \inSet{\varphi} \ldotp \forall \inSet{\alpha} \ldotp
0 \leq 0 \leq 0 \land 0 \leq 0 \leq 0 \land \varphi(k) \leq \alpha(k) \land 0 = 0}
\tag{Cons} \\
&\cassign{x}{0} \cseq \\
&\outline{\exists \inSet{\varphi} \ldotp \forall \inSet{\alpha} \ldotp
0 \leq \varphi(x) \leq \alpha(x) \land 0 \leq 0 \leq 0 \land \varphi(k) \leq \alpha(k) \land 0 = 0}
\tag{AssignS} \\
&\cassign{y}{0} \cseq \\
&\outline{\exists \inSet{\varphi} \ldotp \forall \inSet{\alpha} \ldotp
0 \leq \varphi(x) \leq \alpha(x) \land 0 \leq \varphi(y) \leq \alpha(y) \land \varphi(k)
\leq \alpha(k) \land 0 = 0} 
\tag{AssignS} \\
&\cassign{i}{0} \cseq \\
&\outline{\exists \inSet{\varphi} \ldotp \forall \inSet{\alpha} \ldotp
0 \leq \varphi(x) \leq \alpha(x) \land 0 \leq \varphi(y) \leq \alpha(y) \land \varphi(k) \leq \alpha(k) \land \varphi(i) = \alpha(i)}
\tag{AssignS}
\end{align*}
\caption{First part of the proof: Proving the first loop invariant $\exists \inSet{\varphi} \ldotp P_\varphi$.}
\label{fig:min-part1}
\end{figure}

\begin{figure}[h]
    \scriptsize
\begin{align*}
&\outline{\exists \inSet{\varphi} \ldotp 
(\forall \inSet{\alpha} \ldotp 0 \leq \varphi(x) \leq \alpha(x) \land 0 \leq \varphi(y) \leq \alpha(y) \land \varphi(k) \leq \alpha(k) \land \varphi(i) = \alpha(i))
\land \varphi(i) < \varphi(k) \land v = \varphi(k) {-} \varphi(i)} \\
&\mathbf{if} \; (i < k) \; \{ \\
&\quad \outline{\exists \inSet{\varphi} \ldotp 
(\forall \inSet{\alpha} \ldotp 0 \leq \varphi(x) \leq \alpha(x) \land 0 \leq \varphi(y) \leq \alpha(y) \land \varphi(k) \leq \alpha(k) \land \varphi(i) = \alpha(i))
\land \varphi(i) < \varphi(k) \land v = \varphi(k) {-} \varphi(i) \land \always(i < k)} \\
&\quad \outlineStart{\exists \inSet{\varphi} \ldotp \exists u \ldotp u \geq 2 \land
(\forall \inSet{\alpha} \ldotp \forall v \ldotp v \geq 2 \Rightarrow  0 \leq 2 * \varphi(x) + u \leq 2 * \alpha(x) + v
\land 0 \leq \varphi(y) + \varphi(x) * u \leq \alpha(y) + \alpha(x) * v} \\
&\quad \outlineEnd{
 \land \varphi(k) \leq \alpha(k) \land \varphi(i) + 1 = \alpha(i) + 1)
\land \varphi(k) {-} (\varphi(i) {+} 1) \prec v}
\tag{Cons (1)} \\
&\quad \chavoc{r} \cseq \\
&\quad \cassume{r \geq 2} \cseq \\
&\quad \outlineStart{\exists \inSet{\varphi} \ldotp 
(\forall \inSet{\alpha} \ldotp 0 \leq 2 * \varphi(x) +\varphi(r) \leq 2 * \alpha(x) + \alpha(r)
\land 0 \leq \varphi(y) + \varphi(x) * \varphi(r) \leq \alpha(y) + \alpha(x) * \alpha(r) \land \varphi(k) \leq \alpha(k) \land \varphi(i) + 1 = \alpha(i) + 1)}
\\
&\quad \outlineEnd{
\land \varphi(k) {-} (\varphi(i) {+} 1) \prec v}
\tag{HavocS, AssumeS} \\
&\quad \cassign{t}{x} \cseq \\
&\quad \cassign{x}{2 * x + r} \cseq \\
&\quad \cassign{y}{y + t * r} \cseq \\
&\quad \cassign{i}{i + 1} \\
&\quad \outline{\exists \inSet{\varphi} \ldotp 
(\forall \inSet{\alpha} \ldotp 0 \leq \varphi(x) \leq \alpha(x) \land 0 \leq \varphi(y) \leq \alpha(y) \land \varphi(k) \leq \alpha(k) \land \varphi(i) = \alpha(i))
\land \varphi(k) {-} \varphi(i) \prec v}
\tag{AssignS} \\
&\} \\
&\mathbf{else} \; \{ \\
&\quad \outline{\exists \inSet{\varphi} \ldotp 
(\forall \inSet{\alpha} \ldotp 0 \leq \varphi(x) \leq \alpha(x) \land 0 \leq \varphi(y) \leq \alpha(y) \land \varphi(k) \leq \alpha(k) \land \varphi(i) = \alpha(i))
\land \varphi(i) < \varphi(k) \land v = \varphi(k) {-} \varphi(i) \land \always(i \geq k)} \\
&\quad \outline{\exists \inSet{\varphi} \ldotp 
(\forall \inSet{\alpha} \ldotp 0 \leq \varphi(x) \leq \alpha(x) \land 0 \leq \varphi(y) \leq \alpha(y) \land \varphi(k) \leq \alpha(k) \land \varphi(i) = \alpha(i))
\land \varphi(k) {-} \varphi(i) \prec v}
\tag{Cons (2)} \\
&\quad \cskip \\
&\quad \outline{\exists \inSet{\varphi} \ldotp 
(\forall \inSet{\alpha} \ldotp 0 \leq \varphi(x) \leq \alpha(x) \land 0 \leq \varphi(y) \leq \alpha(y) \land \varphi(k) \leq \alpha(k) \land \varphi(i) = \alpha(i))
\land \varphi(k) {-} \varphi(i) \prec v}
\tag{Skip} \\
&\} \\
&\outline{\exists \inSet{\varphi} \ldotp 
(\forall \inSet{\alpha} \ldotp 0 \leq \varphi(x) \leq \alpha(x) \land 0 \leq \varphi(y) \leq \alpha(y) \land \varphi(k) \leq \alpha(k) \land \varphi(i) = \alpha(i))
\land \varphi(k) {-} \varphi(i) \prec v}
\tag{IfSync}
 \\
\end{align*}
\caption{Second part of the proof:
Proving the first premise of the rule \namerule{While-}$\exists$,\\
$\forall v \ldotp \exists \inSet{\varphi} \ldotp \shoare{P_\varphi \land \varphi(i) < \varphi(k) \land v = \varphi(k) - \varphi(i)}{
    \cifonly{i < k}{C_{body}}}{ \exists \inSet{\varphi} \ldotp P_\varphi \land \varphi(k) - \varphi(i) \prec v}$.\\
For Cons (1), we simply choose $u = 2$.
For Cons (2), we notice that $\varphi(i) < \varphi(k)$ and $\always(i \geq k)$ are inconsistent (this branch is not taken at this stage),
and thus the entailment trivially holds.
}
\label{fig:min-part2}
\end{figure}

\begin{figure}
    \footnotesize
\begin{align*}
&\outline{\forall \inSet{\alpha} \ldotp 0 \leq \varphi(x) \leq \alpha(x) \land 0 \leq \varphi(y) \leq \alpha(y)} \\
&\mathbf{if} \; (*) \; \{ \\
&\quad \outline{\forall \inSet{\alpha} \ldotp 0 \leq \varphi(x) \leq \alpha(x) \land 0 \leq \varphi(y) \leq \alpha(y)} \\
&\quad \outline{\forall \inSet{\alpha} \ldotp \alpha(i) < \alpha(k) \Rightarrow \forall v \ldotp v \geq 2 \Rightarrow 0 \leq \varphi(x) \leq 2 * \alpha(x) + v \land 0 \leq \varphi(y) \leq \alpha(y) + \alpha(x) * v} \tag{Cons} \\
&\quad \cassume{i < k} \cseq \\
&\quad \outline{\forall \inSet{\alpha} \ldotp \forall v \ldotp v \geq 2 \Rightarrow 0 \leq \varphi(x) \leq 2 * \alpha(x) + v \land 0 \leq \varphi(y) \leq \alpha(y) + \alpha(x) * v}
\tag{AssumeS} \\
&\quad \chavoc{r} \cseq \\
&\quad \outline{\forall \inSet{\alpha} \ldotp \alpha(r) \geq 2 \Rightarrow 0 \leq \varphi(x) \leq 2 * \alpha(x) + \alpha(r) \land 0 \leq \varphi(y) \leq \alpha(y) + \alpha(x) * \alpha(r)} \tag{HavocS} \\
&\quad \cassume{r \geq 2} \cseq \\
&\quad \outline{\forall \inSet{\alpha} \ldotp 0 \leq \varphi(x) \leq 2 * \alpha(x) + \alpha(r) \land 0 \leq \varphi(y) \leq \alpha(y) + \alpha(x) * \alpha(r)} 
\tag{AssumeS} \\
&\quad \cassign{t}{x} \cseq \\
&\quad \cassign{x}{2*x + r} \cseq \\
&\quad \cassign{y}{y + t * r} \cseq \\
&\quad \cassign{i}{i + 1} \\
&\quad \outline{\forall \inSet{\alpha} \ldotp 0 \leq \varphi(x) \leq \alpha(x) \land 0 \leq \varphi(y) \leq \alpha(y)} \tag{AssignS} \\
&\} \\
&\mathbf{else} \; \{ \\
&\quad \outline{\forall \inSet{\alpha} \ldotp 0 \leq \varphi(x) \leq \alpha(x) \land 0 \leq \varphi(y) \leq \alpha(y)} \\
&\quad \outline{\forall \inSet{\alpha} \ldotp \alpha(i) \geq \alpha(k) \Rightarrow 0 \leq \varphi(x) \leq \alpha(x) \land 0 \leq \varphi(y) \leq \alpha(y)} \tag{Cons} \\
&\quad \cassume{i \geq k} \cseq \\
&\quad \cskip \\
&\quad \outline{\forall \inSet{\alpha} \ldotp 0 \leq \varphi(x) \leq \alpha(x) \land 0 \leq \varphi(y) \leq \alpha(y)} \tag{AssumeS, Skip} \\
&\} \\
&\outline{(\forall \inSet{\alpha} \ldotp 0 \leq \varphi(x) \leq \alpha(x) \land 0 \leq \varphi(y) \leq \alpha(y))
\otimes
(\forall \inSet{\alpha} \ldotp 0 \leq \varphi(x) \leq \alpha(x) \land 0 \leq \varphi(y) \leq \alpha(y))
} \tag{Choice} \\
&\outline{\forall \inSet{\alpha} \ldotp 0 \leq \varphi(x) \leq \alpha(x) \land 0 \leq \varphi(y) \leq \alpha(y)} \tag{Cons}
\end{align*}
\caption{Third part of the proof.
This proof outline shows
$\forall \varphi \ldotp \shoare{Q_\varphi}{ \cifonly{i < k}{C_{body}} }{ Q_\varphi }$.
}
\label{fig:min-part3}
\end{figure}

%% file: appendix/synchronous.tex
One thesis of this paper is that reasoning about how sets of states are affected by \emph{one} program command 
is powerful enough to reason about any program hyperproperty, which is supported by our completeness result (\thmref{thm:complete}).

However,
reasoning about (for example) two executions
of the same program sometimes boils down to reasoning about two executions of two \emph{different} but similar programs,
because of branching.
One \emph{a priori} appeal of relational program logics over \logic{} is thus the ability
to reason about two different branches \emph{synchronously}.

As an example, imagine that we want to reason about
$C' \triangleq \cnif{(\cassign{x}{x * 2} \cseq C)}{C}$.
Except for the assignment that happens only in one branch, the two branches are extremely similar.
In a relational program logic, we can exploit this similarity
by first reasoning about the assignment on its own,
and then reasoning about the two remaining branches $C$ and $C$ synchronously, since they are the same.

On the other hand, with the rule \namerule{If} from \figref{fig:rules},
we would have to reason about the two branches $\cassign{x}{x * 2} \cseq C$ and $C$
separately, even though they are closely related.

This is not a fundamental limitation of \logic{}.
We can indeed enable this kind of synchronous reasoning in \logic{},
by adding specialized rules, as illustrated by \propref{prop:synchronized-if} below.

Let us first define the following notation:
\begin{notation}
    \begin{align*}
    (A \otimes_{x=1,2} B)(S)
    \triangleq
    (
        A(\{ (l, \sigma) \mid (l, \sigma) \in S \land l(x) = 1 \})
        \land
        B(\{ (l, \sigma) \mid (l, \sigma) \in S \land l(x) = 2 \})
    )
    \end{align*}
\end{notation}

The assertion $A \otimes_{x=1,2} B$ holds in a set $S$
iff
the subset of all states in $S$ such that $l(x) = 1$ satisfies $A$,
and the subset of all states in $S$ such that $l(x) = 2$ must satisfy $B$.
We have proven in \isabelle{} the following synchronous rule:

\begin{proposition}\label{prop:synchronized-if}\textbf{Synchronous if rule.}
    If
    \begin{enumerate}
        \item $\hoare{P}{C_1}{P_1}$
        \item $\hoare{P}{C_2}{P_2}$
        \item $\hoare{P_1 \otimes_{x=1,2} P_2}{C}{R_1 \otimes_{x=1,2} R_2}$
        \item $\hoare{R_1}{C_1'}{Q_1}$
        \item $\hoare{R_2}{C_2'}{Q_2}$
        \item $x \notin \freevars{P_1} \cup \freevars{P_2} \cup \freevars{R_1} \cup \freevars{R_2}$
    \end{enumerate}
    Then
    $
    \hoare{P}{\cnif{(C_1 \cseq C \cseq C_1')}{(C_2 \cseq C \cseq C_2')}}{Q_1 \otimes Q_2}
    $.
\end{proposition}

This proposition shows how to reason synchronously
about the program command
$\cnif{(C_1 \cseq C \cseq C_1')}{(C_2 \cseq C \cseq C_2')}$.
Points 1) and 2) show that we can reason independently
about the different parts of the branches $C_1$ and $C_2$.
Point 3) then shows how we can reason synchronously about the execution of $C$ in both branches.
Finally, points 4) and 5) show how to go back to reasoning independently about each branch.

\clearpage